\documentclass[prodmode,acmec]{ec-acmsmall}
\usepackage[numbers]{natbib}
\pdfoutput=1

\usepackage{xpatch}

\renewcommand{\tab}{}
\def\dpfoot{Discussion Paper 663, Center for the Study of Rationality, Hebrew University of Jerusalem, Revised: February 11, 2015.}
\runningfoot{\dpfoot}
\renewcommand{\firstfoot}{\dpfoot}
\makeatletter
\patchcmd{\@oddhead}{\@acmArticle:}{}{\message{Patched}}{\message{Not Patched}}
\patchcmd{\@evenhead}{\@acmArticle:}{}{\message{Patched}}{\message{Not Patched}}
\patchcmd{\maketitle}{Additional Key Words and Phrases}{Key Words and Phrases}{\message{Patched}}{\message{Not Patched}}
\patchcmd{\acks}{ACKNOWLEDGMENTS}{ACKNOWLEDGEMENTS}{\message{Patched}}{\message{Not Patched}}
\def\endbottomstuff{
\vskip-13pt
\strut
\end@float
}
\makeatother

\usepackage{tikz}
\usepackage[
  bookmarks=true,
  bookmarksnumbered=true,
  bookmarksopen=true,
  pdfborder={0 0 0},
  breaklinks=true,
  colorlinks=true,
  linkcolor=black,
  citecolor=black,
  filecolor=black,
  urlcolor=black,
]{hyperref}
\usepackage{ifthen}
\usepackage{xifthen}
\let\myproof\proof
\let\myendproof\endproof
\let\myqed\qed
\let\proof\undefined
\let\endproof\undefined
\usepackage{amsthm}
\let\proof\myproof
\let\endproof\myendproof
\let\qed\myqed
\usepackage{amssymb,amsmath}
\usepackage{ctable}
\usepackage{enumitem}
\usepackage{algorithm}
\usepackage{algpseudocode}
\makeatletter
\renewcommand\floatc@ruled[2]{{\small{\@fs@cfont #1} #2\par}}
\makeatother
\usepackage{dsfont}
\usepackage{mathtools}
\usepackage{subfigure}
\usepackage[framemethod=tikz]{mdframed}
\usepackage{nicefrac}
\usepackage[capitalise]{cleveref}
\newcommand{\crefrangeconjunction}{ through\nobreakspace}

\algrenewcommand{\algorithmiccomment}[1]{\hfill {$\mathit{//}$} #1}

\usepackage{fnbreak}

\newtheoremstyle{acmthm}{\topsep}{\topsep}{\itshape}{12pt}{\sc}{\normalfont{.}}{5pt}{}
\newtheoremstyle{acmdefn}{\topsep}{\topsep}{\normalfont}{12pt}{\itshape}{\normalfont{.}}{5pt}{}
\theoremstyle{acmthm}
\makeatletter
\patchcmd{\thmhead}{(#3)}{\unskip\hskip5pt\relax$($#3$)$}{\message{Patched}}{\message{Not Patched}}
\patchcmd{\thmhead}{\the\thm@notefont}{}{\message{Patched}}{\message{Not Patched}}
\makeatother

\makeatletter
\let\c@theorem\undefined
\makeatother

\newtheorem{theorem}{Theorem}[section]

\newtheorem{proposition}[theorem]{Proposition}

\newtheorem{corollary}[theorem]{Corollary}
\crefname{corollary}{Corollary}{Corollaries}

\newtheorem{lemma}[theorem]{Lemma}

\theoremstyle{acmdefn}

\newtheorem{example}[theorem]{Example}

\newtheorem{definition}[theorem]{Definition}

\newtheorem{remark}[theorem]{Remark}

\setlist[1]{leftmargin=*,labelindent=0pt,labelsep=5pt}
\setlist[itemize,1]{leftmargin=7.5pt,itemindent=0pt,labelsep=2.5pt}

\newlist{parts}{enumerate}{1}
\crefname{partsi}{Part}{Parts}
\setlist[parts,1]{label=(\arabic*),ref=\arabic*}

\newlist{properties}{enumerate}{1}
\crefname{propertiesi}{Property}{Properties}
\setlist[properties,1]{label=(\arabic*),ref=\arabic*}

\newlist{conditions}{enumerate}{1}
\crefname{conditionsi}{Condition}{Conditions}
\setlist[conditions,1]{label=(\arabic*),ref=\arabic*}

\newcommand{\defnitemtitle}[1]{$($\emph{#1$)$}.}
\newcommand{\thmitemtitle}[1]{$(${\sc#1$)$}.}
\newcommand{\crefpart}[2]{\cref{#1}(\labelcref{#1-#2})}

\newcommand{\refintitle}[1]{\texorpdfstring{\ref{#1}}{\ref*{#1}}}

\newcommand{\eqdef}{\triangleq}
\newcommand{\Rge}{\mathbb{R}_{\ge}}
\newcommand{\timeset}{\mathcal{T}}
\newcommand{\load}[1]{\ell^s_{#1}}
\newcommand{\loadt}[1]{\ell^{s'}_{#1}}
\newcommand{\loadtt}[1]{\ell^{s''}_{#1}}
\newcommand{\expect}[1]{E\bigl[#1\bigr]}
\newcommand{\prob}[1]{P\bigl[#1\bigr]}
\newcommand{\prods}{\mathbb{P}_n}
\newcommand{\prodsm}{\mathbb{P}_{n-1}}
\newcommand{\prodsk}{\mathbb{P}_k}
\newcommand{\prodsg}{\mathbb{P}_{n_g}}
\newcommand{\prodso}{\mathbb{P}_{n_1}}
\newcommand{\prodst}{\mathbb{P}_{n_2}}
\DeclareMathOperator{\supp}{supp}
\DeclareMathOperator*{\Max}{Max}
\newcommand{\coarsegame}{\mbox{$(n,\mu,\succeq_C)$}}
\newcommand{\finegame}{\mbox{$(n,\mu,\succeq_F)$}}
\newcommand{\noconsumption}{\lnot}
\newcommand{\tth}{\textsuperscript{th}}

\newcommand{\spaceship}[5]{\ifthenelse{#1 < #2}{#3}{\ifthenelse{#1 = #2}{#4}{#5}}}

\xdefinecolor{green}{rgb}{0,0.7,0}
\xdefinecolor{purple}{rgb}{0.55,0,0.8}
\newcommand{\vessels}[1]{%
\begin{tikzpicture}[yscale=0.489,xscale=0.943,font=\small]
\pgfmathsetmacro{\containerwidth}{5.8}
\pgfmathsetmacro{\grayamount}{1/32}
\pgfmathsetmacro{\purpleamount}{0.25}
\pgfmathsetmacro{\yellowamount}{3/32}
\pgfmathsetmacro{\blueamount}{5/32}
\pgfmathsetmacro{\redamount}{1/16}
\pgfmathsetmacro{\greenamount}{13/32}
\pgfmathsetmacro{\littlebitleft}{1/16}
\spaceship{#1}{2}{\pgfmathsetmacro{\purpleincontainer}{1}}{\pgfmathsetmacro{\purpleincontainer}{\littlebitleft}}{\pgfmathsetmacro{\purpleincontainer}{0}}
\spaceship{#1}{3}{\pgfmathsetmacro{\yellowincontainer}{1}}{\pgfmathsetmacro{\yellowincontainer}{\littlebitleft}}{\pgfmathsetmacro{\yellowincontainer}{0}}
\spaceship{#1}{4}{\pgfmathsetmacro{\blueincontainer}{1}}{\pgfmathsetmacro{\blueincontainer}{\littlebitleft}}{\pgfmathsetmacro{\blueincontainer}{0}}
\spaceship{#1}{5}{\pgfmathsetmacro{\redincontainer}{1}}{\pgfmathsetmacro{\redincontainer}{\littlebitleft}}{\pgfmathsetmacro{\redincontainer}{0}}
\spaceship{#1}{6}{\pgfmathsetmacro{\greenincontainer}{1}}{\pgfmathsetmacro{\greenincontainer}{\littlebitleft}}{\pgfmathsetmacro{\greenincontainer}{0}}
\draw[fill=gray,draw=none,shift={(0,6)}] (0,0) -- (0,1) -- (\grayamount*\containerwidth,1) -- (\grayamount*\containerwidth,0) -- cycle;
\draw[fill=purple,draw=none,shift={(\grayamount*\containerwidth,6)}] (0,0) -- (0,\purpleincontainer) -- (\purpleamount*\containerwidth,\purpleincontainer) -- (\purpleamount*\containerwidth,0) -- cycle;
\draw[fill=yellow,draw=none,shift={({(\grayamount+\purpleamount)*\containerwidth},6)}] (0,0) -- (0,\yellowincontainer) -- (\yellowamount*\containerwidth,\yellowincontainer) -- (\yellowamount*\containerwidth,0) -- cycle;
\draw[fill=blue,draw=none,shift={({(\grayamount+\purpleamount+\yellowamount)*\containerwidth},6)}] (0,0) -- (0,\blueincontainer) -- (\blueamount*\containerwidth,\blueincontainer) -- (\blueamount*\containerwidth,0) -- cycle;
\draw[fill=red,draw=none,shift={({(\grayamount+\purpleamount+\yellowamount+\blueamount)*\containerwidth},6)}] (0,0) -- (0,\redincontainer) -- (\redamount*\containerwidth,\redincontainer) -- (\redamount*\containerwidth,0) -- cycle;
\draw[fill=green,draw=none,shift={({(\grayamount+\purpleamount+\yellowamount+\blueamount+\redamount)*\containerwidth},6)}] (0,0) -- (0,\greenincontainer) -- (\greenamount*\containerwidth,\greenincontainer) -- (\greenamount*\containerwidth,0) -- cycle;
\draw (0,7) -- (0,6);
\draw[shift={((\grayamount*\containerwidth,0)}] (0,7) -- (0,6) node[below] {$t_0$};
\draw[shift={({(\grayamount+\purpleamount)*\containerwidth},0)}] (0,7) -- (0,6) node[below] {$t_1$};
\draw[shift={({(\grayamount+\purpleamount+\yellowamount)*\containerwidth},0)}] (0,7) -- (0,6) node[below] {$t_2$};
\draw[shift={({(\grayamount+\purpleamount+\yellowamount+\blueamount)*\containerwidth},0)}] (0,7) -- (0,6) node[below] {$t_3$};
\draw[shift={({(\grayamount+\purpleamount+\yellowamount+\blueamount+\redamount)*\containerwidth},0)}] (0,7) -- (0,6) node[below] {$t_4$};
\draw (\containerwidth,7) -- (\containerwidth,6);
\draw (0,7) -- (\containerwidth,7);
\draw (0,6) -- (\containerwidth,6);
\foreach \x/\xtext in {0/0, 0.25/0.25, 0.5/0.5, 0.75/0.75, 1/1}
{
\draw[shift={(\containerwidth*\x,7)}] (0pt,-2pt) -- (0pt,0pt) node[above] {$\xtext$};
}
\ifthenelse{#1 < 2}{\pgfmathsetmacro{\purpleinfirst}{0}}{\pgfmathsetmacro{\purpleinfirst}{16*\purpleamount}}
\ifthenelse{#1 < 3}{\pgfmathsetmacro{\yellowinsecond}{0}}{\pgfmathsetmacro{\yellowinsecond}{16*\yellowamount}}
\ifthenelse{#1 < 4}{\pgfmathsetmacro{\blueinsecond}{0}\pgfmathsetmacro{\blueinthird}{0}}{\pgfmathsetmacro{\blueinsecond}{16*(\blueamount-\yellowamount)/2}\pgfmathsetmacro{\blueinthird}{\blueinsecond+16*\yellowamount}}
\ifthenelse{#1 < 5}{\pgfmathsetmacro{\redinfourth}{0}}{\pgfmathsetmacro{\redinfourth}{16*\redamount}}
\ifthenelse{#1 < 6}{\pgfmathsetmacro{\greeninsecond}{0}\pgfmathsetmacro{\greeninthird}{0}\pgfmathsetmacro{\greeninfourth}{0}\pgfmathsetmacro{\greeninfifth}{0}}{\pgfmathsetmacro{\greeninsecond}{16*(\greenamount-\yellowamount-\blueamount+\redamount)/4}\pgfmathsetmacro{\greeninthird}{\greeninsecond}\pgfmathsetmacro{\greeninfourth}{\greeninthird+16*((\yellowamount+\blueamount)/2-\redamount)}\pgfmathsetmacro{\greeninfifth}{\greeninfourth+16*\redamount}}
\draw[fill=purple,draw=none] (0,0) -- (0,\purpleinfirst) -- (1,\purpleinfirst) -- (1,0) -- cycle;
\ifthenelse{#1 = 2}{
\draw[purple,thick] (.5,6) -- (.5,0);
}{}
\draw[fill=green,draw=none,shift={(1.2,0)}] (0,0) -- (0,\greeninsecond) -- (1,\greeninsecond) -- (1,0) -- cycle;
\draw[fill=blue,draw=none,shift={(1.2,\greeninsecond)}] (0,0) -- (0,\blueinsecond) -- (1,\blueinsecond) -- (1,0) -- cycle;
\draw[fill=yellow,draw=none,shift={(1.2,\greeninsecond+\blueinsecond)}] (0,0) -- (0,\yellowinsecond) -- (1,\yellowinsecond) -- (1,0) -- cycle;
\ifthenelse{#1 = 3}{
\draw[yellow,thick,shift={({(\grayamount+\purpleamount+\yellowamount/2)*\containerwidth},0)}] (0,6) -- (0,0);
}{}
\draw[fill=green,draw=none,shift={(2*1.2,0)}] (0,0) -- (0,\greeninthird) -- (1,\greeninthird) -- (1,0) -- cycle;
\draw[fill=blue,draw=none,shift={(2*1.2,\greeninthird)}] (0,0) -- (0,\blueinthird) -- (1,\blueinthird) -- (1,0) -- cycle;
\ifthenelse{#1 = 4}{
\draw[blue,thick,shift={(2.4,0)}] (.5,6) -- (.5,0);
}{}
\draw[fill=green,draw=none,shift={(3*1.2,0)}] (0,0) -- (0,\greeninfourth) -- (1,\greeninfourth) -- (1,0) -- cycle;
\draw[fill=red,draw=none,shift={(3*1.2,\greeninfourth)}] (0,0) -- (0,\redinfourth) -- (1,\redinfourth) -- (1,0) -- cycle;
\ifthenelse{#1 = 5}{
\draw[red,thick,shift={(3.6,0)}] ({-3.6+(\grayamount+\purpleamount+\yellowamount+\blueamount+\redamount/2)*\containerwidth},6) -- ({-3.6+(\grayamount+\purpleamount+\yellowamount+\blueamount+\redamount/2)*\containerwidth},5.24) -- (.5,5) -- (.5,0);
}{}
\draw[fill=green,draw=none,shift={(4*1.2,0)}] (0,0) -- (0,\greeninfifth) -- (1,\greeninfifth) -- (1,0) -- cycle;
\ifthenelse{#1 = 6}{
\draw[green,thick,shift={(4.8,0)}] (.5,6) -- (.5,0);
}{}
\foreach \x in {0, 1, 2, 3, 4}
{
\ifthenelse{\x > 0}{
\draw[shift={(\x*1.2,0)}] (0,0) -- (0,0.1);
\draw[shift={(\x*1.2,0)}] (0,0.2) -- (0,5);
\foreach \y/\ytext in {0/0, 0.1/0.1, 0.2/0.2, 0.3/0.3}
\draw[shift={(\x*1.2,16*\y)}] (2pt,0pt) -- (0pt,0pt);
}{
\draw[shift={(\x*1.2,0)}] (0,0) -- (0,5);
\foreach \y/\ytext in {0/0, 0.1/0.1, 0.2/0.2, 0.3/0.3}
\draw[shift={(\x*1.2,16*\y)}] (2pt,0pt) -- (0pt,0pt) node[left] {$\ytext$};
}
\draw[shift={(\x*1.2,0)}] (0,0) -- (1,0);
\ifthenelse{\x < 4}{
\draw[shift={(\x*1.2,0)}] (1,0) -- (1,0.1);
\draw[shift={(\x*1.2,0)}] (1,0.1) -- (1.2,0.1);
\draw[shift={(\x*1.2,0)}] (1.2,0.1) -- (1,0.15);
\draw[shift={(\x*1.2,0)}] (1,0.15) -- (1.2,0.2);
\draw[shift={(\x*1.2,0)}] (1.2,0.2) -- (1,0.2);
\draw[shift={(\x*1.2,0)}] (1,0.2) -- (1,5);
}{
\draw[shift={(\x*1.2,0)}] (1,0) -- (1,5) node[right,color=white] {$0.0$};
}
}
\end{tikzpicture}%
}

\hypersetup{
  pdfauthor      = {Yannai A. Gonczarowski <yannai@gonch.name> and Moshe Tennenholtz <moshet@ie.technion.ac.il>},
  pdftitle       = {A Mirage of Market Allocation},
}

\title{A Mirage of Market Allocation}
\author{YANNAI A. GONCZAROWSKI\affil{The Hebrew University of Jerusalem and Microsoft Research}
MOSHE TENNENHOLTZ\affil{Technion --- Israel Institute of Technology}
\vspace{1.55em}
\begin{mdframed}
\textbf{Market Allocation} --- a situation where competitors agree to not compete with each other in specific markets, by dividing up geographic areas, types of products, or types of customers.
\hfill {\footnotesize West's Encyclopedia of American Law, 2nd Edition, Volume 9, Page 168}\end{mdframed}\vspace{-1.2em}
}
\begin{abstract}
Can noncooperative behaviour of merchants lead to a market split that \emph{prima facie} seems anticompetitive?
We introduce a model in which service providers, with internet service providers (ISPs) being the main example,
aim at optimizing the number of customers who use their services, while customers aim at choosing service providers with low customer load (which translates to high effective bandwidth per subscriber, in the case of ISPs).
Each service provider chooses between a variety of levels of service (latencies, in the case of ISPs),
and as long as it does not lose customers, aims at minimizing its level of service; the minimum level of service required to satisfy a customer varies across customers. We consider a two-stage competition, in the first stage of which the service providers select their levels of service, and in the second stage --- customers choose between the service providers. (We show via a novel construction that for any choice of strategies for the service providers, a unique distribution of the customers' mass between them emerges from all Nash equilibria among the customers, showing the incentives of service providers in this two-stage game to be well defined.) In the two-stage game, we show that the competition among the service providers
possesses a unique Nash equilibrium, which is moreover super-strong; we also show that all sequential better-response dynamics of service providers reach this equilibrium, with best-response dynamics doing so surprisingly fast.
If service providers choose their levels of service according to this equilibrium, then the unique Nash equilibrium among customers in the second phase is essentially a split of the market between the service providers, based on the customers' minimum acceptable quality of service; moreover, each service provider's chosen level of service
is the lowest acceptable by the entirety of the slice of the market that chooses it, seemingly making no attempt to attract any other customers.
Our results show that this \emph{prima facie} market allocation (collusive split of the market) arises as the unique and highly robust outcome of noncooperative (i.e.\ free from any form of collusion), even myopic, service-provider behaviour. The results of this paper are applicable to a wide variety of scenarios,
from explaining phenomena observable in some food markets, to shedding a surprising light on aspects of location theory, such as the formation and structure of a city's central business district.
\end{abstract}

\begin{document}%
\begin{bottomstuff}%
Authors' addresses: Y.\ A.\ Gonczarowski, Einstein Institute of Mathematics, Rachel \& Selim Benin School of Computer Science \& Engineering, and Federmann Center for the Study of Rationality,
The Hebrew University of Jerusalem, Israel; and Microsoft Research, \mbox{\emph{Email}: \href{mailto:yannai@gonch.name}{yannai@gonch.name}};
M.\ Tennenholtz, William Davidson Faculty of Industrial Engineering and Management, Technion --- Israel Institute of Technology (work carried out while at Microsoft Research), \mbox{\emph{Email}: \href{mailto:moshet@ie.technion.ac.il}{moshet@ie.technion.ac.il}}.%
\end{bottomstuff}%
\maketitle

\section{Introduction}

\subsection{Setting}

\subsubsection{Shopping for an Internet Connection}
In today's world, an internet connection has become a necessity in many households. Of the many parameters characterizing an internet connection, two have emerged as most important to the home user:\footnote{Indeed, these are precisely the two parameters measured by the popular internet speed-testing website \texttt{\url{www.speedtest.net}}.} the ever-popular \emph{bandwidth}, and the \emph{latency}\footnote{Twice the latency is sometimes referred to as the \emph{ping time}.}. While the bandwidth measures the amount of data transmitted (equivalently, received) per second, the latency measures the time it takes a single packet of data to reach its destination.\footnote{We emphasize that high latency corresponds to bad quality of service.} In the metaphorical highway of the internet, the bandwidth may be thought of as corresponding to the number of lanes, while the latency corresponds the the length of the highway.
While some users may not be sensitive to latency (indeed, when streaming a 45-minute TV show from, say, Netflix or Hulu, most users would not mind waiting an extra second before the show begins; this is also the case when downloading content for future offline consumption), some other users may have very harsh latency limitations (indeed, when playing a multiplayer video game online, it is extremely important for each player that whenever she presses a button on her controller, the associated action happens as soon as possible; a delay longer that of the adversary by as little as one tenth of a second may be unacceptable).

We consider a stylized model, in which each \emph{customer} is interested in precisely one internet connection, and is willing to tolerate a latency of at most $d$ milliseconds (a customer-dependent real value). As long as this customer's latency demand is met, her sole consideration is that of maximizing her effective bandwidth (we think of subscription costs as low and similar, as is the case in real life). The effective bandwidth of each customer subscribed to a given \emph{internet service provider}, or \emph{ISP}, is the total bandwidth available to this ISP (a fixed known ISP-dependent value) divided by the number of customers subscribed to this ISP.\footnote{In \cref{heterogeneous}, we deal with a generalized model, which accommodates also for, e.g.\ some ISPs purchasing more total bandwidth as their subscriber pool grows. Our main results surveyed in the introduction continue to hold even under such generalizations.}
A \emph{Nash equilibrium} among the customers is therefore an assignment of ISPs to customers, s.t.\ for each customer with latency limit $d$, no ISP with latency no greater than $d$ has a subscriber pool smaller than that of the ISP assigned to this customer.

We consider a scenario with finitely many ISPs and continuously many customers, the distribution of $d$ among whom is given by an arbitrary finite measure. Preparing the ground for the main results of this paper, which follow below, in \cref{consumers} we use a novel construction to show the following result (similar in spirit to other results regarding congestion games and crowding games). 

\begin{theorem}[Informal version of \cref{consumer-symmetric-nash-exists,indifference-producers,indifference-consumers,compute-ell}]\label{intro-consumer-nash} Fix the characteristics (latency and total bandwidth) of $n$ ISPs.
\begin{parts}
\item
A Nash equilibrium among the customers exists. Furthermore, there exists such a Nash equilibrium for which the strategies can be computed efficiently.
\item
The effective bandwidth of each customer, as well as the number of subscribers to each ISP, are the same across all Nash equilibria.
\end{parts}
\end{theorem}

\subsubsection{An ISP Game} Obviously, each ISP would like to offer a latency that maximizes its number of subscribers. (By \cref{intro-consumer-nash}, the number of subscribers is well defined given the latencies of all ISPs, assuming a Nash equilibrium among the customers.) That being said, as low-latency infrastructure is costlier to erect, each ISP would like to offer the highest latency possible, as long as this does not reduce the size of its subscriber pool. As we think of the number of subscribers as indicative of monthly income, and of the investment in infrastructure as a one-time expense (with infrastructure upkeep cost being independent of latency), we have that
each ISP would like to offer a latency that first and foremost maximizes its number of subscribers, and only then (as a tie-breaking rule among latency values that yield the same
number of subscribers) is as high as possible. In \cref{producers-fine}, and more generally in \cref{heterogeneous}, we show the following --- the first of our main results regarding this two-stage competition.

\pagebreak
\begin{theorem}[Informal version of \cref{producer-fine-nash-char,heterogeneous-fine}]\label{intro-producer-nash}\leavevmode
\begin{parts}
\item
For every ordering $\pi$ of the $n$ ISPs, there exists a \emph{unique} Nash equilibrium among them s.t.\ their latency levels are ordered according to $\pi$.
\item
This Nash equilibrium is super-strong.\footnote{See \cref{producer-coarse-super-strong} in \cref{producers-coarse} and the preceding discussion for the definition of super-strong equilibrium.}
\item
Each ISP has the same number of subscribers in all Nash equilibria (regardless of the chosen ordering of ISPs $\pi$).
\end{parts}
\end{theorem}

We further demonstrate the robustness of the Nash equilibrium defined in \cref{intro-producer-nash} by considering dynamics among ISPs.
A \emph{sequential} best-response dynamic is a process starting with arbitrary latency levels, and in which at each turn an arbitrary ISP changes its latency level to one that,
\emph{ceteris paribus}, maximizes its preferences (we show that such a latency level always exists for every possible measure on customers latency limits);
we assume that each ISP is allowed to change its latency level infinitely often. A \emph{round} in a best-response dynamic is a sequence of consecutive steps in which each ISP is allowed to change its latency level at least once. Finally, a sequential \emph{$\delta$-better-response} dynamic is a sequential dynamic in which each change in latency need not necessarily maximize the ISP's preferences, as long as it increases the size of its subscriber pool by at least $\delta$ of the entire market size.\footnote{Since we consider continuously many customers, we demand a $\delta$-improvement in order to avoid improvements \emph{\`{a} la} Zeno's ``Race Course'' paradox.}
 In \cref{producers-fine}, we show the following main result.

\begin{theorem}[Informal version of \cref{fine-sequential-cor,fine-best-response-fast-cor}]\label{intro-producer-dynamics}\leavevmode
\begin{parts}
\item
For every $\delta>0$, every sequential $\delta$-better-response dynamic reaches a Nash equilibrium in finitely many steps, and remains constant from that point onward.
\item
Sequential best-response dynamics reach a Nash equilibrium in a small number of rounds.
\end{parts}
\end{theorem}
We also analyse dynamics in which several ISPs change their latency levels simultaneously. (See \cref{coarse-response-equilibrium,coarse-best-response-fast,fine-response-equilibrium}.) We emphasize that \cref{intro-producer-dynamics} does not stem from any ``hidden'' introduction of any exogenous costs on restructuring infrastructure;
i.e.\ this \lcnamecref{intro-producer-dynamics} holds in very general settings, even when the utility of an ISP from a given number of subscribers and a given latency does not decrease with the number of latency changes in previous steps of the studied dynamics.

\subsubsection{Prima Facie Market Allocation}

Our study culminates with the analysis of the structure of the unique equilibrium from \cref{intro-producer-nash} and of the underlying equilibrium among customers, which turns out to be unique as well

\begin{theorem}[Informal version of \cref{producer-fine-nash-char,producer-fine-market-allocation,heterogeneous-fine}]\label{intro-allocation}
Fix a Nash equilibrium among the ISPs. Denote the number of subscribers to the first ISP (the one with lowest latency) by $\ell_1$, the number of subscribers to the second ISP (the one with second-lowest latency) by $\ell_2$, and so forth.
\begin{parts}
\item
The $\ell_1$ customers with smallest latency limits all subscribe to the first ISP, whose latency is the highest that still accommodates all of these customers.
\item
The next $\ell_2$ customers all subscribe to the second ISP, whose latency is the highest that still accommodates all of these customers.
\item
etc.
\end{parts}
\end{theorem}

\pagebreak
\cref{intro-allocation} shows that \textbf{the market is split among the various ISPs based on the willingness of a customer to accept a high latency level, and each ISP chooses the highest latency level acceptable by the entirety of its slice of the market, seemingly making no attempt to attract any other customers}. \cref{intro-producer-nash,intro-producer-dynamics} show that this \emph{prima facie} market allocation (collusive split of the market) among the various ISPs arises as the unique possible outcome, not as a result of anticompetitive practices, but rather as
a result of noncooperative dynamics, each ISP only looking to myopically maximize its preferences at every step; no signalling (via e.g.\ choice of latency level) or any other collusive 
or cooperative ``trick'' whatsoever is used in order to reach and maintain this market split.

\subsection{Alternative Interpretations/Applications}

It is worthwhile to point out that our framework captures far more than merely the ISP-competition scenario introduced above,
by thinking of a latency level more generally as a quality of service (QoS) of sorts of an ISP: the lower the latency of an ISP, the better the quality of service that it provides.
In \cref{main-street,discussion}, we give two examples of other possible applications stemming from this insight.
In each of these examples, QoS is given different meanings, which, in turn, result in different meanings of market split based on acceptable QoS. These examples provide insights into the breadth of meanings that can be captured by the idea of QoS and consequently by our model, and into the meaning of market split based
on acceptable QoS. The example given in \cref{main-street} gives an application to location theory, and derives results for an extended model with multiple types of goods. The example given in \cref{discussion} offers real-world evidence supporting the applicability of our model to certain food markets.

For generality, we henceforth use the more generic term \emph{producers} to refer to e.g.\ ISPs, \emph{consumers} to refer to e.g.\ customers, and \emph{QoS} to refer to e.g.\ latency.

\subsection{Related Work}\label{related-work}

Our consumer games are a form of congestion games, and more specifically, of resource-selection games.  Congestion games with finitely many players have been introduced by \citeN{Rosenthal73}; in fact, the term has been coined in a paper by \citeN{MondererShapley96}, titled ``Potential Games'', where it is shown that a game has a potential iff it is a congestion game. While the discussion there refers to atomic games with finitely many players, work in computer science and game theory also deals with nonatomic games, in which there may be a continuum of players as in our model (see e.g.\ \cite{RoughgardenTardos} for work in CS that uses such games). While substantial parts of our introductory \cref{intro-consumer-nash} can also be deduced from results by by \citeN{Schmeidler73} and by \citeN{Beckmann56}, we emphasize that the novel constructive machinery that we introduce in order to prove it is simpler, allows for efficient calculation, provides for more general results, and provides for auxiliary results useful in the analysis of our producer game and in obtaining our main results. \citeN{HolzmanCong1,HolzmanRoute} look at restrictions on strategy sets of atomic congestion games; one way to view our consumer games is as a special form of restricted nonatomic congestion games defined for general measure functions on agents' types, capturing their possible strategy sets. As it turns out, this set of games possess
many desired game-theoretic properties.

The actual games that we study are in fact two-stage games, where the second stage is a congestion game as discussed above; the first stage can be viewed as a form of facility-location game among producers, with QoS playing the role of location (see \cref{main-street}), where the main aim of producers is to select a QoS to be selected by as many consumers as possible. This resembles the literature on location theory initiated by \citeN{Hotelling}, although the utility function of the producers in our setting is different, and allows for fine preferences based on distance from a location most preferred by consumers. Given the above, our model can be viewed as a novel combination of facility-location games among producers with congestion games among consumers.

Another type of related literature deals with scheduling and queuing with multiple machines, where the jobs choose among available services and the level of service they receive depends on the selections by other jobs. Recently, two-stage games in these contexts have been studied, consisting of a strategic selection by machines between queuing policies \cite{AshlagiLT13} or scheduling policies \cite{AshlagiTZ10}, followed by a strategic selection by jobs between the various selected policies. Our work introduces a novel type of a two-stage scenario, which
may be considered as somewhat related. More remotely is the literature on competing mechanisms in the context of auctions, which employs such two-phase setting, but in a very different context of mechanism design with money (see e.g.\ \cite{McAfee}). 

The proof of \cref{intro-consumer-nash} draws its intuition from an analogy to a hydraulic system of communicating vessels (see \cref{vessels}). \citeN{Kaminsky} (see also \cite{AumannThreeWives}) uses an analogy to quite a different system of communicating vessels to solve rationing problems; his motivation is quite different, and involves extending bilateral rationing rules. While \citeauthor{Kaminsky} uses a set of two-way communicating vessels, we use a set of one-way communicating vessels. In this context, the problem of finding a Nash equilibrium among consumers may be regarded as a rationing problem with certain ``reserves'' for producers with high quality of service. Our treatment, especially in light of the discussion in \cref{heterogeneous} (see in particular \cref{vessels-odd-shapes}), also sheds new light on rationing problems, as congestion games of sorts among a continuum of good-fragments.

\section{Notation}

\begin{definition}[Notation]
\leavevmode
\begin{itemize}
\item
\defnitemtitle{Naturals}
We denote the natural numbers by $\mathbb{N}\eqdef\{0,1,2,\ldots\}$.
\item
\defnitemtitle{Nonnegative Reals}
We denote the nonnegative reals by $\Rge\eqdef\{r \in \mathbb{R} \mid r \ge 0\}$.
\item
\defnitemtitle{Maximizing Arguments}
Given a set $S$ and a function $f:S\rightarrow\mathbb{R}$ that  attains a maximum value on $S$, we denote the \emph{set} of arguments in $S$ maximizing $f$ by $\arg\Max_{s\in S}f(s)\eqdef\{s \in S \mid f(s)=m\}$, where $m\eqdef\Max_{s\in S}f(s)$.
\item
\defnitemtitle{Simplex}
For a finite set $S$ and a nonempty subset $S' \subseteq S$, we define \[\Delta^{S'}=\Bigl\{s \in [0,1]^S \:\Big|\: \sum_{j \in S'}s_j=1 \And \forall j \in S\setminus S':s_j=0\Bigr\}.\]
(The set $S$ will be clear from context.)
\item
For every $n \in \mathbb{N}$, we define $\prods\eqdef\{0,1,\ldots,n-1\}$.
\item
Given a tuple $\bar{t}=(t_0,\ldots,t_{n-1})\in S^{\prods}$ for some set $S$ and some $n\in\mathbb{N}$, and given $j\in\prods$ and $t'\in S$, we define $(\bar{t}_{-j},t')\eqdef(t_0,\ldots,t_{j-2},t_{j-1},t',t_{j+1},t_{j+2},\ldots,t_{n-1})\in S^{\prods}$.
\item
For every $n \in \mathbb{N}$, we denote the set of permutations on $\prods$ by $\prods!$.
\end{itemize}
\end{definition}

\section{Prelude: The Consumer (Customer) Game}\label{consumers}

Preparing the ground for the main results of this paper, in this \lcnamecref{consumers} we define the congestion game among consumers, and use a novel construction\footnote{See \citeN{hydraulic-selection} for a significant, highly nontrivial, generalization of our treatment of only the consumer game (without the producer game) to arbitrary resource-selection games (in which the resources available to a player may be any subset of $\prods$ and not merely a ``QoS-prefix'' of $\prods$ to which the construction of this \lcnamecref{consumers} is inherently tailored) and beyond.} to prove the existence of
Nash equilibrium and the uniqueness of equilibrium loads, and to efficiently calculate these loads.\footnote{As mentioned in \cref{related-work}, while the existence of Nash equilibrium and the uniqueness of equilibrium loads can also be derived from theorems by \citeN{Schmeidler73} and by \citeN{Beckmann56}, respectively, the novel constructive machinery that we introduce here is simpler, provides for more general uniqueness results, provides for auxiliary results useful in analysing dynamics of the producer game in \cref{producers}, and allows for efficient calculation.} 
Full proofs and auxiliary results are provided in \cref{consumers-proofs,ell-analysis}.

In this \lcnamecref{consumers} and in \cref{producers}, for ease of presentation, we present a model in which each consumer would like to consume from a least-loaded producer (i.e.\ in which all ISPs have the same total bandwidth); we remove this requirement in \cref{heterogeneous}.

\begin{definition}[Quality-of-Service Space]
For ease of presentation, we use $\timeset\eqdef[0,1]$ as the \emph{type space} in the consumer game (and later as the \emph{strategy space} in the producer game). We consider lower values as indicating better qualities of service.
\end{definition}

For the duration of this \lcnamecref{consumers}, fix a finite measure $\mu$ on $\timeset$, a natural $n \in \mathbb{N}$ and producer QoS levels (e.g.\ ISP latencies)
$\bar{t}=(t_0,\ldots, t_{n-1}) \in \timeset^{\prods}$. We consider the $n$-producers consumer game $(\mu;\bar{t})=(\mu;t_0,\ldots,t_{n-1})$, which we now define.

\begin{definition}[Strategies]
For every $d \in \timeset$, we define the set of \emph{strategies} available to a player with type (i.e.\ worst acceptable QoS) $d$
as $S_d\eqdef\{j \mid t_j \le d\}\cup\{\noconsumption\}$, where $\noconsumption$ denotes not consuming from any producer. We define $S\eqdef\cup_{d \in \timeset}S_d=\prods\cup\{\noconsumption\}$ --- the set of pure strategies available to any player, and consider $S$ as a measurable space with the $\sigma$-algebra $2^S$ of all of its subsets.
\end{definition}

\begin{definition}[Pure-Consumption Profile/Nash Equilibrium]
\leavevmode
\begin{enumerate} 
\item
A \emph{pure-consumption (strategy) profile} in the $n$-producers consumer game $(\mu;\bar{t})$ is a measurable function $s:\timeset\rightarrow S$ s.t.\ $s(d) \in S_d$ for every $d \in \timeset$.
\item
Given a pure-consumption profile $s$ in the $n$-producers consumer game $(\mu;\bar{t})$, we define $\load{j}\eqdef\mu\bigl(s^{-1}(j)\bigr)$ for every $j \in S$ ---
the load on producer $j$. ($\load{\noconsumption}$ is the measure of consumers not consuming from any producer.)
\item
A \emph{pure-consumption Nash equilibrium} in the $n$-producers consumer game $(\mu;\bar{t})$ is a pure-consumption profile $s$ s.t.\ for every $d\in\timeset$,
both the following hold.
\begin{enumerate}
\item
$s(d)=\noconsumption$ only if $S_d=\{\noconsumption\}$.
\item
$\load{s(d)} \le \load{j}$ for every $j \in S_d\setminus\{\noconsumption\}$.\footnote{\label{heterogeneous-eq}As mentioned above, our results generalize also for a more general definition of the consumers' preferences (capturing, e.g. scenarios in which different ISPs have different bandwidths), in which each consumer consumes from a producer $j$ with minimal $f_j\bigl(\load{j}\bigr)$ (as opposed
to minimal~$\load{j}$), where $(f_j)_{j\in\prods}$ is a specification of an increasing continuous function for each producer. See \cref{heterogeneous} for more details.}
\end{enumerate}
\end{enumerate}
\end{definition}

We now turn to define mixed-consumption strategies. We think of such a strategy not as a probabilistic one, but rather as meaning ``a certain fraction of the continuum of players with type $d$ have one strategy, while others have other strategies''.

\begin{definition}[Mixed-Consumption Profile/Nash Equilibrium]
\leavevmode
\begin{enumerate} 
\item
A \emph{mixed-consumption (strategy) profile} in the $n$-producers consumer game $(\mu;\bar{t})$ is a measurable function $s:\timeset\rightarrow [0,1]^S$ s.t.\ $s(d) \in \Delta^{S_d}$ for every $d \in \timeset$.
\item
Given a mixed-consumption profile $s$ in the $n$-producers consumer game $(\mu;\bar{t})$, we define $\load{j}\eqdef\int_{\timeset}s_j\,d\mu$ for every $j \in S$ ---
the load on producer $j$. ($\load{\noconsumption}$ is the measure of consumers not consuming from any producer in this case as well.)
\item
A \emph{mixed-consumption Nash equilibrium} in the $n$-producers consumer game $(\mu;\bar{t})$ is a mixed-consumption profile $s$ s.t.\ for every $d \in \timeset$,
both of the following hold.
\begin{enumerate}
\item
$\noconsumption\in\supp\bigl(s(d)\bigr)$ only if $S_d=\{\noconsumption\}$.
\item
$\load{k} \le \load{j}$ for every $k \in \supp\bigl(s(d)\bigr)$ and $j \in S_d\setminus\{\noconsumption\}$.\textsuperscript{\ref{heterogeneous-eq}}
\end{enumerate}
\end{enumerate}
\end{definition}

\begin{theorem}[$\exists$ Pure-Consumption Nash Equilibrium]\label{consumer-pure-nash-exists}
If $\mu$ is atomless, then a pure-consumption Nash equilibrium exists in the $n$-producers consumer game $(\mu;\bar{t})$.
\end{theorem}

\begin{example}[Necessity of Atomlessness Condition]
Consider a nonzero measure $\mu$ concentrated entirely on the atom $d=1 \in \timeset$. For $n>1$, no pure-consumption Nash equilibrium exists in any induced $n$-producers consumer game.
Indeed, in any pure-consumption profile, all consumers with type $d=1$ would consume from the same producer, leaving another producer with a strictly lower load of $0$; as
this producer is acceptable by all consumers with type $d=1$, they would all rather deviate to it.
\end{example}

\begin{definition}[Effective Type]
We say that two types $d_1,d_2 \in \timeset$ are of the same \emph{effective type} if $S_{d_1}=S_{d_2}$.
\end{definition}

The Nash equilibrium constructed in the proof of \cref{consumer-pure-nash-exists} is asymmetric in
the sense that players with the same effective type may behave differently. As we momentarily below, this asymmetry cannot be avoided. Nonetheless,
a reader who finds this asymmetry aesthetically unpleasing may instead consider a more-symmetric, yet mixed-consumption, Nash equilibrium, which in fact exists even when $\mu$ is not atomless.

\begin{definition}[Symmetric Strategy Profile]
A strategy profile $s$ is said to be \emph{symmetric} if $S_{d_1}=S_{d_2}\Longrightarrow s(d_1)=s(d_2)$ for every $d_1,d_2 \in \timeset$, i.e.\ each player's strategy depends only on the player's effective type.
\end{definition}

\begin{example}[Nonexistence of a Symmetric Pure-Strategy Nash Equilibrium]
Consider any nonzero measure $\mu$. For $n>1$, if $t_j=0$ for every $j\in\prods$, then all consumers are of the same effective type. Thus, no symmetric pure-consumption equilibrium exists in the induced $n$-producers consumer game.
Indeed, in any symmetric pure-consumption profile, since all consumers are of the same effective type, all would consume from the same producer, leaving another producer (acceptable to
all) with a strictly lower load of $0$; therefore, all consumers would rather deviate to this producer.
\end{example}

\begin{theorem}[$\exists$ Symmetric Mixed-Consumption Nash Equilibrium]\label{consumer-symmetric-nash-exists}
A symmetric mixed-consumption Nash equilibrium exists in the $n$-producers consumer game $(\mu;\bar{t})$.
Furthermore, there exists such an equilibrium for which the strategies can be computed in $O(n^2)$ time.
\end{theorem}

See \cref{vessels} for an illustration of the constructive proof of \cref{consumer-symmetric-nash-exists}; as illustrated, the intuition underlying this novel construction
builds upon hydraulic systems of communicating vessels (nonetheless, the proofs given in \cref{consumers-proofs} are completely formal, of course).\begin{figure}[p]
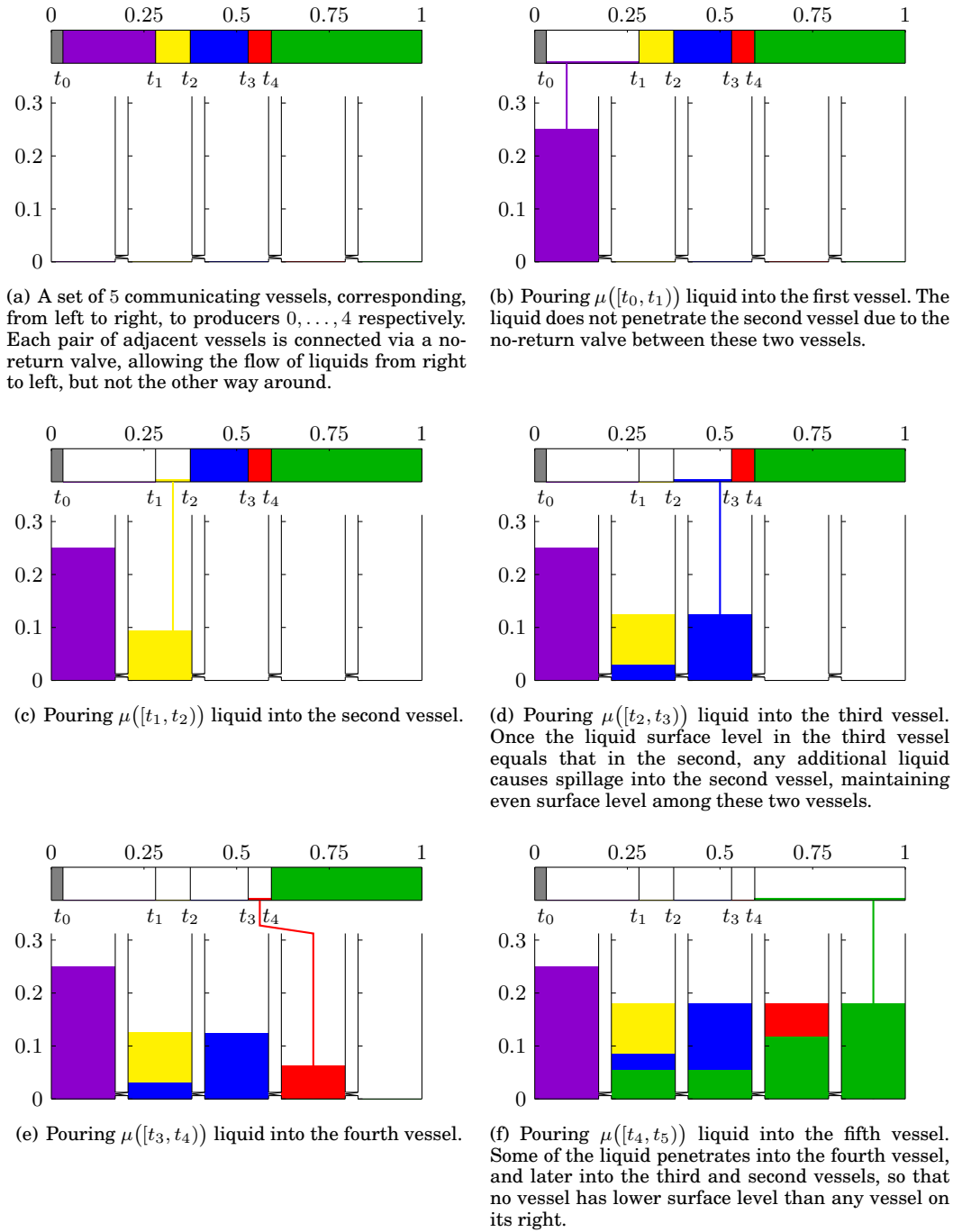
%
\centering%
\subfigure[A set of $5$ communicating vessels, corresponding, from left to right, to producers $0,\ldots,4$ respectively. Each pair of adjacent vessels is connected via a no-return valve, allowing the flow of liquids from right to
left, but not the other way around.]{%
\vessels{1}%
}\quad
\subfigure[Pouring $\mu\bigl([t_0,t_1)\bigr)$ liquid into the first vessel. The liquid does not penetrate the second vessel due to the no-return valve between these two vessels.]{%
\vessels{2}%
}
\subfigure[Pouring $\mu\bigl([t_1,t_2)\bigr)$ liquid into the second vessel.]{%
\vessels{3}%
}\quad
\subfigure[Pouring $\mu\bigl([t_2,t_3)\bigr)$ liquid into the third vessel. Once the liquid surface level in the third vessel equals that in the second, any additional liquid causes spillage into the second vessel, maintaining even surface level among these two vessels.]{%
\vessels{4}%
}
\subfigure[Pouring $\mu\bigl([t_3,t_4)\bigr)$ liquid into the fourth vessel.]{%
\vessels{5}%
}\quad
\subfigure[Pouring $\mu\bigl([t_4,t_5)\bigr)$ liquid into the fifth vessel. Some of the liquid penetrates into the fourth vessel, and later into the third and second vessels, so that no vessel has lower surface level than any vessel on its right.]{%
\label{vessels:green}%
\vessels{6}%
}%
\caption{%
Illustration of the construction in the proof of \cref{consumer-symmetric-nash-exists} for $n=5$. E.g.\ as exactly~$80\%$ of the blue (i.e.\ darkest when viewed in b/w) liquid in \cref{vessels:green} is in the third vessel and the remaining $20\%$ is in the second one, the strategy for all consumer types $d \in [t_2,t_3)$ in the symmetric mixed-consumption Nash equilibrium that we construct is $0.8$ consumption from producer~$2$ and $0.2$ consumption from producer~$1$.%
}%
\label{vessels}%
\end{figure}
We now show that while in general many Nash equilibria may exist in the consumer game, they result in the same load for both consumers and producers.

\begin{theorem}[Producers are Indifferent between Nash Equilibria]\label{indifference-producers}
$\load{j}=\loadt{j}$ for every $j\in\prods$ and every mixed-consumption Nash equilibria $s,s'$ in $(\mu;\bar{t})$.
\end{theorem}

\begin{corollary}[Consumers are Indifferent between Nash Equilibria]\label{indifference-consumers}
$\load{k}=\loadt{k'}$ for every $k\in\supp\bigl(s(d)\bigr)$ and $k'\in\supp\bigl(s'(d)\bigr)$, for every $d\in\timeset$ and every mixed-consumption Nash equilibria $s,s'$ in $(\mu;\bar{t})$.
\end{corollary}

By \cref{consumer-symmetric-nash-exists,indifference-producers}, the following is well defined.

\begin{definition}[Producer Load]\label{ell}
For every $j\in\prods$, we define $\ell_j(\bar{t})$ to equal $\load{j}$ in any mixed-consumption Nash equilibrium $s$ in $(\mu;\bar{t})$.
\end{definition}

By the proof of \cref{consumer-symmetric-nash-exists}, we obtain
\cref{compute-ell} --- a
simple algorithm for directly calculating $\ell_j(\bar{t})$ for all $j$, without the need to first calculate consumer's strategies. While this algorithm runs in $O(n^2)$ time, i.e.\
has same worst-case asymptotic behaviour as explicitly computing a Nash equilibrium via \cref{consumer-symmetric-nash-exists} and then
deducing all loads, it is considerably simpler, and also computes the loads sequentially, and so may be stopped mid-way,
allowing to calculate the loads on the $j$ producers with lowest latency levels in $O(j\cdot n)$ time for any $j$.
\begin{algorithm}[ht]%
\floatname{algorithm}{ALGORITHM}%
\caption{Direct computation of $\ell_j(\bar{t})$ for all $j\in\prods$}\label{compute-ell}%
\small%
\begin{algorithmic}[1]%
\Procedure{Compute-$\ell$}{$\mu;t_0,\ldots,t_{n-1}$}\Comment{Assumes $t_0\le t_1\le\cdots\le t_{n-1}$.}
\State $t_n \gets 2$\Comment{Any value $>1$ will do here; assumes $\mu$ is defined on $[0,t_n]$, but has support $\timeset$.}
\State $\ell_{\noconsumption} \gets \mu\bigl([0,t_0)\bigr)$
\State $k \gets 0$
\While{$k<n$}
\State $k' \gets \Max\;\arg\Max_{k<k'\le n}\frac{\mu([t_k,t_{k'}))}{k'-k}$
\State $\ell \gets \frac{\mu([t_k,t_{k'}))}{k'-k}$
\ForAll{$k\le j<k'$}
\State $\ell_j \gets \ell$
\EndFor
\State $k \gets k'$
\EndWhile
\State \Return{$(\ell_{\noconsumption},\ell_0,\ldots,\ell_{n-1})$}
\EndProcedure
\end{algorithmic}%
\end{algorithm}

See \cref{ell-analysis} for an analytic study of $\ell_j$, formalizing some main properties thereof, which we utilize in our proofs in the following \lcnamecrefs{producers}. In particular,
we show there that for every $j$, $\ell_j(\bar{t})$ is nonincreasing in $t_j$, weakly quasiconvex in $t_k$ for $k\ne j$, and Lipschitz (w.r.t.\ $\mu$) in each coordinate with Lipschitz constant $1$.

\section{The Producer (ISP) Game}\label{producers}

We now turn to the producer game, and to the main results of this paper. In this two-stage game, each producer chooses a strategy (i.e.\ QoS) in $\timeset$, and the utilities are determined according to the loads on producers in Nash equilibria in the induced consumer game.
For the duration of this \lcnamecref{producers}, fix a natural $n \in \mathbb{N}$ and a finite measure $\mu$ on $\timeset$.
Full proofs, as well as auxiliary results, are provided in \cref{producers-coarse-proofs,producers-fine-proofs}.

In \cref{producers-coarse}, we define a simplified version of the producer game; the definition of the (more intricate) producer game that is surveyed in the introduction is given in \cref{producers-fine}. While the simpler game defined in \cref{producers-coarse} has some trivialities that we point out, its analysis is nonetheless interesting,
and the obtained results are useful when analysing the more-involved version in \cref{producers-fine}.

Recall that as in \cref{consumers}, for ease of presentation we present a model in which each consumer would like to consume from a least-loaded producer (i.e.\ in which all ISPs have the same total bandwidth); as noted above,
we remove this requirement in \cref{heterogeneous}.

\subsection{Coarse Preferences (A Simplified Producer Game)}\label{producers-coarse}

\begin{definition}[Producer Game with Coarse Preferences]
We define the \emph{producer game with coarse preferences} \coarsegame\ as the $n$-player game, with set of players (called \emph{producers}) $\prods$, in which the pure-strategy space available to each producer
is $\timeset$, and in which for each pure-strategy profile $\bar{t}\in\timeset^{\prods}$, the utility for each producer $j\in\prods$ is strictly increasing in $\ell_j(\bar{t})$ (as defined in \cref{ell}).
\end{definition}

\subsubsection{Static Analysis}\label{producers-coarse-statics}

We begin with an analysis of domination in the producer game with coarse preferences, pointing out the trivialities in this simplified game, which will disappear in the more-involved
version thereof that we analyse in \cref{producers-fine}.

\begin{definition}[Safe Alternative; Dominant Strategy]\label{dominant}
Let $t$ be a strategy in the game \coarsegame.
\begin{itemize}
\item
We say that $t$ is a \emph{safe alternative to} some strategy $t'$ if for every strategy profile for all but one of the
producers, playing $t$ gives the remaining producer utility at least as high a utility as playing $t'$.
\item
We say that $t$ is a \emph{dominant strategy} if it is a safe alternative to all strategies.
\end{itemize}
\end{definition}

\begin{theorem}[Dominant Strategies]\label{producer-coarse-dominant}
$t \in \timeset$ is a dominant strategy in \coarsegame\ iff $\mu\bigl([0,t)\bigr)=0$.
Furthermore, each such dominant strategy guarantees a load of at least $\frac{\mu(\timeset)}{n}$ on each producer playing it.
\end{theorem}

In particular, we have that every producer playing $0\in\timeset$ constitutes a Nash equilibrium. (We emphasize that this is by far not the only Nash equilibrium --- see \cref{producer-coarse-nash-char} below.)
This and other trivialities that result from domination (as well as the domination itself) disappear in \cref{producers-fine}, when we refine the order of preferences of the various producers. Before that, though, we
continue to explore the consumer game with coarse preferences, obtaining results that aid our analysis of the consumer game with refined preferences in \cref{producers-fine} below.
Our next step is to not only characterize the Nash equilibrium loads (an immediate corollary of \cref{producer-coarse-dominant}), but furthermore, show that every strategy profile inducing these loads is a Nash equilibrium.

\begin{theorem}[Nash Equilibrium Loads]\label{producer-coarse-nash-loads}
A pure-strategy profile $\bar{t} \in \timeset^{\prods}$ constitutes a Nash equilibrium in \coarsegame\ iff $\ell_j(\bar{t})=\frac{\mu(\timeset)}{n}$ for every $j\in\prods$.
\end{theorem}

We proceed to directly characterize the strategies played in Nash equilibria, in a way that does not necessitate solving the induced consumer game.

\begin{theorem}[Nash Equilibrium Characterization]\label{producer-coarse-nash-char}
Let $t_0\le\cdots\le t_{n-1} \in \timeset$. The pure-strategy profile $\bar{t}\eqdef(t_1,\ldots,t_{n-1})$ constitutes a Nash equilibrium in \coarsegame\ iff $\mu\bigl([0,t_j)\bigr)\le\frac{j}{n}\cdot\mu(\timeset)$ for
every $j\in\prods$.
\end{theorem}

It should be emphasized that \cref{producer-coarse-nash-char} does \emph{not} imply that Nash equilibria are interchangeable (i.e.\ that the set of Nash equilibria is a Cartesian product
of sets of strategies for the various producers). Consider,
for example, $\mu=U(\timeset)$ --- the uniform measure on $\timeset$.
In this case, by \cref{producer-coarse-nash-char}, $(0,\frac{1}{n},\frac{2}{n},\ldots,\frac{n-1}{n})$
is a Nash equilibrium in \coarsegame, and so is any permutation thereof.
Nonetheless, every player playing $\frac{n-1}{n}\in\timeset$ does not constitute a Nash equilibrium. We now move on to examine the stability of
the Nash equilibria in \coarsegame\ against group deviations.

The study of stability against group deviations was initiated by \citeN{Aumann59}, who considers
deviations from which all deviators gain. Recently, the CS literature considers a considerably stronger solution concept, according to which a deviation is considered beneficial even if
only some of the participants in the deviating coalition gain, as long as none of the participants lose (see e.g.\ \cite{Rozenfeld06}). While stability against the classical all-gaining coalitional deviation is termed \emph{strong equilibrium}, this more-demanding concept is referred to as \emph{super-strong equilibrium}; there are very few results showing its existence in nontrivial settings.

\begin{theorem}[All Nash Equilibria are Super-Strong]\label{producer-coarse-super-strong}
Let $\bar{t} \in \timeset^{\prods}$ be a pure-strategy Nash equilibrium in \coarsegame.
There exist no coalition $P\subseteq\prods$ and strategies $\bar{t}'=(t'_j)_{j\in P}\in\timeset^P$ s.t.\ $\ell_j(\bar{t}_{-P},\bar{t}')\ge\ell_j(\bar{t})$ for every $j\in P$, with a strict
inequality for at least one producer $j\in P$.
\end{theorem}

We conclude the static analysis of \coarsegame\ by deducing generalizations of \cref{producer-coarse-dominant,producer-coarse-nash-loads,producer-coarse-nash-char,producer-coarse-super-strong}
for mixed-strategy profiles, as well as showing that no mixed-strategy Nash equilibrium exhibits any ex-post regret.

\begin{theorem}[Mixed Strategies]\label{producer-coarse-mixed}
In \coarsegame,
\begin{parts}
\item\label{producer-coarse-mixed-dominant}\thmitemtitle{Dominant Strategies}
Let $p$ be a mixed strategy.\footnote{We consider a \emph{mixed strategy} to be a random variable taking values in $\timeset$.} $p$ is a dominant strategy iff $\mu\bigl([0,\Max\supp(p))\bigr)=0$.
Furthermore, each such dominant strategy guarantees a load of at least $\frac{\mu(\timeset)}{n}$ with probability $1$ on each producer playing it.
\item\label{producer-coarse-mixed-nash-loads}\thmitemtitle{Nash Equilibrium Loads}
A mixed-strategy profile $\bar{p}=(p_0,\ldots,p_{n-1})$\footnote{We emphasize that mixed-strategies of distinct producers are independent random variables.} constitutes a Nash equilibrium iff $\ell_j(\bar{p})=\frac{\mu(\timeset)}{n}$ for every $j\in\prods$ with probability $1$.
\item\label{producer-coarse-mixed-nash-char}\thmitemtitle{Nash Equilibrium Characterization}
A mixed-strategy profile $\bar{p}$ constitutes a Nash equilibrium iff there exists a permutation on the producers $\pi\in\prods!$ s.t.\ $\mu\bigl([0,\Max\supp(p_{\pi(j)}))\bigr)\le\frac{j}{n}\cdot\mu(\timeset)$ for every $j\in\prods$.
\item\label{producer-coarse-mixed-super-strong}\thmitemtitle{All Nash Equilibria are Super-Strong}
Let $\bar{p}$ be a mixed-strategy Nash equilibrium.
There exist no coalition $P\subseteq\prods$ and mixed strategies $\bar{p}'=(p'_j)_{j\in P}$ s.t.\ $\expect{\ell_j(\bar{p}_{-P},\bar{p}')}\ge\expect{\ell_j(\bar{p})}$ for every $j\in P$, with a strict inequality for at least one producer $j\in P$.
\end{parts}
\end{theorem}

\begin{theorem}[No Ex-Post Regret in Mixed-Strategy Nash Equilibria]\label{producer-coarse-mixed-nash-no-regret}
In any mixed-strategy Nash equilibrium in \coarsegame, with probability $1$ there exists no ex-post regret for any producer. In other words, a realization of a mixed-strategy Nash equilibrium
is with probability $1$ a pure-strategy Nash equilibrium.
\end{theorem}

\subsubsection{Dynamics}\label{producers-coarse-dynamics}
When analysing dynamics henceforth, we assume that $\mu(\timeset)>0$. (Otherwise, by \cref{producer-coarse-dominant}, all strategies are equivalent and so the analysis is trivial.)

\begin{definition}[Schedule; Sequential/Simultaneous Schedule; Round]\leavevmode
\begin{enumerate}
\item
A \emph{schedule} is a sequence $(P_i)_{i=0}^{\infty}$ of nonempty subsets of $\prods$, s.t.\ $j\in P_i$ for infinitely many values of $i\in\mathbb{N}$, for every $j\in\prods$.
\item
A schedule $(P_i)_{i=0}^{\infty}$ is said to be \emph{sequential} if $|P_i|=1$ for every $i\in\mathbb{N}$.
\item
A schedule $(P_i)_{i=0}^{\infty}$ is said to be \emph{simultaneous} if $P_i=\prods$ for every $i\in\mathbb{N}$.
\item
Let $i_1\le i_2\in\mathbb{N}$. We say that $\{i\in\mathbb{N} \mid i_1 \le i \le i_2\}$ constitutes a \emph{round} (in the schedule $(P_i)_{i=0}^{\infty}$) if
$\cup_{i=i_1}^{i_2}P_i=\prods$. (We emphasize that this union need not be a disjoint union.)
\item
Let $i_1\le i_2\in\mathbb{N}$ and let $r\in\mathbb{N}$. We say that \emph{$i_2$ is reached from $i_1$ in $r$ rounds} if $r-1$ is the largest number of pairwise-disjoint rounds into which $\{i_1,i_1+1,\ldots,i_2-2\}$ can be partitioned. (Therefore, $\{i_1,i_1+1,\ldots,i_2-1\}$ cannot be partitioned into $r$ pairwise-disjoint rounds with a nonzero amount of ``spare'' trailing steps.)
\end{enumerate}
\end{definition}

\begin{remark}
In a simultaneous schedule, each \emph{step} $\{i\}$ constitutes a round.
\end{remark}

\begin{definition}[Weakly-/$\delta$-Better-/Best-Response Dynamics; Lazy Dynamics]\label{dynamics}\leavevmode
\begin{enumerate}
\item
A \emph{weakly-better-response dynamic} in \coarsegame\ is a sequence $(\bar{t}_i,P_i)_{i=0}^{\infty}$, where $(P_i)_{i=0}^{\infty}$ is a schedule
and $(\bar{t}^i)_{i=0}^{\infty}$ is a sequence of strategy profiles  s.t.\ both of the following hold
for every $i\in\mathbb{N}$.
\begin{itemize}
\item
For every $j \in P_i$,\ \ $t^{i+1}_j$ is a weakly better response than $t^i_j$ to $\bar{t}^i_{-j}$ (by $j$), i.e.\ $\ell_j(\bar{t}^i_{-j},t^{i+1}_j)\ge\ell_j(\bar{t}^i)$.
\item
For every $j \notin P_i$,\ \ $t^{i+1}_j=t^i_j$.
\end{itemize}
By slight abuse of notation, we sometimes write $(\bar{t}_i)_{i=0}^{\infty}$ to refer to $(\bar{t}_i,P_i)_{i=0}^{\infty}$, when the schedule is either inconsequential or clear from context.
\item
A weakly-better-response dynamic is said to be a \emph{best-response dynamic} if for every $i\in\mathbb{N}$ and $j \in P_i$,\ \ $t^{i+1}_j$ is a best response to $\bar{t}^i_{-j}$,
i.e.\ $t^{i+1}_j\in\arg\Max_{t\in\timeset}\ell_j(\bar{t}^i_{-j},t)$.
\item
Let $\delta>0$. A weakly-better-response dynamic is said to be a \emph{$\delta$-better-response dynamic} if for every $i\in\mathbb{N}$ and $j \in P_i$,\ \ $t^{i+1}_j$ is either a best response to $\bar{t}^i_{-j}$, or a better response increasing $j$'s load by at least $\delta$ compared to $t^i_j$, i.e.\
$\ell_j(\bar{t}^i_{-j},t^{i+1}_j)\ge\ell_j(\bar{t}^i)+\delta$.\footnote{Due to the continuous nature of strategies and loads, we require an improvement by at least
$\delta$, and not just any positive improvement, in order to avoid improvements \emph{\`{a} la} Zeno's ``Race Course'' paradox.}
\item
A weakly-better-response dynamic is said to be \emph{lazy} if for every $i\in\mathbb{N}$ and $j \in P_i$,\ \ $t^{i+1}_j=t^i_j$ whenever $t^i_j$ is a best response to $\bar{t}^i_{-j}$.
\end{enumerate}
\end{definition}

\begin{remark}\label{coarse-better-vs-best}
In \coarsegame,
\begin{itemize}
\item
Every best-response dynamic is a $\delta$-better-response dynamic, for every $\delta>0$.
\item
Every $\delta$-better-response dynamic is also a $\delta'$-better-response one, for every $0<\delta'<\delta$.
\item
A weakly-better-response dynamic is a best-response dynamic iff it is a $\delta$-better-response dynamic for $\delta=\mu(\timeset)$.
\end{itemize}
\end{remark}

\begin{remark}[A Best Response Always Exists]\label{coarse-best-response}
Let $j\in\prods$ and let $\bar{t}_{-j}\in\timeset^{\prods\setminus\{j\}}$. By \cref{producer-coarse-dominant}, a best response (by $j$) to $\bar{t}_{-j}$ exists in \coarsegame.
\end{remark}

We commence with a negative result, showing that even best-response dynamics can go out of equilibrium.

\begin{example}[Nonsequential Nonlazy Best-Response Dynamics may Go Out of Equilibrium]\label{coarse-bad-response}
Let $\mu=U(\timeset)$. By \cref{producer-coarse-nash-char}, the (cyclically repeating) strategy-profile sequence
$(0,0,\ldots,0)$, $(\frac{n-1}{n},\frac{n-1}{n},\ldots,\frac{n-1}{n})$, $(0,0,\ldots,0)$, $(\frac{n-1}{n},\frac{n-1}{n},\ldots,\frac{n-1}{n})$, \ldots\
constitutes a (nonlazy) simultaneous best-response dynamic in \coarsegame\ that visits nonequilibria infinitely often.
\end{example}

We continue by showing that the dynamic in \cref{coarse-bad-response} visiting Nash equilibria infinitely often is no coincidence.

\begin{theorem}[$\delta$-Better-Response Dynamics Visit Nash Equilibria Infinitely Often]\label{coarse-response-equilibrium}
Let $\delta>0$ and let $(\bar{t}^i)_{i=0}^{\infty}$ be a $\delta$-better-response dynamic in \coarsegame.
$\bar{t}^i$ is a Nash equilibrium for infinitely many values of $i$. Moreover, the first Nash equilibrium is reached (from $0$) in at most $n\cdot\bigl\lceil\frac{\mu(\timeset)}{\delta n}\bigr\rceil$ rounds,
and from any later nonequilibrium, the next Nash equilibrium is reached in at most $(n-1)\cdot\bigl\lceil\frac{\mu(\timeset)}{\delta n}\bigr\rceil$ rounds.
\end{theorem}

\begin{remark}\label{coarse-response-equilibrium-slightly-faster}In \cref{coarse-response-equilibrium},
\begin{itemize}
\item
if $(\bar{t}^i)_{i=0}^{\infty}$ is simultaneous, then ``rounds'' may be replaced with ``steps''.
\item
Finer analysis of similar nature may be used to show both that $n\cdot\bigl\lceil\frac{\mu(\timeset)}{\delta n}\bigr\rceil$ may be replaced with $\sum_{h=1}^{n}\bigl\lceil\frac{\mu(\timeset)}{h\delta n}\bigr\rceil\approx\Max\bigl\{\ln n\cdot\bigl\lceil\frac{\mu(\timeset)}{\delta n}\bigr\rceil,n\bigr\}$, and that $(n-1)\cdot\bigl\lceil\frac{\mu(\timeset)}{\delta n}\bigr\rceil$ may be replaced with
$\sum_{h=2}^{n}\bigl\lceil\frac{\mu(\timeset)}{h\delta n}\bigr\rceil\approx\Max\bigl\{(\ln n-1)\cdot\bigl\lceil\frac{\mu(\timeset)}{\delta n}\bigr\rceil,n-1\bigr\}$. We conjecture that considerably tighter bounds (esp.\ for small $\delta$) can be attained as well.
\end{itemize}
\end{remark}

We now show that in a sense, \cref{coarse-bad-response} describes all the ``issues'' that might prevent best- and even $\delta$-better-response dynamics from remaining in Nash equilibria.

\begin{remark}[Lazy Better-Response Dynamics Remain in Nash Equilibrium]\label{coarse-lazy}
Once a lazy weakly-better-response dynamic reaches a Nash equilibrium, it remains constant. (Directly by definition of Nash equilibrium and laziness.)
\end{remark}

\begin{theorem}[Sequential Better-Response Dynamics Remain in Nash Equilibria]\label{coarse-sequential}
Let $(P_i)_{i=0}^{\infty}$ be a schedule. If $(P_i)_{i=0}^{\infty}$ is sequential from some point, then once a $(P_i)_{i=0}^{\infty}$-scheduled weakly-better-response dynamic in \coarsegame\ reaches a Nash equilibrium after that point, it never visits a nonequilibrium afterward.
\end{theorem}

\begin{corollary}[Sequential/Lazy $\delta$-Better-Response Dynamics Reach Nash Equilibria and Remain]\label{coarse-sequential-lazy-cor}
For every $\delta>0$, every sequential or lazy $\delta$-better-response dynamic in \coarsegame\ reaches a Nash equilibrium in a finite number of steps,
and never visits a nonequilibrium after that point.
\end{corollary}

\begin{proposition}[Every Nonsequential Schedule and Initial Profile have Nonlazy Best-Response Dynamics that Go Out of Equilibrium]\label{coarse-nonsequential}
If $\mu$ has no atom measuring $\frac{n-1}{n}\cdot\mu(\timeset)$ or more and no tail of $(P_i)_{i=0}^{\infty}$ is sequential, then for every pure-strategy profile $\bar{t}^0$ there exists a nonlazy best-response dynamic in \coarsegame\ that is scheduled by $(P_i)_{i=0}^{\infty}$, starts at $\bar{t}^0$ and visits nonequilibria infinitely often.
\end{proposition}

\begin{remark}
Analogues of \cref{coarse-lazy,coarse-sequential,coarse-nonsequential} hold for mixed-strategy dynamics as well.
\end{remark}

As every best-response dynamic is a $\delta$-better-response one for $\delta=\mu(\timeset)$, we conclude from \cref{coarse-response-equilibrium} that such a dynamic reaches a Nash equilibrium in at most \mbox{$n\cdot\bigl\lceil\frac{\mu(\timeset)}{\mu(\timeset) n}\bigr\rceil=n$} rounds and afterward always ``re-reaches'' a Nash equilibrium in at most $(n-1)\cdot\bigl\lceil\frac{\mu(\timeset)}{\mu(\timeset) n}\bigr\rceil=n-1$ rounds. By applying some finer analysis, we can slightly improve this bound, and show that the  new bound is tight.

\begin{theorem}[Best-Response Time-to-Equilibrium and Time between Equilibria]\label{coarse-best-response-fast}
Every best-response dynamic in \coarsegame\ reaches a Nash equilibrium in at most $n-1$ rounds. Furthermore, if $n>2$, then from any later nonequilibrium,
the next Nash equilibrium is reached in at most $n-2$ rounds.
\end{theorem}

\begin{remark}
In \cref{coarse-best-response-fast}, as in \cref{coarse-response-equilibrium}, if the dynamic in question is simultaneous, then ``rounds'' may be replaced with ``steps''.
\end{remark}

\begin{example}[Tightness of \cref{coarse-best-response-fast}]\label{coarse-best-response-fast-tight}
Let $\mu=U(\timeset)$.
The following is a (nonlazy) simultaneous best-response dynamic in \coarsegame, in which\ \ i)~no two consecutive strategy profiles are both Nash equilibria,\ \ ii)~the first Nash equilibrium is reached in precisely $n-1$ rounds (steps), and\ \ iii)~from any
nonequilibrium that follows a Nash equilibrium, the next Nash equilibrium is reached in precisely $n-2$ rounds (steps):
$(1,0,0,\ldots,0)$, $(\frac{n-1}{n},\frac{n-2}{n-1},0,0,\ldots,0)$, $(\frac{n-2}{n},\frac{n-2}{n},\frac{n-3}{n-1},0,0,\ldots,0)$,
\ldots, $(\frac{3}{n},\frac{3}{n},\ldots,\frac{3}{n},\frac{2}{n-1},0,0)$, $(\frac{2}{n},\frac{2}{n},\ldots,\frac{2}{n},\frac{1}{n-1},0)$,
$(\frac{1}{n},\frac{1}{n},\ldots,\frac{1}{n},0)$ (first Nash equilibrium),
$(\frac{n-1}{n},\frac{n-2}{n-1},0,0,\ldots,0)$, $(\frac{n-2}{n},\frac{n-2}{n},\frac{n-3}{n-1},0,0,\ldots,0)$, \ldots\
(cyclically repeating).
\end{example}

To summarize \cref{coarse-best-response-fast,coarse-best-response-fast-tight,coarse-lazy,coarse-sequential}:

\begin{corollary}[Sequential/Lazy Best-Response Dynamics Reach Nash Equilibria Fast and Remain]\label{coarse-best-response-fast-cor}
Every sequential or lazy best-response dynamic in \coarsegame\ reaches a Nash equilibrium in at most a tight bound of $n-1$ rounds,
and never visits a nonequilibrium after that point.
\end{corollary}

\subsection{Fine Preferences}\label{producers-fine}

\begin{definition}[Producer Game with Fine Preferences]
We define the \emph{producer game with fine preferences} \finegame\ as the $n$-player game, with set of players (called \emph{producers}) $\prods$, in which the pure-strategy space available to each producer
is $\timeset$, and in which for each pure-strategy profile $\bar{t}\in\timeset^{\prods}$, the utility for each producer $j\in\prods$ is strictly increasing in $\ell_j(\bar{t})$ (as defined in \cref{ell}), with tie breaking (i.e.\ infinitesimal improvement) in favour of larger values of $t_j$ over smaller ones.
\end{definition}

\subsubsection{Static Analysis}\label{producers-fine-statics}

We define safe alternatives and dominant strategies in \finegame\ as in \cref{dominant}, only w.r.t.\ fine preferences.
The following proposition shows that the tie-breaking refinement of the producers' preferences into ``fine preferences'' indeed
successfully removes the triviality captured by \cref{producer-coarse-dominant}, in a strong sense.

\begin{proposition}[(No) Dominant and (Few) Dominated Strategies]\label{producer-fine-dominant}
If~$\mu$ has no atom measuring $\frac{n-1}{n}\cdot\mu(\timeset)$ or more,\footnote{The triviality in case of a very large atom should
be compared to the triviality captured under the exact condition in \cref{coarse-nonsequential}. Indeed, both trivialities are possible exactly iff there exists $t\in\timeset$  s.t.\ $ \mu\bigl([0,t)\bigr)=0$ and
$\mu\bigl([t',1]\bigr)<\frac{\mu(\timeset)}{n}$ for every $t'>t$.}
then no strategies are dominant in \finegame.
Moreover, at least $\frac{n-1}{n}$ of the strategies in $\timeset$ (as measured by $\mu$) have no safe alternatives (other than themselves).
\end{proposition}

We now formally conclude the results captured informally in \cref{intro-producer-nash}:

\begin{theorem}[$\exists!$ Nash Equilibrium, and it is Super-Strong\footnote{See \cref{producer-coarse-super-strong} in
\cref{producers-coarse} above, as well as the preceding discussion, for the definition of super-strong equilibrium, as well as a discussion regarding various group-deviation concepts.}]\label{producer-fine-nash-char}
A unique (up to permutations) pure-strategy Nash equilibrium exists in \finegame.
The sorted Nash-equilibrium strategies $t_0\le\cdots\le t_{n-1} \in \timeset$ are
$t_j\eqdef\Max\bigl\{t \in \timeset \mid \mu\bigl([0,t)\bigr)\le\frac{j}{n}\cdot\mu(\timeset)\bigr\}$ for every $j\in\prods$. The load on each producer in this equilibrium is $\frac{\mu(\timeset)}{n}$. Furthermore, this equilibrium is super-strong.
\end{theorem}

\begin{corollary}[Nash Equilibrium Characterization --- Special Case]\label{producer-fine-nash-char-special}
If the CDF of $\mu$ is continuous (i.e.\ $\mu$ is atomless) and strictly increasing, then for every $j\in\prods$, the $j$\tth\ sorted Nash-equilibrium strategy, $t_j$, is the unique strategy
satisfying $\mu\bigl([0,t_j)\bigr)=\frac{j}{n}\cdot\mu(\timeset)$.
\end{corollary}

\begin{proposition}[Nash Equilibria are in Pure Strategies]\label{producer-fine-nash-pure}
Every mixed-strategy Nash equilibria in \finegame\ is in fact in pure strategies (and is thus given by \cref{producer-fine-nash-char}/\cref{producer-fine-nash-char-special}).
\end{proposition}

If $\mu$ is atomless, then
in the Nash equilibrium defined in \cref{producer-fine-nash-char,producer-fine-nash-char-special}, almost all (i.e.\ except for maybe an amount of measure zero) of the $\nicefrac{1}{n}$ of consumers (as measured by $\mu$) with numerically smallest types consume from producer~$0$, whose chosen strategy is the numerically largest one that accommodates almost all of this $\nicefrac{1}{n}$; almost all of the next~$\nicefrac{1}{n}$ of consumers consume from producer $1$, whose chosen strategy is the numerically largest one that accommodates almost all of this $\nicefrac{1}{n}$, and so forth. Essentially, the market is split between the various producers based on consumer types, and each producer chooses the numerically largest strategy that accommodates almost all of its slice of the market, seemingly making no attempt to attract any other consumers. We conclude the static analysis of \finegame\ by formalizing these results.

\begin{theorem}[Market Split]\label{producer-fine-market-allocation}
Let $t_0\le\cdots\le t_{n-1} \in \timeset$ s.t.\ $\bar{t}$
constitutes a Nash equilibrium in \finegame,
and let $s$ be a mixed-consumption Nash equilibrium in the induced consumer game $(\mu;\bar{t})$.
If $\mu$ is atomless, then for every $j\in\prods$,\ \ $s_j(d)=1$ for almost all (w.r.t.\ $\mu$) consumer types
$d\in\timeset$ s.t.\ $\mu\bigl([0,d)\bigr)\in\bigl(\frac{j}{n}\cdot\mu(\timeset),\frac{j+1}{n}\cdot\mu(\timeset)\bigr)$.
\end{theorem}

\begin{remark}[Prima Facie Market Allocation]\label{producer-fine-strategy-by-allocation}
By \cref{producer-fine-nash-char}, if $\mu$ is atomless, then the $j$\tth\ sorted Nash-equilibrium strategy, $t_j$, is the numerically largest strategy (i.e.\ worst QoS)  acceptable by almost all consumer types $d\in\timeset$ s.t.\ $\mu\bigl([0,d)\bigr)\in\bigl(\frac{j}{n}\cdot\mu(\timeset),\frac{j+1}{n}\cdot\mu(\timeset)\bigr)$.
\end{remark}

\subsubsection{Dynamics}\label{producers-fine-dynamics}

We define weakly-/$\delta$-better/best-response dynamics in \finegame\ as in \cref{dynamics}, only with best responses defined w.r.t.\ fine preferences. In particular, the definition of improvement by at least $\delta$ remains
unchanged (i.e.\ it is defined solely w.r.t.\ the load). The analogue of the last part of \cref{coarse-better-vs-best} is therefore:

\begin{remark}\label{finebetter-vs-best}
In \finegame, a weakly-better-response dynamic is a best-response dynamic iff it is a $\delta$-better-response dynamic for some (equivalently, for all) $\delta>\mu(\timeset)$.
\end{remark}

We start by noting that best responses always exist --- an observation that for general $\mu$ is considerably less trivial w.r.t.\ fine preferences than w.r.t.\ coarse ones.
\begin{proposition}[A Unique Best Response Always Exists]\label{fine-unique-best-response}
Let $j\in\prods$ and let $\bar{t}_{-j}\in\timeset^{\prods\setminus\{j\}}$. A unique best response (by $j$) to $\bar{t}_{-j}$ exists in \finegame.
\end{proposition}
We give two proofs for \cref{fine-unique-best-response}: the first ---
quite-concise, and the second, while requiring more involved arguments, is constructive in the sense that in
contrast to the first, it presents the
best response in the form $\Max\{t\in\timeset\mid\mu\bigl([0,t)\bigr)\le m\}$, for $m$ that can be \emph{explicitly} calculated.

As in \finegame\ no producer is ever indifferent between two strategies, all weakly-better-response dynamics in this game are lazy;
therefore, if such a dynamic reaches a Nash equilibrium, it remains constant from that point on.
We also note that \cref{coarse-bad-response} is also an \lcnamecref{coarse-bad-response} of a lazy best-response dynamic in \finegame\ that never reaches a Nash equilibrium. Moreover,
as best responses in \finegame\ are unique, and as the strategies in each Nash equilibrium are distinct if $\mu$ is atomless, we have that no simultaneous best-response dynamic starting from a strategy profile with two or more identical strategies ever reaches a Nash equilibrium. It should be noted that this is not a boundary phenomenon; for example, if $\mu\bigl([0,t^0_j)\bigr)>0$
for all $j\in\prods$, then the best responses of all producers are identical (see \cref{must-jump-to-zero-fine} in \cref{producers-fine-proofs}). Many more such examples may be constructed. We now show that these
phenomena are all avoided by sequential dynamics.

\begin{corollary}\label{fine-response-equilibrium}
\cref{coarse-response-equilibrium,coarse-best-response-fast} hold also regarding reaching a Nash equilibrium w.r.t.\ \coarsegame\
by dynamics in the game \finegame.
\end{corollary}

\begin{theorem}[Sequential $\delta$-Better-Response Dynamics Converge from Coarse-Preferences Nash Equilibrium]\label{fine-sequential}
If $(P_i)_{i=0}^{\infty}$ is sequential from some point, then for every $\delta>0$,
at most one round after a $\delta$-better-response dynamic in \finegame\ reaches a Nash equilibrium w.r.t.\ \coarsegame\ after that point, it reaches a Nash equilibrium w.r.t.\ \finegame,
and remains constant from that point onward.
\end{theorem}

We hence formally conclude the results captured informally in \cref{intro-producer-dynamics}:

\begin{corollary}[Sequential $\delta$-Better-Response Dynamics Converge]\label{fine-sequential-cor}
For every $\delta>0$, every sequential $\delta$-better-response dynamic in \finegame\ reaches a Nash equilibrium in a finite number of steps,
and remains constant from that point onward.
\end{corollary}

\begin{corollary}[Sequential Best-Response Dynamics Converge Fast]\label{fine-best-response-fast-cor}
Every sequential or lazy best-response dynamic in \finegame\ reaches a Nash equilibrium in at most $n$ rounds,
and never visits a nonequilibrium after that point.
\end{corollary}

We conjecture than an even-tighter bound on convergence time than in \cref{fine-response-equilibrium,fine-best-response-fast-cor}
is attainable.
\cref{fine-sequential-cor,fine-best-response-fast-cor} show that the \emph{prima facie} market allocation (collusive split of the market) among the various producers shown in \cref{producer-fine-market-allocation,producer-fine-strategy-by-allocation} to be exhibited in every Nash equilibrium in \finegame\ arises as the unique possible outcome,
not as a result of anticompetitive practices, but rather as a result of noncooperative dynamics, each producer only looking to myopically maximize its preferences at every step;
as the best response to any strategy profile is unique, no signalling or any other collusive or cooperative ``trick'' whatsoever is used in order to reach and maintain this market split.

\section{Heterogeneous Products}\label{heterogeneous}

We have so far (in \cref{consumers,producers}) assumed that each consumer wishes to consume from a producer with least load. More generally, however, as
in the introduction, we may imagine that
some ISPs have different total bandwidth than others, while some other ISPs may purchase more total bandwidth as their subscriber pool grows. In such a scenario,
in order to surf with greatest speed, each consumer would no longer like to consume from a producer with least $\ell_j$ (i.e.\ with as few subscribers as possible),
but would rather consume from a producer with least~$f_j(\ell_j)$, where $f_j$ is an increasing continuous function for every $j\in\prods$, possibly differing between producers (e.g.\ $\nicefrac{\ell_j}{b_j}$, where $b_j$ is the total bandwidth of ISP $j$).
The results of \cref{consumers,producers} lend to generalization also to such a scenario via similar methods, with only quantitative rather than qualitative changes (the results regarding $\delta$-better-response dynamics require also that the functions $f_j$ be Lipschitz);
notably, the unique market-share division in both fine- and coarse-preferences Nash equilibria among producers is generally no longer of $\nicefrac{1}{n}$ of the market to each of the producers. E.g.\ \cref{producer-coarse-dominant,producer-coarse-nash-loads,producer-coarse-nash-char} thus become:

\begin{theorem}[Heterogeneous Products --- Coarse Preferences]\label{heterogeneous-coarse}
There exist amounts $\tilde{\ell}_0,\tilde{\ell}_1\ldots,\tilde{\ell}_{n-1}\in \bigl[0,\mu(\timeset)\bigr]$ (for homogeneous products, $\tilde{\ell}_j=\frac{\mu(\timeset)}{n}$ for
every $j\in\prods$), s.t.\ all of the following hold.
\begin{parts}
\item\thmitemtitle{Dominant Strategies}
Each dominant strategy in \coarsegame\ (the characterization of such strategies is unchanged from that given in the first part of \cref{producer-coarse-dominant}), when played by a producer $j\in\prods$,
guarantees a load of at least $\tilde{\ell}_j$ on this producer.
\item\thmitemtitle{Nash Equilibrium Loads}
A pure-strategy profile $\bar{t} \in \timeset^{\prods}$ constitutes a Nash equilibrium in \coarsegame\ iff $\ell_j(\bar{t})=\tilde{\ell}_j$ for every $j\in\prods$.
\item\thmitemtitle{Nash Equilibrium Characterization}
Let $\bar{t}$ be a pure-strategy profile and let $\pi\in\prods!$ be a permutation s.t.\ $t_{\pi(0)}\le t_{\pi(1)} \le\cdots\le t_{\pi(n-1)}$.
$\bar{t}$ constitutes a Nash equilibrium in \coarsegame\ iff $\mu\bigl([0,t_{\pi(j)})\bigr)\le\sum_{k=0}^{j-1}\tilde{\ell}_{\pi(k)}$ for every $j\in\prods$.
\end{parts}
\end{theorem}

Consequently, \cref{producer-fine-nash-char,producer-fine-nash-char-special} become:
\begin{theorem}[Heterogeneous Products --- Fine Preferences]\label{heterogeneous-fine}\leavevmode
\begin{parts}
\item\thmitemtitle{$\exists!$ Nash Equilibrium, and it is Super-Strong}
Let $\pi\in\prods!$ be a permutation s.t.\ there do not exist $j<k\in\prods$ s.t.\ $\tilde{\ell}_{\pi(j)}=0$ while $\tilde{\ell}_{\pi(k)}\ne0$.
A unique pure-strategy Nash equilibrium s.t.\ $t_{\pi(0)}\le\cdots\le t_{\pi(n-1)}$ exists in \finegame. The strategies of this equilibrium
are given by $t_{\pi(j)}\eqdef\Max\bigl\{t \in \timeset \mid \mu\bigl([0,t)\bigr)\le\sum_{k=0}^{j-1}\tilde{\ell}_{\pi(k)}\bigr\}$ for every $j\in\prods$. The load on each producer $j\in\prods$ in this Nash equilibrium is $\tilde{\ell}_j$. Furthermore, this equilibrium is super-strong. No other Nash equilibria exist in \finegame. 
\item\thmitemtitle{Nash Equilibrium Characterization --- Special Case}
If the CDF of $\mu$ is continuous (i.e.\ $\mu$ is atomless) and strictly increasing, then for every $j\in\prods$, in the Nash equilibrium corresponding to a permutation $\pi\in\prods!$
with the above properties,
$t_{\pi(j)}$ is the unique strategy satisfying $\mu\bigl([0,t_{\pi(j)})\bigr)=\sum_{k=0}^{j-1}\tilde{\ell}_{\pi(k)}$.
\end{parts}
\end{theorem}

The remainder of the results of \cref{consumers,producers}, including those regarding dynamics, readily generalize to this scenario as well.
So, we once again have that in a Nash equilibrium in \finegame, the market is split between producers based on consumer types; if $\mu$ is atomless and
$t_{\pi(0)}\le \cdots\le t_{\pi(n-1)}$, then almost all of the $\tilde{\ell}_{\pi(0)}$ consumers with numerically smallest types consume from producer $\pi(0)$ (who chooses the largest strategy
acceptable by almost all of them, seemingly making no attempt to attract any other consumers), almost all of the next $\tilde{\ell}_{\pi(1)}$ consumers consume from producer $\pi(1)$ (who chooses the largest strategy acceptable
by almost all of them, seemingly making no attempt to attract any others),
and so forth.
See \cref{vessels-odd-shapes} for an illustration regarding the adaptation of the results from \cref{consumers} to this generalized model, and the calculation of
$\tilde{\ell}_0,\ldots,\tilde{\ell}_{n-1}$.\footnote{See \citeN{hydraulic-selection} for a formalization, as part of (as mentioned above) a significant, highly nontrivial, generalization of our treatment of only the consumer game (without the producer game) to arbitrary resource-selection games (in which the resources available to a player may be any subset of $\prods$ and not merely a ``QoS-prefix'' of $\prods$ to which the above construction is inherently tailored) and beyond.}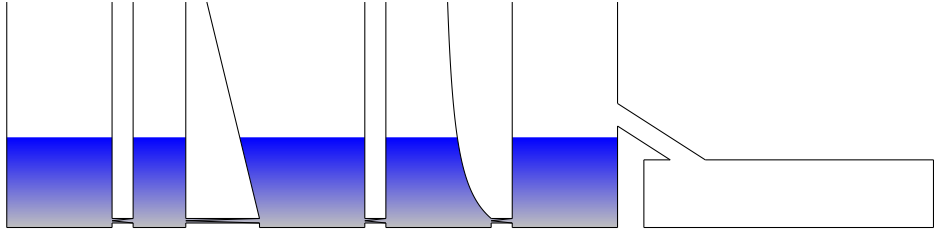
\begin{figure}%
\centering%
\begin{tikzpicture}[yscale=0.6,xscale=1.4]
\shade[top color=blue,bottom color=gray!50]
(0,0) -- (0,2) -- (1,2) -- (1,0.2) --
(1.2,0.2) -- (1.2,2) -- ({1.2+.5},2) -- ({1.2+.5},0.2) --
(2.4,0.2) -- ({2.4-.5/4.8*1.8},2) -- ({2.4+1},2) -- ({2.4+1},0.2) --
(3.6,0.2) -- (3.6,2) -- plot[variable=\t,samples=100,domain=2.8:1] ({3.6+.5+.5/\t,0.2+\t-1}) -- ({3.6+1},0.2) --
(4.8,0.2) -- (4.8,2) -- ({4.8+1},2) -- ({4.8+1},0) --
(4.8,0) -- (4.8,0.1) --
({3.6+1},0.1) -- ({3.6+1},0) -- (3.6,0) -- (3.6,0.1) --
({2.4+1},0.1) -- ({2.4+1},0) -- (2.4,0) -- (2.4,0.1) --
({1.2+.5},0.1) -- ({1.2+.5},0) -- (1.2,0) -- (1.2,0.1) --
(1,0.1) -- (1,0) -- cycle;
\draw[shift={(0*1.2,0)}] (0,0) -- (0,5);
\draw[shift={(0*1.2,0)}] (0,0) -- (1,0);
\draw[shift={(0*1.2,0)}] (1,0) -- (1,0.1);
\draw[shift={(0*1.2,0)}] (1,0.2) -- (1,5);
\draw[shift={(0*1.2,0)}] (1,0.1) -- (1.2,0.1);
\draw[shift={(0*1.2,0)}] (1.2,0.1) -- (1,0.15);
\draw[shift={(0*1.2,0)}] (1,0.15) -- (1.2,0.2);
\draw[shift={(0*1.2,0)}] (1.2,0.2) -- (1,0.2);
\draw[shift={(1*1.2,0)}] (0,0) -- (0,0.1);
\draw[shift={(1*1.2,0)}] (0,0.2) -- (0,5);
\draw[shift={(1*1.2,0)}] (0,0) -- (.5,0);
\draw[shift={(1*1.2,0)}] (.5,0) -- (.5,0.1);
\draw[shift={(1*1.2,0)}] (.5,0.2) -- (.5,5);
\draw[shift={(1*1.2,0)}] (.5,0.1) -- (1.2,0.1);
\draw[shift={(1*1.2,0)}] (1.2,0.1) -- (.5,0.15);
\draw[shift={(1*1.2,0)}] (.5,0.15) -- (1.2,0.2);
\draw[shift={(1*1.2,0)}] (1.2,0.2) -- (.5,0.2);
\draw[shift={(2*1.2,0)}] (0,0) -- (0,0.1);
\draw[shift={(2*1.2,0)}] (0,0.2) -- (-.5,5);
\draw[shift={(2*1.2,0)}] (0,0) -- (1,0);
\draw[shift={(2*1.2,0)}] (1,0) -- (1,0.1);
\draw[shift={(2*1.2,0)}] (1,0.2) -- (1,5);
\draw[shift={(2*1.2,0)}] (1,0.1) -- (1.2,0.1);
\draw[shift={(2*1.2,0)}] (1.2,0.1) -- (1,0.15);
\draw[shift={(2*1.2,0)}] (1,0.15) -- (1.2,0.2);
\draw[shift={(2*1.2,0)}] (1.2,0.2) -- (1,0.2);
\draw[shift={(3*1.2,0)}] (0,0) -- (0,0.1);
\draw[shift={(3*1.2,0)}] (0,0.2) -- (0,5);
\draw[shift={(3*1.2,0)}] (0,0) -- (1,0);
\draw[shift={(3*1.2,0)}] (1,0) -- (1,0.1);
\draw[shift={(3*1.2,0)}] plot[variable=\t,samples=100,domain=1:5.8] ({.5+.5/\t,0.2+\t-1});
\draw[shift={(3*1.2,0)}] (1,0.1) -- (1.2,0.1);
\draw[shift={(3*1.2,0)}] (1.2,0.1) -- (1,0.15);
\draw[shift={(3*1.2,0)}] (1,0.15) -- (1.2,0.2);
\draw[shift={(3*1.2,0)}] (1.2,0.2) -- (1,0.2);
\draw[shift={(4*1.2,0)}] (0,0) -- (0,0.1);
\draw[shift={(4*1.2,0)}] (0,0.2) -- (0,5);
\draw[shift={(4*1.2,0)}] (0,0) -- (1,0);
\draw[shift={(4*1.2,0)}] (1,0) -- (1,2.25);
\draw[shift={(4*1.2,0)}] (1,2.25) -- (1.5,1.5);
\draw[shift={(4*1.2,0)}] (1.5,1.5) -- (1.25,1.5);
\draw[shift={(4*1.2,0)}] (1.25,1.5) -- (1.25,0);
\draw[shift={(4*1.2,0)}] (1.25,0) -- (4,0);
\draw[shift={(4*1.2,0)}] (4,0) -- (4,1.5);
\draw[shift={(4*1.2,0)}] (4,1.5) -- ({1+2.5/3},1.5);
\draw[shift={(4*1.2,0)}] ({1+2.5/3},1.5) -- (1,2.75);
\draw[shift={(4*1.2,0)}] (1,2.75) -- (1,5);
\end{tikzpicture}%
\caption{%
A system of $5$ one-way communicating vessels, corresponding to $5$ heterogeneous ISPs (see the introduction) with the following characteristics, from left to right (i.e.\ from lowest latency/best QoS to highest latency/worst QoS):
A ``normal'' ISP, an ISP with half the total bandwidth of a ``normal'' one, an ISP whose total bandwidth somewhat increases with its number of subscribers, an ISP whose total bandwidth
somewhat decreases with its number of subscribers, and a ``normal''  ISP who buys additional bandwidth if needed, so that the bandwidth for a single subscriber never
drops below some threshold. (After the surface of the liquid in the fifth vessel reaches the tube connecting this vessel to the container on its right,
which we consider as part of the fifth vessel, any additional liquid poured into this vessel
accumulates in the container on the right; assume that this container is large enough so as to never fill up.)
We emphasize that the technical modifications to \cref{add-water} to accommodate any collection of increasing continuous functions $(f_j)_{j\in\prods}$
are straightforward and do not require defining any shapes for any vessels --- this is done purely to convey intuition. (We require that the functions be strictly increasing for simplicity,
however these results still hold if one of them is merely nondecreasing, e.g.\ as in the scenario depicted in the \lcnamecref{vessels-odd-shapes}; however, if more than one of these functions is not strictly increasing, e.g.\ if a sixth vessel identical to the fifth one is added in this \lcnamecref{vessels-odd-shapes}, then \cref{indifference-producers} may no longer hold.)
The producer-equilibrium loads $\tilde{\ell}_0,\ldots,\tilde{\ell}_{n-1}$ can be found by pouring
the entire $\mu(\timeset)$ of liquid into the rightmost vessel (i.e.\ computing the loads when each producer~$j$'s strategy is the $\tilde{\ell}_j$-guaranteeing strategy $0\in\timeset$), or, equivalently, by simply removing the
one-way valves (i.e.\ permitting liquid flow in both directions) and pouring $\mu(\timeset)$ liquid into the system
(observe that either way, if all vessels are of the same shape, then we indeed obtain $\tilde{\ell}_j=\frac{\mu(\timeset)}{n}$ for every $j\in\prods$, as in \cref{producers});
a similar ``two-way'' calculation among subsets of vessels
generalizes \cref{compute-ell} to this scenario.
}%
\label{vessels-odd-shapes}%
\end{figure}%

\section{Multiple Products and The Formation of Main Street}\label{main-street}

We conclude the body of this paper with an aesthetically appealing corollary,
obtained by extending our model to allow for multiple good types. (While we restrict ourselves in this \lcnamecref{main-street} to the case of two good types, the results below readily generalize also to the case of more than two good types.)
Consider the following alternative (non-ISP) interpretation of our model.

\begin{example}[Wine Market; QoS=Centrality of Location]\label{intro-wine}
Consider the downtown area of the fictional city of Metropolis, the wine capital of the world.
At its heart lies Metropolis Central Station. Every morning, shoppers (consumers) from throughout the Metropolis metropolitan area (and beyond) disembark the train at Metropolis Central,
at the vicinity of which many wine shops (producers) are located, and go about their wine-shopping errands. Each shopper is interested in purchasing a single bottle of wine, and is willing to walk at most $d$ minutes (a shopper-dependant real value) in each direction in order to get it.
All other wine characteristics being the same, each shopper would like the bottle of wine that she buys to be as exclusive as possible, i.e.\ she prefers to get her wine at the shop that sells the fewest bottles of wine throughout the day (so that it can be considered a ``boutique wine''), as long as it is no more than $d$ minutes away from Metropolis Central, of course.
As some wines may be known to be of superior types, are more extravagantly packaged, or have some other attractive quality, shoppers may be willing to compromise on ``exclusivity'' in favour of superior quality. Therefore, each shopper would like to minimize a wine-seller-dependent increasing function of the wine's circulation, e.g.\ shoppers may wish to maximize the quotient of quality and circulation.

Obviously, each wine seller would like to locate her store in a way that would maximize its sales volume. That being said, as real-estate prices rise the closer (in walking time) a shop is to Metropolis Central (we think of sales as indicative of daily income, and of real-estate cost as a one-time expense), each wine seller would like to place her store the
farthest possible from the station, as long as this does not hurt sales.
\end{example}

Our results from the previous \lcnamecrefs{heterogeneous} imply that in the scenario described in \cref{intro-wine}, the unique possible noncooperative outcome is once again for the market to be split between the various wine sellers based on the shoppers' types, i.e.\ each
shopper shopping at the store closest to Metropolis Central has a smaller walking-time limit than any of those shopping at the store second-closest to Metropolis Central, each of whom
in turn having a smaller walking-time limit than all of those shopping at the third-closest store, etc., and each wine seller chooses the farthest location accessible by the entirety of its
slice of the market, seemingly making no attempt to attract any other shoppers.
While this characterizes the distance of each wine shop from Metropolis Central, the direction from Metropolis Central to each such shop can be arbitrary. Not for long, though.

Suppose now that merchants from the nearby town of Smallville, the extra-extra-extra-virgin-olive-oil capital of the world, wishing to widen the visibility of their product, have started moving their stores to Downtown Metropolis as well. Now that Metropolis has become both the wine- and the extra-extra-extra-virgin-olive-oil capital of the world, each shopper arriving at Metropolis Central would like to purchase not only a bottle of wine, but also a bottle of olive oil. 
Nonetheless, the walking-time limit of each shopper does not change --- each shopper is still willing to walk at most $2d$ minutes in order to obtain both products. (This indeed introduces no change, as each shopper was previously willing to walk at most $d$ minutes \emph{in each direction}.)
As with wine, each shopper prefers to minimize a seller-dependent function of the circulation of the type of olive oil that she purchases, as long as her walking-time constraint is met.
(One may again consider e.g.\ the case in which one would like to maximize the quotient of quality to circulation, optimizing some form of tradeoff between quality and ``boutiqueness''.)
Olive-oil merchants have preferences similar to those of wine sellers.

Formally, we have $n_1\in\mathbb{N}$ producers of the first good (e.g.\ wine) and $n_2\in\mathbb{N}$
producers of the second good (e.g.\ olive oil). The strategy of each producer is a point on the plane; a pure-consumption strategy of a consumer with type $d\in\timeset$ is
a pair $(j,k)\in\prodso\times\prodst$, denoting consumption of the first good from producer $j$ of this good, and of the second good --- from producer $k$ of that good; each consumer would like to minimize
\mbox{$f_j^1(\ell_j)+f_k^2(\ell_k)$}\footnote{This sum may be replaced with any increasing continuous function of $f_j^1,f_k^2$,
e.g.\ their weighted average.} (e.g.\ the sum of the quotients of the quality and circulation for each good), subject to the constraint the that circumference of the triangle,
whose vertices are the origin (Metropolis Central Station) and the locations (strategies) of producer $j$ of good $1$ and of producer $k$ of good $2$, does not exceed $2d$ (the density of consumer types,
as given by $\mu$, remains unchanged). Each producer would like to first and foremost maximize its number of consumers, and only as a tie-breaker, maximize the norm of its strategy (i.e.\ its distance
from the origin).

Under these conditions, roughly speaking, each producer would like to be located so that visiting it would never be too much of a detour on the way from the origin
to a producer of the other good.
Indeed, we now show that the unique stablest outcome, in a precise sense, is for all shops to be placed \textbf{on the same ray} originating at Metropolis Central (with the distance of each store from Metropolis Central set as before, as if its good type were the only one on the market). 
(See \cref{main-street-proofs} for a proof, as well as a discussion regarding the necessity of the conditions below.)

\begin{theorem}[The Unique Super-Strong Equilibrium is a Main Street originating from the Origin]\label{main-street-super-strong}
Let $\tilde{\ell}^1_0,\ldots,\tilde{\ell}^1_{n_1-1}$ be the producer-equilibrium loads when only the first good is on the market (i.e.\ as defined in \cref{heterogeneous} when the only producers are the $n_1$ producers of good $1$) and let $\tilde{\ell}^2_0,\ldots,\tilde{\ell}^2_{n_2-1}$ be the producer-equilibrium loads when only the second good is on the market.
If no nonempty proper subset of the former loads and no nonempty proper subset of the latter loads have the same sum, and if $\tilde{\ell}^g_j>0$ for all $g$ and $j$, then a producer strategy profile is a super-strong equilibrium iff
the strategies of all producers of both products are \emph{on the same ray} from the origin, with distances from the origin as in \cref{heterogeneous-fine} (when computed separately for each good).
\end{theorem}

While most readers are likely to consider the formation of a main street as a fairly natural phenomenon due to its abundance in many cities,
some readers may find it somewhat less natural for this main street, as deduced in \cref{main-street-super-strong}, to originate from the city centre (e.g.\ Metropolis Central Station), rather than having the city centre in its middle.
Such readers may compare this with the structure of many old European towns, at the heart of which lies the old stone-cobbled main street, on one end of which (as opposed to at the middle of which) lies the main town church.

\section{Discussion}\label{discussion}
This paper shows, under quite general setting, that the appearance of collusion need not imply the actuality.
While we believe a main strength of our model to lie in its theoretic generality and aesthetics (its novel combination of congestion and location games, its clean results despite complex nontrivial multistage game analysis, and its (surprising) qualitative lesson), the question of the applicability of our model to a real-life market is a valid one.
Although our work is motivated mainly by internet monetization, and while we believe that its predictions will be confirmed with time, it is hard to validate its predictions on today's home internet market for several reasons, the main of which being that in many countries, many customers are not yet educated enough regarding latency, which leads ISPs to differentiate themselves from their competitors using other traits. In this \lcnamecref{discussion}, we offer real-world evidence supporting the applicability of our model to the food market in Israel.

The vast majority of groceries sold in Israel are Kosher. In fact, a nonnegligible part of the Jewish population in Israel, and in particular ultra-orthodox Jews, are only willing to buy food which is not only Kosher, but even more strictly monitored and restricted; we henceforth refer to such food as \emph{extra-monitored}. As extra-monitoring can be certified only by a handful of third-party monitors, manufacturing the same food product from the same ingredients costs
more when it is to be labelled as extra-monitored than when it is to be labelled as (``regular'') Kosher. Due to heavy lobbying on behalf of ultra-orthodox groups, though, even though producing an extra-monitored version of the same product costs more than producing a Kosher version of that product, both versions are sold by retailers for identical prices. (This holds in particular for products whose prices are regulated; no retailer would ever charge extra, beyond the regulated price, as compensation for extra-monitoring.) This property of the prices, together with the fact that a considerable amount of the population in Israel is primarily concerned with the monitoring level of their groceries, makes the food market in Israel~\textbf{fit squarely} in our model, with retailers as providers, shoppers for
``a week's worth of groceries'' as consumers, and the monitoring level of groceries as their QoS (there are in fact quite a few monitoring levels). Indeed, each shopper has a minimum required level of monitoring, beyond which she or he is indifferent (as it is physically the same product, at the exact same price), and it is quite reasonable that shoppers in a certain neighbourhood would therefore choose the least-crowded grocery store in the neighbourhood (no one likes to wait in line\ldots) out of those stores that meet their minimum required level of monitoring. From the retailers' point of view, they would like to first and foremost maximize their number of shoppers, and as long as this number is not hurt, minimize the monitoring level of each of their products (the price difference for monitoring, while nonzero, is negligible relative to capturing more market share).

Our results from the previous \lcnamecrefs{heterogeneous} predict that under these conditions, the unique possible noncooperative outcome is for the market to be split between the various retailers based on the shoppers' minimum required monitoring level, i.e.\ each shopper shopping at the ``minimum monitoring'' retailer has a lower minimum required monitoring level than any of those shopping at the ``second-lowest monitoring'' retailer, each of whom in turn having a lower minimum required monitoring level than all of those shopping at the ``third-lowest monitoring'' retailer, etc., and each retailer chooses the minimum monitoring level that satisfies the entirety of its slice of the market, seemingly making no attempt to attract any other (stricter) shoppers. Indeed, in neighbourhoods with both nonnegligible ultra-orthodox population and nonnegligible orthodox populations, one notices that grocery stores label themselves by their specific monitoring level, which is applied to all products in the store.
Our early study, to be further explored in a companion work, suggests
that the number of stores of each monitoring level (when weighted by store size) roughly corresponds to the demand for this monitoring level as a minimum required level.

\begin{acks}
The first author was supported in part by ISF grant 230/10, by
the European Research Council under the European Community's Seventh Framework
Programme (FP7/2007-2013) / ERC grant agreement no.\ [249159], and by an Adams Fellowship of the Israeli Academy of Sciences and Humanities.
We would like to thank Sergiu Hart, the first's author Ph.D.\ advisor, for pointing out an intuitive resemblance between our communicating-vessels analogy and that of \citeN{Kaminsky}.
\end{acks}

\bibliography{noncoop-market-alloc}

\widowpenalty150
\clubpenalty150

\appendix
\section*{APPENDIX}

\section{Proofs and Auxiliary Results}\label{proofs}

\subsection{Proofs and Auxiliary Results for Section~\refintitle{consumers}}\label{consumers-proofs}

We commence with a few lemmas used in the proofs of \cref{consumer-pure-nash-exists,consumer-symmetric-nash-exists}.

\begin{lemma}[Load is Nonincreasing in Strategy]\label{decreasing-load}
Under the definitions of \cref{consumers}, if $t_0\le t_1\le\cdots\le t_{n-1}$, then for every mixed-consumption Nash equilibrium $s$ in the $n$-producers consumer game $(\mu;\bar{t})$, we have
$\load{0}\ge\load{1}\ge\cdots\ge\load{n-1}$.
\end{lemma}

\begin{proof}
Let $j \in \{0,\ldots,n-2\}$. If $\load{j+1}=0$, then $\load{j} \ge \load{j+1}$. Assume, therefore, that $\load{j+1}>0$.
Hence, there exists
$d \ge t_{j+1}$ s.t.\ $s_{j+1}(d)>0$. By definition of $s$, and as $d \ge t_{j+1} \ge t_j$, we thus have $\load{j+1} \le \load{j}$,
as required.
\end{proof}

As mentioned above, the construction in the proof of \cref{consumer-pure-nash-exists,consumer-symmetric-nash-exists} is illustrated in \cref{vessels}. In the context of that \lcnamecref{vessels}, the
following \lcnamecref{add-water} can be thought of as answering the following question: if the amount of liquid in each vessel $j\in\prods$ is $\ell_j$, by how much would the liquid in each vessel rise if we pour an additional amount $m$ of liquid into vessel $n-1$? (The rise in the amount of liquid in vessel $j$ is given by $p_j$.)

\begin{lemma}\label{add-water}
Let $\ell_0 \ge \ell_1 \ge \cdots \ge \ell_{n-1}$ be a finite nonincreasing sequence in $\Rge$. For every $m \in \Rge$, there exists $p \in [0,m]^{\prods}$,
which may be computed in $O(n)$ time, s.t.\ all of the following hold.
\begin{parts}
\item\label{add-water-sum}
$\sum_{j=0}^{n-1}p_j=m$.
\item\label{add-water-decreasing}
$\ell_0+p_0 \ge \ell_1+p_1 \ge \cdots \ge \ell_{n-1}+p_{n-1}$.
\item\label{add-water-min-height}
$\ell_k+p_k = \min_{j\in\prods}\{\ell_j+p_j\}$ for every $k\in\prods$ s.t.\ $p_k>0$.
\end{parts}
\end{lemma}

\begin{proof}
We iteratively define a sequence
$p^n \le p^{n-1} \le\ldots\le p^0 \in [0,m]^{\prods}$ s.t.\ the following hold for every $i \in \{0,\ldots,n\}$.
\begin{properties}
\item\label{consumer-symmetric-nash-exists-zero-less}
$p^i_j=0$ for every $j < i$.
\item\label{consumer-symmetric-nash-exists-small-sum}
$\sum_{j=0}^{n-1}p^i_j\le m$, with equality when $i=0$.
\item\label{consumer-symmetric-nash-exists-equal-height}
There exists $h_i \in \Rge$  s.t.\ all of the following hold.
\begin{itemize}
\item
If $\sum_{j=0}^{n-1}p^{i+1}_j < m$, then $\ell_j+ p^i_j=h_i$ for every $j \ge i$, 
\item
If $\sum_{j=0}^{n-1}p^{i}_j < m$, then $\ell_{i-1}+ p^i_{i-1}=h_i$ as well.
\item
$\ell_j \ge h_i$ for every $j < i$.
\end{itemize}
\end{properties}
In the setting of \cref{vessels}, $p^n$ describes the rise of liquid before we begin pouring the additional amount $m$, while for every $i\in\prods$, $p^i$ describes the rise of liquid
at the last instant during the pouring process, in which no water has risen except in vessels $i,i+1,\ldots,n-1$. (This can be either the final rise in liquid if $i=0$ or if the final rise does
not involve a change in the amount of liquid in vessels $j<i$, or alternatively the rise in liquid just before the liquid in vessel $i-1$ begins to rise.)

For the base case, we define $p^n\equiv0$, and all parts trivially hold (with $h_n\eqdef\ell_{n-1}$). For the construction step, let $i\in\{0,\ldots,n-1\}$ and assume that $p^{i+1}$ has been defined. Let $c\eqdef\sum_{j =0}^{n-1}p^{i+1}_j$. By \cref{consumer-symmetric-nash-exists-small-sum} for $i+1$, $c \le m$.
If $i=0$, then we define $r\eqdef m-c\ge0$; otherwise, we define $r\eqdef\min\bigl\{(n-i)\cdot(\ell_{i-1} - h_{i+1}), m-c\bigr\}$, and by \cref{consumer-symmetric-nash-exists-equal-height}, $r \ge 0$ in this case as well. We define $p^i_j \eqdef 0$ for every $j < i$ (and so \cref{consumer-symmetric-nash-exists-zero-less} holds for $i$), and $p^i_j \eqdef p^{i+1}_j + \frac{r}{n-i} \ge p^{i+1}_j$ for every $j \ge i$. \cref{consumer-symmetric-nash-exists-small-sum} holds for $i$ as $\sum_{j=0}^{n-1}p^i_j = \sum_{j=0}^{n-1}p^{i+1}_j + r = c + r \le m$, with equality when $i=0$. Finally, we show that \cref{consumer-symmetric-nash-exists-equal-height} holds for $i$, with $h_i\eqdef h_{i+1}+ \frac{r}{n-i}$.
If $\sum_{j=0}^{n-1}p^{i+1}_j < m$, then as $p^{i+2}  \le p^{i+1}$, we have that ($i+1=n$ or) $\sum_{j=0}^{n-1}p^{i+2}_j < m$ as well. Therefore, by \cref{consumer-symmetric-nash-exists-equal-height} for $i+1$, we have for every $j \ge i$ that $\ell_j+p^i_j=\ell_j+p^{i+1}_j + \frac{r}{n-i} = h_{i+1} + \frac{r}{n-i} = h_i$.
If $\sum_{j=0}^{n-1}p^{i}_j < m$, then $r < m-c$ and so by definition, $r=(n-i)\cdot(\ell_{i-1} - h_{i+1})$. Therefore, $h_i=h_{i+1}+\frac{r}{n-i}=\ell_{i-1}=\ell_{i-1}+p^i_{i-1}$.
Finally, for every $j < i$, by $\ell$ nonincreasing and by definition of $r$ we have $\ell_j \ge \ell_{i-1} \ge h_{i+1}+\frac{r}{n-i}=h_i$, and the proof of the construction is complete.

Let $\tilde{\imath} \in \{0,\ldots,n\}$ be largest s.t.\ $\sum_{j=0}^{n-1}p^{\tilde{\imath}}_j = m$; $\tilde{\imath}$ is well defined by \cref{consumer-symmetric-nash-exists-small-sum} for $i=0$.
We now show that $p\eqdef p^{\tilde{\imath}}$ meets the conditions of the \lcnamecref{add-water}. By definition, $\sum_{j=0}^{n-1}p_j=m$.

Let $j\in\prods\setminus\{n-1\}$. If $j < \tilde{\imath}-1$, then by \cref{consumer-symmetric-nash-exists-zero-less} for $i=\tilde{\imath}$, we have $\ell_j+p_j=\ell_j \ge \ell_{j+1} = \ell_{j+1}+p_{j+1}$.
If $j=\tilde{\imath}-1$, then by \cref{consumer-symmetric-nash-exists-zero-less,consumer-symmetric-nash-exists-equal-height} for $i=\tilde{\imath}$ and by definition of $\tilde{\imath}$, we have $\ell_j+p_j=\ell_j=\ell_{\tilde{\imath}-1} \ge h_{\tilde{\imath}} = \ell_{\tilde{\imath}}+p_{\tilde{\imath}}=\ell_{j+1}+p_{j+1}$.
Otherwise, i.e.\ if $j > \tilde{\imath}-1$, then by \cref{consumer-symmetric-nash-exists-equal-height} for $i=\tilde{\imath}$ and by definition of $\tilde{\imath}$, we have $\ell_j+p_j = h_{\tilde{\imath}} = \ell_{j+1}+p_{j+1}$.

We conclude that $\min_{j\in\prods}\{\ell_j+p_j\}=\ell_{n-1}+p_{n-1}$.
For every $k\in\prods$ s.t.\ $p_k > 0$, by \cref{consumer-symmetric-nash-exists-zero-less} for $i=\tilde{\imath}$ we have $k \ge \tilde{\imath}$. Therefore, by \cref{consumer-symmetric-nash-exists-equal-height} for $i=\tilde{\imath}$ and by definition of $\tilde{\imath}$, we have $\ell_k+p_k=h_{\tilde{\imath}}=\ell_{n-1}+p_{n-1}=\min_{j\in\prods}\{\ell_j+p_j\}$, as required.

Finally, although it may seem in first glance that $O(n^2)$ time may be required to compute $p$, we note that the sequence $(h_i)_{i=0}^n$ can be computed in
$O(n)$ time, that from this sequence $\tilde{\imath}$ can be deduced in $O(n)$ time as the largest s.t.\ $h_{\tilde{\imath}}=h_0$, and that from both, $p$ can be calculated
in $O(n)$ time: $p_j=0$ for $j<\tilde{\imath}$ by \cref{consumer-symmetric-nash-exists-zero-less} for $i=\tilde{\imath}$, while $p_j=h_{\tilde{\imath}}-\ell_j$ for $j\ge\tilde{\imath}$ by \cref{consumer-symmetric-nash-exists-equal-height} for $i=\tilde{\imath}$.
\end{proof}

We now prove \cref{consumer-symmetric-nash-exists}, and then deduce \cref{consumer-pure-nash-exists} therefrom. Alternatively, \cref{consumer-pure-nash-exists} can also be
proven directly from \cref{add-water}, similarly to the proof of \cref{consumer-symmetric-nash-exists}.

\begin{definition}
For a finite measure $\mu$ on $\timeset$ and
a measurable set $E\subseteq\timeset$, we denote by $\mu|_{\cap E}$ the finite
measure on $\timeset$ given by $\mu|_{\cap E}(A)\eqdef\mu(A\cap E)$.
\end{definition}

\begin{proof}[of \cref{consumer-symmetric-nash-exists}]
We prove the claim by induction on $n$. (Recall that the construction underlying this proof is illustrated by \cref{vessels}; also recall the explanation preceding the statement of
\cref{add-water} regarding the meaning of that \lcnamecref{add-water} in the context of that \lcnamecref{vessels}.)

Base ($n=0$): In this case, $S=\{\noconsumption\}$, and so $s\equiv\mathds{1}_{\{\noconsumption\}}$ is a Nash equilibrium as required.

Step ($n>0$): Assume w.l.o.g.\ that $t_0\le t_1\le \cdots \le t_{n-1}$. By the induction hypothesis, there exists a symmetric mixed-consumption Nash equilibrium $s'$ in the $(n-1)$-producers consumer game $(\mu|_{\cap[0,t_{n-1})};t_0,\ldots,t_{n-2})$.
If $\mu\bigl([t_{n-1},1]\bigr)=0$, then we define a mixed-consumption profile $s$ in $(\mu;\bar{t})$ s.t.\ $s|_{[0,t_{n-1})}=s'_{[0,t_{n-1})}$, and
$s|_{[t_{n-1},1]}\equiv\mathds{1}_{\{n-1\}}$. As $s'$ is symmetric in $(\mu|_{\cap[0,t_{n-1})};t_0,\ldots,t_{n-2})$, so is $s$ in $(\mu;\bar{t})$.
As $\load{j}=\loadt{j}$ for every $j\in\prodsm$, by $s'$ being a Nash equilibrium in $(\mu|_{\cap[0,t_{n-1})};t_0,\ldots,t_{n-2})$ we have that no player of any type
$d \in [0,t_{n-1})$ has any incentive to unilaterally deviate from $s$. As $\mu\bigl([t_{n-1},1]\bigr)=0$, we have $\load{n-1}=0$, and so
players of types $d \in [t_{n-1},1]$ have no incentive to deviate from $s$ either. Therefore, $s$ is a symmetric Nash equilibrium
as required, and the proof for this case is complete. Assume therefore, henceforth, that $\mu\bigl([t_{n-1},1]\bigr)>0$.

Recall that $\loadt{0},\loadt{1},\ldots,\loadt{n-2}$ are the loads on producers in $s'$, and by slight abuse of notation, define $\loadt{n-1}\eqdef0\le\loadt{n-2}$; by \cref{decreasing-load},
$\loadt{0}\ge\loadt{1}\ge\cdots\ge\loadt{n-2}\ge\loadt{n-1}$. Let $p$ be as in \cref{add-water} for $\ell_j=\loadt{j}$ for every $j\in\prods$, and for $m=\mu\bigl([t_{n-1},1]\bigr)>0$.
We define a mixed-consumption profile $s$ in $(\mu;\bar{t})$ s.t.\ $s|_{[0,t_{n-1})}=s'|_{[0,t_{n-1})}$, and
$s|_{[t_{n-1},1]}\equiv \frac{p}{\mu([t_{n-1},1])}$ (by \crefpart{add-water}{sum}, indeed $\frac{p}{\mu([t_{n-1},1])} \in \Delta^{S_d}$ for all $d \ge t_{n-1}$). Once again, as $s'$ is symmetric, so is $s$. It remains to show that $s$ is indeed a Nash equilibrium as required.

By definition of $s'$ and of $s$, we have that $\load{j}=\loadt{j}+p_j$ for every $j\in\prods$.
Let $d \in [0,t_{n-1})$. As $\loadt{0},\ldots,\loadt{n-2}$ and $\load{0},\ldots,\load{n-2}$ are both nonincreasing (the former by \cref{decreasing-load}, and the latter --- by \crefpart{add-water}{decreasing}), and as $S_d$ is the same in both $(\mu|_{\cap[0,t_{n-1})};t_0,\ldots,t_{n-2})$ and $(\mu;\bar{t})$, we have that as no player of type $d$ has any incentive to unilaterally deviate from $s'$ in the former, neither does it from $s$ in the latter. Let now $d \in [t_{n-1},1]$. For every $k\in\supp\bigl(s(d)\bigr)$, we have by definition $p_k > 0$, and so, by \crefpart{add-water}{min-height}, $\load{k}=\min_{j\in\prods}\load{j}$, and the proof is complete.

The complexity claim follows as each inductive step requires $O(n)$ time --- the time required to calculate $p$, by \cref{add-water}.
\end{proof}

\begin{corollary}\label{high-untouched}
Let $h\in\prods$, let $s'$ be a mixed-consumption Nash equilibrium in the $h$-producers consumer game $(\mu|_{\cap[0,t_h)};t_0,\ldots,t_{h-1})$, and let $s$ be the mixed-consumption Nash equilibrium in the $n$-producers consumer game $(\mu;\bar{t})$ constructed iteratively
from $s'$ as in the proof of \cref{consumer-symmetric-nash-exists}. For every $0 \le j < h$, we have $\load{j}\ge\loadt{j}$, with equality
if $\load{h-1}>\load{h}$.
\end{corollary}

\begin{proof}
By following the construction in the proof of \cref{consumer-symmetric-nash-exists}, and by \crefpart{add-water}{min-height}.
\end{proof}

\cref{consumer-pure-nash-exists} follows from \cref{consumer-symmetric-nash-exists} and from the following \lcnamecref{mixed-to-pure}.

\begin{lemma}[\cref{consumer-symmetric-nash-exists} $\Rightarrow$ \cref{consumer-pure-nash-exists}]\label{mixed-to-pure}
If a mixed-consumption Nash equilibrium exists in the $n$-producers consumer game $(\mu;\bar{t})$, and if $\mu$ is atomless, then a pure-consumption Nash equilibrium exists in this game
as well.
\end{lemma}

\begin{proof}
Assume w.l.o.g.\ that $t_0\le t_1\le \cdots \le t_{n-1}$.
Let $s$ be a mixed-consumption Nash equilibrium in the game $(\mu;\bar{t})$. For every $i \in \{0,\ldots,n-2\}$, set $C^i\eqdef[t_i,t_{i+1})$, and set $C^{\noconsumption}\eqdef[0,t_0)$ and $C^{n-1}\eqdef[t_{n-1},1]$; note that $S_d=\{\noconsumption\}$ for all $d \in C^{\noconsumption}$, and that $S_d=\{0,\ldots,i\}\cup\{\noconsumption\}$ for all $d \in C^i$, for every $i\in\prods$. For every $i,j\in S$, define $p^i_j\eqdef\int_{C^i}s_j\,d\mu$; note that if $p^i_j>0$, then $s_j(d)>0$ for some $d \in C^i$.
Let $i \in S$.
We first consider the case in which either $\mu(C^i)>0$ or $C^i=\emptyset$. In this case, as $\mu$ is atomless, there exists a partition of $C^i$ into $n$ pairwise-disjoint measurable sets $(C^i_j)_{j \in S}$, s.t.\ $\mu(C^i_j)=p^i_j$ for all $j \in S$, and s.t.\ $C^i_j=\emptyset$ whenever $p^i_j=0$. Otherwise, i.e.\ if $\mu(C^i)=0$ yet $C^i\ne\emptyset$, then let $k\in S$ s.t.\ $s_k(d)>0$ for some $d \in C^i$, and define $C^i_k\eqdef C_i$, and $C^i_j\eqdef\emptyset$ for every $j \in S\setminus\{k\}$. Note that in this case we also have that $(C^i_j)_{j\in S}$ is a partition of $C^i$ and $\mu(C^i_j)=0=\int_{C^i}s_j\,d\mu=p^i_j$ for all $j \in S$.

We define a measurable function $s':\timeset\rightarrow S$ by $s'|_{\cup_{i \in S}C^i_j}\equiv j$ for every $j \in S$. For every $j \in S$, we note that $\loadt{j}=\mu(\cup_{i \in S}C^i_j)=\sum_{i\in S}p^i_j=\sum_{i \in S}\int_{C^i}s_j\,d\mu=\int_{\timeset}s_j\,d\mu=\load{j}$.

We conclude by showing that $s'$ is indeed a pure-strategy profile, and moreover --- a Nash equilibrium. Let $d \in \timeset$; by definition there exists $i \in S$ s.t.\ $d \in C^i_{s'(d)}\subseteq C^i$. As $C^i_{s'(d)}\ne\emptyset$, by definition of $C^i_{s'(d)}$ we have that $s_{s'(d)}(d')>0$ for some $d' \in C^i$, and so $s'(d) \in S_{d'}$. As by definition of $C^i$ we have $S_d=S_{d'}$, we obtain $s'(d) \in S_d$, and so $s'$ is a pure-strategy profile. Furthermore, as $s_{s'(d)}(d')>0$, we obtain $\loadt{s'(d)}=\load{s'(d)}=\min_{j\in S_{d'}}\load{j}=\min_{j\in S_d}\load{j}=\min_{j\in S_d}\loadt{j}$, and the proof is complete.
\end{proof}

\begin{proof}[of \cref{indifference-producers}]
Assume w.l.o.g.\ that $t_0\le t_1\le \cdots \le t_{n-1}$.
Let $s,s'$ be two mixed-consumption Nash equilibria in the game $(\mu;\bar{t})$.
By definition of Nash equilibrium, we have $s'_{\noconsumption}=\mathds{1}_{[0,t_0)}=s_{\noconsumption}$, and so $\sum_{j=0}^{n-1}\loadt{j}=\mu\bigl([t_0,1]\bigr)=\sum_{j=0}^{n-1}\load{j}$.
Assume for contradiction that there exists $j\in\prods$ s.t.\ $\loadt{j}\ne\load{j}$; let $j$ be minimal with this property, and assume w.l.o.g.\ that $\loadt{j} > \load{j}$.

Let $j \le k < n$ be maximal s.t.\ $\loadt{k} = \loadt{j}$. By \cref{decreasing-load}, for every $j \le i \le k$, we have $\loadt{i}=\loadt{j}>\load{j}\ge\load{i}$. Therefore, and as $\loadt{i}=\load{i}$ for every $0\le i<j$, we have $\sum_{j=0}^{k}\loadt{j}>\sum_{j=0}^{k}\load{j}$. We thus obtain both that $k<n-1$, and that $\sum_{j=0}^{k}\loadt{j}>\sum_{j=0}^{k}\load{j}\ge \mu\bigl([t_0,t_{k+1})\bigr)$. Therefore, $\sum_{j=k+1}^{n-1}\loadt{j} < \mu\bigl([t_{k+1},1]\bigr)$, and hence there exists $d \in [t_{k+1},1]$ s.t.\ $s'_i(d)>0$ for some $0\le i\le k$. As by \cref{decreasing-load} we have $\loadt{i}\ge\loadt{k}>\loadt{k+1}$, we conclude that $s'$ is not a Nash equilibrium --- a contradiction.
\end{proof}

\begin{proof}[of \cref{indifference-consumers}]
Assume w.l.o.g.\ that $t_0\le t_1\le \cdots \le t_{n-1}$.
Let $s$ be a mixed-consumption Nash equilibrium in the game $(\mu;\bar{t})$, let $d\in\timeset$ and let $k\in\supp\bigl(s(d)\bigr)$.
If $d<t_j$ for all $j\in\prods$, then $k=\noconsumption$ and so $\load{k}=\mu\bigl([0,t_0))$. Otherwise, $k\ne\noconsumption$ and so $d\ge t_0$; let
$i\in\prods$ be largest s.t.\ $t_i\le d$. By definition of $s$ and by \cref{decreasing-load}, we have $\load{k}=\min\{\load{j}\mid t_j \le d\}=\load{i}$. Either way (and by \cref{indifference-producers} when $k\ne\noconsumption$), $\load{k}$ does not depend on the choice of $s$, as required.
\end{proof}

From \cref{compute-ell}, we obtain the following recursive identity for $\ell_k(\bar{t})$.

\begin{corollary}\label{compute-ell-no-algo}
If $t_0\le t_1\le\cdots\le t_{n-1}$, then defining $t_n\eqdef2$, we have
\[\ell_k(t_0,\ldots,t_{n-1})=\Max_{k<j\le n}\frac{\mu\bigl([t_0,t_j)\bigr)-\sum_{i=0}^{k-1}\ell_i(\bar{t})}{j-h}=\Max_{k<j\le n}\frac{\mu\bigl([0,t_j)\bigr)-\sum_{i\in\prodsk\cup\{\noconsumption\}}\ell_i(\bar{t})}{j-h}\]
(where by slight abuse of notation, $\mu$ is treated as a measure on $[0,2]$ with support $\timeset$)
for every $k\in\prods$.
\end{corollary}

\begin{proof}
A direct corollary of \cref{compute-ell}, by considering two cases: in the first, either $k=0$ or $\ell_k(\bar{t})<\ell_{k-1}(\bar{t})$ (and so the given value $k$ 
is the value of the variable $k$ in some iteration of \cref{compute-ell}); in the second, $k>0$ and $\ell_k(\bar{t})=\ell_{k-1}(\bar{t})$ (and so \cref{compute-ell}
calculates both $\ell_k(\bar{t})$ and $\ell_{k-1}(\bar{t})$ in the same iteration of the while loop, and therefore they are identical; it is straightforward to verify that the expression in the statement evaluates to the same value for both $k-1$ and $k$ in this case).
\end{proof}

\subsection{Analysis of \texorpdfstring{$\ell$}{the Load Function l}}\label{ell-analysis}

Before moving on to prove the results presented in \cref{producers}, we now formalize three analytic properties of the function $\ell$ (defined in \cref{ell}),
which we later utilize when proving the results of \cref{producers}. The first property is that the load on a producer
does not decrease if the producer raises the offered quality of service (i.e.\ lowers its strategy).

\begin{lemma}[$\ell_j$ is Nonincreasing in $t_j$]\label{ell-decreasing}
For every $j\in\prods$ and for every $\bar{t}\in\timeset^{\prods}$ and $t'_j\in\timeset$, if $t_j<t'_j$, then $\ell_j(\bar{t}_{-j},t'_j)\le\ell_j(\bar{t})$.
\end{lemma}

\begin{proof}
Let $\bar{t}\in\timeset^{\prods}$, $k\in\prods$ and $t'_k\in(t_k,1]$.
Assume w.l.o.g.\ that $t_0\le t_1\le\cdots\le t_{n-1} \in \timeset$. If $k\ne n-1$, then it is enough to consider the case $t_k<t'_k\le t_{k+1}$.
Let $s$ be a mixed-consumption Nash equilibrium in the induced consumer game $(\mu;\bar{t})$.
For every $j \in \prods\setminus\{k\}$, define $t'_j\eqdef t_j$.
Let $s'$ be a mixed-consumption Nash equilibrium in $(\mu;\bar{t}')$, and assume for contradiction that $\loadt{k} > \load{k}$.
Let $i\in\prods$ be maximal s.t.\ $\loadt{i}=\loadt{k}$; by definition, $i\ge k$.

We claim that $\loadt{j}\ge\load{j}$ for every $0 \le j \le i$. Let $0 \le h \le i$ be minimal s.t.\ $\load{h} \le \loadt{k}$  (such $h$ exists, and $h \le k$, as $\load{k} < \loadt{k}$); we will show that $\loadt{j}\ge\load{j}$ separately for every $0 \le j < h$ (if $h>0$) and for every $h \le j \le i$. For every $h \le j \le i$, by \cref{decreasing-load}, by definition of $i$ and by definition of $h$, we have $\loadt{j} \ge \loadt{i} = \loadt{k} \ge \load{h} \ge \load{j}$, as required.
We move on to show that $\loadt{j}\ge\load{j}$ for every $0 \le j < h$; assume that $h>0$ (otherwise, there is nothing to show).
Let $\tilde{\ell}_0,\ldots,\tilde{\ell}_{h-1}$ be the loads on producers $0,\ldots,h-1$ in a Nash equilibrium in the game $(\mu|_{\cap[0,t_h)};t_0,\ldots,t_{h-1})$;
similarly, let $\tilde{\ell}'_0,\ldots,\tilde{\ell}'_{h-1}$ be the loads on producers $0,\ldots,h-1$ in a Nash equilibrium in the game $(\mu|_{\cap[0,t'_h)};t'_0,\ldots,t'_{h-1})$.
As $h>0$, by definition of $h$ we have $\load{h-1} > \load{h}$;
therefore, by \cref{high-untouched} and \cref{indifference-producers}, we have that $\load{j}=\tilde{\ell}_j$ for every $0\le j<h$.
By \cref{high-untouched} and \cref{indifference-producers}, we obtain also that $\loadt{j}\ge\tilde{\ell}'_j$ for every $0\le j<h$.
As $k\ge h$, we have that $t'_j=t_j$ for every $0\le j<h$ and that $t'_{h-1}\ge t_{h-1}$; therefore,
$\tilde{\ell}'_j\ge\tilde{\ell}_j$ for every $0\le j<h$. (This follows by
by tracing the construction in the proof of \cref{consumer-symmetric-nash-exists}, as all inductive steps but the last are identical, while the last, examining \cref{add-water},
increases each load by no less when computing $\tilde{\ell}'$ than when computing $\tilde{\ell}$; in the context of \cref{vessels}, pouring a additional nonnegative amount of liquid
into the rightmost vessel does not cause the liquid level in any vessel to fall. Alternatively, this can also be seen by
tracing \cref{compute-ell}, as each iteration when computing $\tilde{\ell}'$ either computes the same load
values for the producers as the corresponding iteration when computing $\tilde{\ell}$, or is the last, thus computing loads that are not lower than those computed for $\tilde{\ell}$.)
Combining all of these, we have $\loadt{j}\ge\tilde{\ell}'_j\ge\tilde{\ell}_j=\load{j}$ for every $0\le j<h$, as required.

We conclude that $\sum_{j=0}^i\loadt{j}>\sum_{j=0}^i\load{j}$, as $\loadt{j}\ge\load{j}$ for every $0\le j\le i$, with a strict inequality for $j=k$.
If $i=n-1$, then $\sum_{j=0}^{n-1}\loadt{j}>\sum_{j=0}^{n-1}\load{j}=\mu\bigl([t_0,1]\bigr)\ge\mu\bigl([t'_0,1]\bigr)$ --- a contradiction; assume, therefore, that $i<n-1$.
Hence, $\sum_{j=0}^{i}\loadt{j}>\sum_{j=0}^{i}\load{j}\ge\mu\bigl([t_0,t_{i+1})\bigr)\ge\mu\bigl([t'_0,t_{i+1})\bigr)$. Therefore, there exists $d \ge t'_{i+1}$ s.t.\
$s'_j(d)>0$ for some $0\le j\le i$, but by \cref{decreasing-load} and by definition of $i$ we notice that $\loadt{j}\ge\loadt{i}>\loadt{i+1}$, so $s'$ is not a Nash equilibrium in $(\mu;\bar{t}')$ --- a contradiction as well.

We note that an alternative proof may also be given via \cref{compute-ell,compute-ell-no-algo}.
\end{proof}

The second property is that the load on producer $j$ cannot increase as a result of other producers moving closer to $j$'s quality of service.

\begin{lemma}[$\ell_j$ is Weakly Quasiconvex in $t_k$]\label{ell-single-valley}
For every $j\in\prods\cup\{\noconsumption\}$ and $j\ne k\in\prods$ and for every $\bar{t}\in\timeset^{\prods}$ and $t'_k\in\timeset$, if $j\ne\noconsumption$ and either $t_k<t'_k\le t_j$ or $t_j\le t'_k<t_k$, or if $j=\noconsumption$ and $t'_k<t_k$, then $\ell_j(\bar{t}_{-k},t'_k)\le\ell_j(\bar{t})$.
\end{lemma}

\begin{proof}
Assume w.l.o.g.\ that $t_0\le t_1\le\cdots\le t_{n-1}$. If $k\ne n-1$, then it is enough to consider the case $t_k<t'_k\le t_{k+1}$.
Let $s$ and $s'$ be mixed-consumption Nash equilibria in $(\mu;\bar{t})$ and $(\mu;\bar{t}_{-k},t'_k)$, respectively.

Assume for contradiction that $\loadt{i}>\load{i}$ for some $k<i<n$, and let $i$ be minimal with this property.
Therefore (and by \cref{ell-decreasing} if $i=k+1$), $\loadt{i-1}\le\load{i-1}$. By \cref{decreasing-load}, we obtain $\load{i-1}\ge\loadt{i-1}\ge\loadt{i}>\load{i}$.
Therefore, $s_j|_{[t_i,1]}\equiv0$ for every $0\le j<i$, and so $\sum_{j=i}^{n-1}\loadt{j}\le\mu\bigl([t_i,1])=\sum_{j=i}^{n-1}\load{j}$. Hence, and as $\loadt{i}>\load{i}$, there
exists $i<h<n$ s.t.\ $\loadt{h}<\load{h}$ --- let $h$ be minimal with this property. Therefore, $\loadt{h-1}\ge\load{h-1}$, and by \cref{decreasing-load},
we obtain $\loadt{h-1}\ge\load{h-1}\ge\load{h}>\loadt{h}$. By definition of $h$, we have that $\loadt{j}\ge\load{j}$ for every $i\le j<h$, with a strict inequality
for $j=i$ by definition of $i$, and so $\sum_{j=i}^{h-1}\loadt{j}>\sum_{j=i}^{h-1}\load{j}\ge\mu\bigl([t_i,t_h)\bigr)$, with the last inequality
since $s_j|_{[t_i,1]}\equiv0$ for every $0\le j<i$. Therefore, there exists $d \ge t_h$ s.t.\
$s'_j(d)>0$ for some $i\le j<h$, but by \cref{decreasing-load}, $\loadt{j}\ge\loadt{h-1}>\loadt{h}$, so $s'$ is not a Nash equilibrium in $(\mu;\bar{t}_{-k},t'_k)$ --- a contradiction.

Assume now for contradiction that $\loadt{i}<\load{i}$ for some $0\le i<k$, and let $i$ be minimal with this property.
As $k>0$, we have $\load{\noconsumption}=\mu\bigl([0,t_0)\bigr)=\loadt{\noconsumption}$. Therefore, 
$\sum_{j=0}^{n-1}\loadt{j}=\sum_{j=0}^{n-1}\load{j}$ and by definition of $i$ there exists $h\in\prods$ s.t.\ $\loadt{h}>\load{h}$
--- let $h$ be minimal with this property.
By \cref{ell-decreasing} and by the first part of this proof, $h<k$. We now consider two cases: $h<i$ and $i<h$. We start with
the case $h<i$.
In this case, by definition of $i$ and by \cref{decreasing-load}, $\loadt{i-1}\ge\load{i-1}\ge\load{i}>\loadt{i}$, and so, by \cref{high-untouched} and as $i<k$,
$\loadt{h}=\ell_h(\mu|_{\cap[0,t_i)};t_0,\ldots,t_{i-1})\le\load{h}$ --- a contradiction.
Similarly, if $i<h$, then by definition of $h$ and by \cref{decreasing-load}, $\load{h-1}\ge\loadt{h-1}\ge\loadt{h}>\load{h}$, and so, by \cref{high-untouched} and as $h<k$,
$\load{i}=\ell_i(\mu|_{\cap[0,t_h)};t_0,\ldots,t_{h-1})\le\loadt{i}$ --- a contradiction as well.

We conclude by examining the effect on $\ell_{\noconsumption}$.
If $k\ne0$, then let $t'_0\eqdef t_0$. Regardless of the value of $k$, we have $t'_0\ge t_0$.
By definition of $s,s'$, we have $\load{\noconsumption}=\mu\bigl([0,t_0)\bigr)\le\mu\bigl([0,t'_0)\bigr)=\loadt{\noconsumption}$.

We note that an alternative proof may also be given via \cref{compute-ell,compute-ell-no-algo}.
\end{proof}

Finally, the last property is that small perturbations in the producers' strategies result in quantifiably small changes in the loads on the various producers.

\begin{lemma}[$\ell$ is Lipschitz in Each Coordinate with Lipschitz constant $1$]\label{ell-lipschitz}
For every $j,k\in\prods$ and for every $\bar{t}\in\timeset^{\prods}$ and $t'_k\in\timeset$, if $t_k<t'_k$, then
$\bigl|\ell_j(\bar{t}_{-k},t'_k)-\ell_j(\bar{t})\bigr|\le\mu\bigl([t_k,t'_k)\bigr)$.
\end{lemma}

\begin{proof}
Assume w.l.o.g.\ that $t_0\le t_1\le\cdots\le t_{n-1}$. If $k\ne n-1$, then it is enough to consider the case $t_k<t'_k\le t_{k+1}$.
Let $s$ and $s'$ be mixed-consumption Nash equilibria in $(\mu;\bar{t})$ and $(\mu;\bar{t}_{-k},t'_k)$, respectively.
Define $S_k\eqdef\prodsk\cup\{\noconsumption\}$. We start by showing that $\sum_{j\in S_k}\loadt{j}-\sum_{j\in\ S_k}\load{j}\le\mu\bigl([t_k,t'_k)\bigr)$.

If $k=0$, then this claim holds, as in this case $S_k=\{\noconsumption\}$, and by definition of $s$ and $s'$ we have $\loadt{\noconsumption}-\load{\noconsumption}=\mu\bigl([0,t'_k)\bigr)-\mu\bigl([0,t_k)\bigr)=\mu\bigl([t_k,t'_k)\bigr)$. Assume therefore that $k>0$ and
assume for contradiction that $\sum_{j\in S_k}\loadt{j}-\sum_{j\in S_k}\load{j}>\mu\bigl([t_k,t'_k)\bigr)$. As $k>0$, we have $\loadt{\noconsumption}=\mu\bigl([0,t_0)\bigr)=\load{\noconsumption}$, and so $\sum_{j=0}^{k-1}\loadt{j}-\sum_{j=0}^{k-1}\load{j}>\mu\bigl([t_k,t'_k)\bigr)$ as well.
Let $i\in\prods$ be maximal s.t.\ $\loadt{i}=\loadt{k-1}$.
By definition, $i \ge k-1$. For every $k\le j\le i$, by \cref{decreasing-load,ell-single-valley}, we have $\loadt{j}\ge\loadt{i}=\loadt{k-1}\ge\load{k-1}\ge\load{j}$.
Therefore, we have $\sum_{j=0}^i\loadt{j}-\sum_{j=0}^i\load{j}>\mu\bigl([t_k,t'_k)\bigr)$.
If $i=n-1$, then we obtain $\sum_{j=0}^{n-1}\loadt{j}>\sum_{j=0}^{n-1}\load{j}=\mu\bigl([t_0,1]\bigr)$ --- a contradiction; assume, therefore, that $i<n-1$.
If $i+1\ne k$, then let $t'_{i+1}\eqdef t_{i+1}$. Hence,
\begin{align*}
&\sum_{j=0}^i\loadt{j}>\sum_{j=0}^i\load{j}+\mu\bigl([t_k,t'_k)\bigr)\ge\mu\bigl([0,t_{i+1})\bigr)+\mu\bigl([t_k,t'_k)\bigr)=\\
&\qquad\qquad\qquad\qquad\qquad\qquad\qquad=\begin{cases}
\mu\bigl([0,t_k)\bigr)+\mu\bigl([t_k,t'_k)\bigr)=\mu\bigl([0,t'_{i+1})\bigr) & i+1=k \\
\mu\bigl([0,t'_{i+1})\bigr)+\mu\bigl([t_k,t'_k)\bigr)\ge\mu\bigl([0,t'_{i+1})\bigr) & i+1\ne k
\end{cases}.
\end{align*}
Thus, there exists $d\ge t'_{i+1}$ s.t.\ $s'_j(d)>0$ for some $0\le j\le i$, but by \cref{decreasing-load} and by definition of $i$ we notice that $\loadt{j}\ge\loadt{i}>\loadt{i+1}$, so $s'$ is not a Nash equilibrium in $(\mu;\bar{t}_{-k},t'_k)$ --- a contradiction as well.

As $\sum_{j\in S_k}\loadt{j}-\sum_{j\in\ S_k}\load{j}\le\mu\bigl([t_k,t'_k)\bigr)$ and as by \cref{ell-single-valley} $\loadt{j}\ge\load{j}$ for every $j\in S_k$,
we obtain $0\le\loadt{j}-\load{j}\le\mu\bigl([t_k,t'_k)\bigr)$ for every $j\in S_k$.
Similarly, As $\sum_{j\in S_k}\loadt{j}-\sum_{j\in\ S_k}\load{j}\le\mu\bigl([t_k,t'_k)\bigr)$
and as $\sum_{j\in S}\loadt{j}=\mu(\timeset)=\sum_{j\in S}\load{j}$, we have $\sum_{j=k}^{n-1}\load{j}-\sum_{j=k}^{n-1}\loadt{j}\le\mu\bigl([t_k,t'_k)\bigr)$,
and as by \cref{ell-decreasing,ell-single-valley} $\loadt{j}\le\load{j}$ for every $k\le j<n$, we obtain $0\le\load{j}-\loadt{j}\le\mu\bigl([t_k,t'_k)\bigr)$ for every $k\le j<n$, and the proof is complete.

We note that an alternative proof may also be given via \cref{compute-ell} / \cref{compute-ell-no-algo}.
\end{proof}

\subsection{Proofs and Auxiliary Results for Section~\refintitle{producers-coarse}}\label{producers-coarse-proofs}

\subsubsection{Proofs and Auxiliary Results for Section~\refintitle{producers-coarse-statics}}

\begin{lemma}\label{min-max-strategy-load}
Let $j\in\prods$ and $t_j\in\timeset$. For every $\bar{t}_{-j}\in\timeset^{\prods\setminus\{j\}}$,
we have $\ell_j(\bar{t})\ge\frac{\mu([t_j,1])}{n}$, which constitutes a tight bound.
\end{lemma}

\begin{proof}
Let $s$ be mixed-consumption Nash equilibrium in $(\mu;\bar{t})$.
If $\mu\bigl([t_j,1]\bigr)=0$, then $\load{j}=0$ and the claim trivially holds. Assume therefore that $\mu\bigl([t_j,1]\bigr)>0$.
Let $k\in\arg\Max_{i\in\prods}\int_{[t_j,1]}s_i\,d\mu$. By definition of $s$, $\int_{[t_j,1]}s_{\noconsumption}\,d\mu=0$, and so $\int_{[t_j,1]}s_k\,d\mu\ge\frac{\mu([t_j,1])}{n}$ by definition of $k$. Since $\int_{[t_j,1]}s_k\,d\mu\ge\frac{\mu([t_j,1])}{n}>0$, there exists $d \ge t_j$ s.t.\ $s_k(d)>0$. As $j \in S_d$, by definition of Nash equilibrium we have $\int_{[t_j,1]}s_j\,d\mu=\load{j} \ge \load{k} \ge \int_{[t_j,1]}s_k\,d\mu \ge \frac{\mu([t_j,1])}{n}$.

Alternatively, by \cref{ell-single-valley}, $\ell_j(\bar{t})$ is minimal given $t_j$ when $t_i=t_j$ for every $i\in\prods\setminus\{j\}$.
By \cref{indifference-producers}, by anonymity, and by definition of Nash equilibrium,
the load on each producer in this case is exactly $\frac{\mu([t_j,1])}{n}$.
\end{proof}

\begin{corollary}\label{zero-strategy-load}
Let $\bar{t} \in \timeset^{\prods}$. For every $j\in\prods$, if $t_j=0$, then $\ell_j(\bar{t})\ge\frac{\mu(\timeset)}{n}$.
\end{corollary}

\begin{proof}
A direct corollary of \cref{min-max-strategy-load}.
\end{proof}

\begin{definition}[Domination]
Let $t,t'$ be strategies in \coarsegame.
\begin{itemize}
\item
We say that $t$ \emph{weakly dominates} $t'$ if $t$ is a safe alternative to $t'$ and moreover, there exists some strategy profile for all but one of the producers,
s.t.\ playing $t$ gives the remaining producer strictly higher utility than playing $t'$.
\item
We say that $t$ \emph{strongly dominates} $t'$ if for every strategy profile for all but one of the producers,
playing $t$ gives the remaining producer strictly higher utility than playing~$t'$.
\end{itemize}
\end{definition}

\begin{lemma}[Domination]\label{producer-coarse-domination}
Let $t<t' \in \timeset$ be strategies in \coarsegame.
\begin{parts}
\item\label{producer-coarse-domination-safe}
$t$ is a safe alternative to $t'$.
\item\label{producer-coarse-domination-weak}
$t$ weakly dominates $t'$ iff $\mu\bigl([t,t')\bigr)>0$.
\item\label{producer-coarse-domination-strong}
$t$ strongly dominates $t'$ iff $\mu\bigl([t',1]\bigr)<\frac{\mu([t,1])}{n}$.
\end{parts}
\end{lemma}

\begin{proof}
\cref{producer-coarse-domination-safe} follows from \cref{ell-decreasing}.

We move on to proving \cref{producer-coarse-domination-weak}.
$\Rightarrow$:
Assume that $\mu\bigl([t,t')\bigr)=0$; by \cref{compute-ell} / \cref{compute-ell-no-algo}, $t$ and $t'$ are equivalent, and \emph{a fortiori} $t'$ is a safe alternative to $t$.
Nonetheless, we now also directly show that $t$ and $t'$ are equivalent.
Let $\bar{t}_{-0}\in\timeset^{\prods\setminus\{0\}}$, and let $s$ be a mixed-consumption Nash equilibrium in $(\mu;\bar{t}_{-0},t)$.
Let $s'$ be the mixed-consumption profile in $(\mu;\bar{t}_{-0},t')$ s.t.\ $s'|_{\timeset\setminus[t,t')}=s|_{\timeset\setminus[t,t')}$ and s.t.\ for every $d \in [t,t')$,
if $S_d=\{\noconsumption\}$ w.r.t.\ $(\mu;\bar{t}_{-0},t')$ then $s'(d)=\mathds{1}_{\{\noconsumption\}}$, and otherwise $s'(d)=\mathds{1}_{\{j\}}$ for
some $j\in\arg\Max_{j\in S_d}t_j$. As $s=s'$ almost
everywhere w.r.t.\ $\mu$, we have that $\loadt{j}=\load{j}$ for all $j\in\prods$. By definition of $s$, therefore no type $d \in \timeset\setminus[t,t')$ has any incentive to deviate from $s'$ in the latter game. By \cref{decreasing-load} (for $s$) and by definition of $s'$, neither does any type $d \in [t,t')$ have any incentive to deviate from $s'$ in the latter game.
Therefore, $s'$ is a Nash equilibrium in $(\mu;\bar{t}_{-0},t')$. As in particular $\loadt{0}=\load{0}$, the proof of this direction is complete.

$\Leftarrow$:
Assume that $\mu\bigl([t,t')\bigr)>0$; we will show that $t'$ is not a safe alternative to $t$.
Define $a\eqdef\mu\bigl([t,t')\bigr)>0$, $b\eqdef\mu\bigl([t',1)\bigr)$ and $c\eqdef\mu\bigl(\{1\}\bigr)$.
By \cref{compute-ell},
we have $\ell_0(t',1,1,\ldots,1) = \Max\{b,\frac{b+c}{n}\}<\Max\{a+b,\frac{a+b+c}{n}\}=\ell_0(t,1,1,\ldots,1)$, and the proof of this direction is complete as well.

We conclude by proving \cref{producer-coarse-domination-strong}.
$\Rightarrow$:
Assume that $\mu\bigl([t',1]\bigr)\ge\frac{\mu([t,1])}{n}$. Therefore, $\mu\bigl([t,t')\bigr)\le\frac{n-1}{n}\cdot\mu\bigl([t,1]\bigr)$, and hence $\frac{\mu([t,t'))}{n-1}\le\frac{\mu([t,1])}{n}$.
By \cref{indifference-producers}, by anonymity, and by definition of Nash equilibrium, the load on every producer, and in particular on producer $0$, in a Nash equilibrium in the $n$-producer game $(\mu;t,t,\ldots,t)$ is $\frac{\mu([t,1])}{n}$.
As $\Max\{\frac{\mu([t,t'))}{n-1},\frac{\mu([t,1])}{n}\}=\frac{\mu([t,1])}{n}$, by \cref{compute-ell} the load on every producer, and in particular on producer $0$, in a Nash equilibrium in the game $(\mu|_{\cap[0,t')};t',t,t,\ldots,t)$ is $\frac{\mu([t,1])}{n}$ as well, as required.

$\Leftarrow$:
Assume that $\mu\bigl([t',1]\bigr)<\frac{\mu([t,1])}{n}$.
Let $\bar{t}_{-0}\in\timeset^{\prods\setminus\{0\}}$.
By \cref{min-max-strategy-load} and by definition of legal strategies in the consumer game, we obtain $\ell_0(\bar{t}_{-0},t)\ge\frac{\mu([t,1])}{n}>\mu\bigl([t',1]\bigr)\ge\ell_0(\bar{t}_{-0},t')$.
\end{proof}

\begin{proof}[of \cref{producer-coarse-dominant}]
The first statement is a direct corollary of \cref{producer-coarse-domination}, and the second --- of \cref{zero-strategy-load}.
\end{proof}

\begin{lemma}\label{equal-loads-consumer-nash}
Let $\bar{t} \in \timeset^{\prods}$ and
let $s$ be a mixed-consumption profile in the consumer game $(\mu;\bar{t})$.
If $\load{j}=\frac{\mu(\timeset)}{n}$ for every $j\in\prods$, and if $s_{\noconsumption}(d)=0$ whenever $S_d\ne\{\noconsumption\}$, then $s$ constitutes a Nash equilibrium in this game.
\end{lemma}

\begin{proof}
Directly from definition of mixed-consumption Nash equilibrium, no consumer has any incentive to unilaterally deviate from $s$.
\end{proof}

\begin{proof}[of \cref{producer-coarse-nash-loads}]
The first part ($\Rightarrow$) follows directly from \cref{zero-strategy-load}.
For the second part ($\Leftarrow$), let 
$\bar{t}$ be a pure-strategy profile in \coarsegame\ s.t.\ $\ell_j(\bar{t})=\frac{\mu(\timeset)}{n}$ for
every $j\in\prods$. Assume w.l.o.g.\ that $t_0\le\cdots\le t_{n-1}$, and
let $s$ be a mixed-consumption Nash equilibrium in the induced consumer game $(\mu;\bar{t})$; therefore, $\load{j}=\ell_j(\bar{t})=\frac{\mu(\timeset)}{n}$ for every $j\in\prods$. Let $k\in\prods$ and let $t'_k \in \timeset$; we will show that producer $k$ has no incentive
to deviate to $t'_k$ from $t_k$. If $t_k < t'_k$, then this follows directly from \crefpart{producer-coarse-domination}{safe}. We therefore consider
the case in which $t'_k<t_k$.
Let $N\eqdef[t'_k,t_0)$ (if $t_0 \le t'_k$, then $N=\emptyset$). By definition of $s$, we have that $\mu(N) \le \mu\bigl([0,t_0)\bigr) = \load{\noconsumption} = 0$. We define a mixed-consumption profile $s'$ in the consumer game $(\mu;\bar{t}_{-k},t'_k)$ 
by $s'|_N\equiv\mathds{1}_{\{k\}}$ and $s'|_{\timeset\setminus N}=s|_{\timeset\setminus N}$ (if $N=\emptyset$, then $s'=s$).
(This indeed is a mixed-consumption profile since $t'_k<t_k$ and by definition of $N$.)
As $\mu(N)=0$, we have $\loadt{j}=\load{j}=\frac{\mu(\timeset)}{n}$ for every $j\in\prods$. Thus, by \cref{equal-loads-consumer-nash} and by definition of $s'$ via $N$, we conclude that $s$ is a Nash equilibrium in $(\mu;\bar{t}_{-k},t'_k)$, and so $k$ has no incentive to deviate to $t'_k$ from $t_k$ in this case either.
\end{proof}

\begin{corollary}[Least/Most Nash Equilibrium Load]\label{producer-coarse-nash-least-most-loads}
Let $t_0\le\cdots\le t_{n-1} \in \timeset$. The pure-strategy profile $\bar{t}$ constitutes a Nash equilibrium in \coarsegame\ iff
either of the following equivalent conditions hold.
\begin{itemize}
\item
$\ell_{n-1}(\bar{t})=\frac{\mu(\timeset)}{n}$.
\item
$\mu\bigl([0,t_0)\bigr)=0$ and $\ell_0(\bar{t})=\frac{\mu(\timeset)}{n}$.
\end{itemize}
\end{corollary}

\begin{proof}
A direct corollary of \cref{producer-coarse-nash-loads,decreasing-load}.
\end{proof}

\begin{proof}[of \cref{producer-coarse-nash-char}]
$\Rightarrow$:
Assume that $\bar{t}$ constitutes a pure-strategy Nash equilibrium in \coarsegame.
Let $s$ be a mixed-consumption Nash equilibrium in the induced consumer game $(\mu;\bar{t})$.
For every $j\in\prods$, by \cref{producer-coarse-nash-loads}, we obtain $\mu\bigl([t_j,1]\bigr)\ge\sum_{k=j}^{n-1}\load{k}=\frac{n-j}{n}\cdot\mu(\timeset)$, and so $\mu\bigl([0,t_j)\bigr) \le \frac{j}{n}\cdot\mu(\timeset)$.

$\Leftarrow$:
Assume that $\mu\bigl([0,t_j)\bigr) \le \frac{j}{n}\cdot\mu(\timeset)$ for every $j\in\prods$;
in particular, $\mu\bigl([0,t_0)\bigl)=0$.
Let $s$ be a mixed-consumption profile in the induced consumer game $(\mu;\bar{t})$.
By \cref{producer-coarse-nash-loads}, it is enough to show that $\load{j}=\frac{\mu(\timeset)}{n}$ for every $j\in\prods$.
Assume for contradiction that this is not the case. Therefore, as $\load{\noconsumption}=\mu\bigl([0,t_0)\bigl)=0$, there exists $k\in\prods$
s.t.\ $\load{k}>\frac{\mu(\timeset)}{n}$; let $k$ be maximal with this property. By \cref{decreasing-load}, we have $\sum_{j=0}^k\load{k}\ge(k+1)\cdot\load{k}>\frac{(k+1)\cdot\mu(\timeset)}{n}\ge\mu\bigl([0,t_{k+1})\bigr)$. Therefore, there exists $d \ge t_{k+1}$ s.t.\ $s'_j(d)>0$ for some $0\le j\le k$, but by definition of $k$ we notice that $\loadt{k}>\frac{\mu(\timeset)}{n}\ge\load{k+1}$, so $s'$ is not a Nash equilibrium --- a contradiction.

We note that an alternative proof of the second direction ($\Leftarrow$) may also be given via \cref{compute-ell} / \cref{compute-ell-no-algo}.
\end{proof}

The second direction ($\Leftarrow$) of \cref{producer-coarse-nash-char} can also be proven constructively. Such a proof is quite tedious in the general case, but simplifies greatly when $\mu$ is atomless. E.g.\ if $\mu=U(\timeset)$ the uniform measure, then whenever $t_0\le\cdots\le t_{n-1}$ meet the conditions of \cref{producer-coarse-nash-char}, then a Nash equilibrium can be formed by splitting the market as follows: every $d \in [0,\frac{1}{n})$ plays the pure strategy $0\in\prods$ (the conditions of \cref{producer-coarse-nash-char} guarantee that this is a legal strategy for all such $d$), every $d \in [\frac{1}{n},\frac{2}{n})$ plays the pure strategy $1\in\prods$ (once again, the conditions of \cref{producer-coarse-nash-char} guarantee that this is a legal strategy for all such $d$), and so on, until every $d \in [\frac{n-1}{n},1]$, playing the pure strategy $n-1\in\prods$. (We remark that if $t_j=\frac{j}{n}$ for every $j\in\prods$, then this is in fact the unique Nash equilibrium among consumers, up to modifications of measure zero. More about this split and Nash equilibrium --- in
\cref{producer-fine-market-allocation,producer-fine-strategy-by-allocation} in \cref{producers-fine-proofs}.)
The load on each producer in this case is precisely $\frac{1}{n}$, and by \cref{equal-loads-consumer-nash,producer-coarse-nash-loads}, the proof is complete.

\begin{proof}[of \cref{producer-coarse-super-strong}]
Assume w.l.o.g.\ that $t_0\le t_1\le\cdots\le t_{n-1}$ and
assume for contradiction that $P\subseteq\prods$ and $\bar{t}'\in\timeset^P$ as in the statement exist.
Let $s$ be a mixed-consumption Nash equilibrium in $(\mu;\bar{t}_{-P},\bar{t}')$.
As there exists $j\in P$ s.t.\ $\load{j}>\ell_j(\bar{t})=\frac{\mu(\timeset)}{n}$ (with the last equality by \cref{producer-coarse-nash-loads}), and as $\sum_{j=0}^{n-1}\load{j}\le\mu(\timeset)$, there thus exists a producer $k\in\prods$ s.t.\ $\load{k}<\frac{\mu(\timeset)}{n}$ --- let $\emptyset\ne K\subseteq\prods$ be the set of all such producers; by definition of $P$ and by \cref{producer-coarse-nash-loads}, we have $K \subseteq \prods\setminus P$.
By \cref{decreasing-load}, $K=\{n-1,n-2,\ldots,n-|K|\}$. Therefore, and by \cref{producer-coarse-nash-char},
$\sum_{k=n-|K|}^{n-1}\load{k}<\frac{|K|\cdot\mu(\timeset)}{n}\le\mu\bigl([t_{n-|K|},1]\bigr)$, and so there exists $d\in[t_{n-|K|},1]$ s.t.\ $s_j(d)>0$ for some $0\le j<n-|K|$, but by definition of $K$ we notice that $\loadt{j}\ge\frac{\mu(\timeset)}{n}>\loadt{n-|K|}$, and so (as $n-|K|\notin P$) we have that $s$ is not a Nash equilibrium in $(\mu;\bar{t}_{-P},\bar{t}')$ --- a contradiction.
\end{proof}

\begin{proof}[of \cref{producer-coarse-mixed,producer-coarse-mixed-nash-no-regret}]
\crefpart{producer-coarse-mixed}{dominant} follows directly from \cref{producer-coarse-dominant}.
We move on to proving \cref{producer-coarse-mixed-nash-no-regret} and \crefpart{producer-coarse-mixed}{nash-loads}.
Let $\bar{p}$ be a mixed-strategy Nash equilibrium in \coarsegame.
By \cref{zero-strategy-load}, $\expect{\ell_j(\bar{p})}\ge\frac{\mu(\timeset)}{n}$ for every $j\in\prods$; since $\sum_{j=0}^{n-1}\ell_j(\bar{p})\le\mu(\timeset)$, and by linearity of expectation, we obtain $\expect{\ell_j(\bar{p})}=\frac{\mu(\timeset)}{n}$ for every $j\in\prods$.
Let $j \in \prods$.
By \crefpart{producer-coarse-domination}{safe},
we have $\ell_j(\bar{p}_{-j},\mathds{1}_{\{0\}})\ge\ell_j(\bar{p})$. As $j$ has no incentive to deviate from $p_j$ to playing $0\in\timeset$, we thus
have that $\ell_j(\bar{p}_{-j},\mathds{1}_{\{0\}})=\ell_j(\bar{p})$ with probability $1$. By \cref{zero-strategy-load}, we have that $\ell_j(\bar{p}_{-j},\mathds{1}_{\{0\}})\ge\frac{\mu(\timeset)}{n}$, and so $\ell_j(\bar{p})\ge\frac{\mu(\timeset)}{n}$ with probability $1$. As $\expect{\ell_j(\bar{p})}=\frac{\mu(\timeset)}{n}$,
we have that  $\ell_j(\bar{p})=\frac{\mu(\timeset)}{n}$ with probability $1$.

Let now $\bar{p}$ be a mixed-strategy profile in \coarsegame, s.t.\ $\ell_j(\bar{p})=\frac{\mu(\timeset)}{n}$ with probability~$1$ for every $j \in \prods$. By \cref{producer-coarse-nash-loads},
the resulting realization is a pure-strategy Nash equilibrium with probability $1$, and so with probability $1$ no ex-post regret exists
and \emph{a fortiori} $\bar{p}$ is a Nash equilibrium.

We move on to proving \crefpart{producer-coarse-mixed}{nash-char}.
Let $\bar{p}$ be a mixed-strategy Nash equilibrium in \coarsegame.
We iteratively build a permutation $\pi\in\prods!$ s.t.\ $\mu\bigl([0,\Max\supp(p_{\pi(j)}))\bigr)\le\frac{j}{n}$ for every $j \in \prods$.
Let $k \in \prods$ and assume that $\pi(j)$ has been defined for all $0 \le j < k$. Define $U\eqdef\prods\setminus\{\pi(0),\ldots,\pi(k-1)\}$ and
let $t_{\min}$ be a random variable denoting the numerically smallest strategy realization of $U$, i.e.\
$t_{\min}\eqdef\min_{j\in U}p_j$.
By \crefpart{producer-coarse-mixed}{nash-loads}, with probability $1$ we have 
$\sum_{j\in U}\ell_j(\bar{p})=\frac{n-k}{n}\cdot\mu(\timeset)$; therefore, with probability $1$ we have $\mu\bigl([0,t_{\min})\bigr)\le\frac{k}{n}\cdot\mu(\timeset)$.
Hence, by independence, it is not possible that $\prob{\mu\bigr([0,p_j)\bigl)>\frac{k}{n}\cdot\mu(\timeset)}>0$ for every $j\in U$. Therefore, there exists
$\pi(k)\in U$ s.t.\ $\prob{\mu\bigr([0,p_{\pi(k)})\bigl)\le\frac{k}{n}\cdot\mu(\timeset)}=1$, and so $\mu\bigl([0,\Max\supp(p_{\pi(k)}))\bigr)\le\frac{k}{n}\cdot\mu(\timeset)$
and the construction is complete.

Let now $\bar{p}$ be a mixed-strategy profile in \coarsegame, s.t.\ $\mu\bigl([0,\Max\supp(p_{\pi(j)}))\bigr)\le\frac{j}{n}\cdot\mu(\timeset)$ for every $j \in \prods$. Therefore, $\mu\bigl([0,p_{\pi(j)})\bigr)\le\mbox{$\frac{j}{n}\cdot\mu(\timeset)$}$ for every $j \in \prods$ with probability $1$. By \cref{producer-coarse-nash-char},
the resulting realization is a pure-strategy Nash equilibrium with probability $1$, and so with probability $1$ no ex-post regret exists
and \emph{a fortiori} $\bar{p}$ is a Nash equilibrium.

We conclude by proving \crefpart{producer-coarse-mixed}{super-strong}.
Let $\bar{p}$ be a mixed-strategy Nash equilibrium in \coarsegame.
By \cref{producer-coarse-mixed-nash-no-regret,producer-coarse-super-strong}, a realization of $\bar{p}$ is with probability $1$ a~super-strong
equilibrium, and so \emph{a fortiori} $\bar{p}$ is a super-strong Nash equilibrium.
\end{proof}

\subsubsection{Proofs and Auxiliary Results for Section~\refintitle{producers-coarse-dynamics}}

\begin{lemma}\label{coarse-best-response-far-from-least-load}
Let $\bar{t}\in\timeset^{\prods}$ be a pure-strategy profile, let $h\in\arg\Max_{j\in\prods} t_j$ and let $k\in\prods$.
If $\bar{t}$ is not a Nash equilibrium in \coarsegame, but nonetheless $t_k$ is a best response to $\bar{t}_{-k}$, then both $\ell_h(\bar{t})<\ell_k(\bar{t})$ and $\mu\bigl([t_k,t_h)\bigr)\ge\frac{\mu(\timeset)}{n}$.
\end{lemma}

\begin{proof}
As $t_k$ is a best response to $\bar{t}_{-k}$ and by \cref{zero-strategy-load}, $\ell_k(\bar{t})\ge\frac{\mu(\timeset)}{n}$. As $\bar{t}$ is not a Nash equilibrium, by \cref{producer-coarse-nash-least-most-loads} we have
$\ell_h(\bar{t})<\frac{\mu(\timeset)}{n}$, and so $\ell_h(\bar{t})<\frac{\mu(\timeset)}{n}\le\ell_k(\bar{t})$.
Therefore, in any mixed-consumption Nash equilibrium in $(\mu;\bar{t})$, no consumer with type $d\ge t_h$ consumes a positive amount from producer $k$, and so $\mu\bigl([t_k,t_h)\bigr)\ge\ell_k(\bar{t})\ge\frac{\mu(\timeset)}{n}$, as required.
\end{proof}

\begin{definition}
Let $\bar{t}\in\timeset^{\prods}$ be a pure-strategy profile in \coarsegame. 
\begin{itemize}
\item
For every $q\in\{0,\ldots,n\}$, we define \[Q_q(\bar{t})\eqdef\Bigl\{j\in\prods\mid \mu\bigl([0,t_j)\bigr) \in \bigl(\tfrac{q-1}{n}\cdot\mu(\timeset),\tfrac{q}{n}\cdot\mu(\timeset)\bigr]\Bigr\}\subseteq\prods.\]
\item
We define $M(\bar{t})\eqdef\Max\bigl\{q\in\{0,\ldots,n\}~\big|~ Q_q(\bar{t})\ne\emptyset\bigr\}$.
\end{itemize}
\end{definition}

\begin{remark}\leavevmode
\begin{itemize}
\item
 $\bigl(Q_q(\bar{t})\bigr)_{0\le q\le n}$ is a partition of $\prods$.
\item
When the CDF of $\mu$ is continuous and strictly increasing, then $Q_0(\bar{t})$ is the set of producers
with strategy $0$, while $Q_q(\bar{t})$ for $0<q\le n$ is the set of producers whose strategies lie in the $q$\tth\ $\nicefrac{1}{n}$ of $\timeset$ (as measured
by $\mu$), i.e.\ above the \mbox{$(q-1)$}\tth\ $n$-tile yet not above the $q$\tth\ $n$-tile; for such a CDF, $M(\bar{t})$ is the index of the $\nicefrac{1}{n}$
of $\timeset$ containing the numerically largest strategy (or $0$ is all strategies are $0$), i.e.\ it is the index of the lowest $n$-tile above which no strategies lie.
\end{itemize}
\end{remark}

\begin{lemma}\label{coarse-jump-out-of-q}
Let $\bar{t}\in\timeset^{\prods}$ and let $t'_k$ be a best response to $\bar{t}_{-k}$.
If $\bar{t}$ is not a Nash equilibrium, then $\mu\bigl([0,t'_k)\bigr)\le\frac{M(\bar{t})-1}{n}\cdot\mu(\timeset)$.
\end{lemma}

\begin{proof}
Let $h\in\arg\Max_{j\in\prods} t'_j$, where $\bar{t}'_{-k}\eqdef\bar{t}_{-k}$.
We consider two cases. If $(\bar{t}_{-k},t'_k)$ is not a Nash equilibrium, then by \cref{coarse-best-response-far-from-least-load} we have $\mu\bigl([t'_k,t_h)\bigr)\ge\frac{\mu(\timeset)}{n}$, and so $\mu\bigl([0,t'_k)\bigr)\le\mu\bigl([0,t_h)\bigr)-\frac{\mu(\timeset)}{n}\le\frac{M(\bar{t})-1}{n}\cdot\mu(\timeset)$, as required. Otherwise, $(\bar{t}_{-k},t'_k)$ is a Nash equilibrium
while $\bar{t}$ is not. Therefore, by \cref{producer-coarse-nash-char}, there exists $\tilde{\jmath}\in\{0,\ldots,n-1\}$ s.t.\ $\mu\bigl([0,t'_k)\bigr)\le\frac{\tilde{\jmath}}{n}\cdot\mu(\timeset)<\mu\bigl([0,t_k)\bigr)$. As $\mu\bigl([0,t_k)\bigr)\le\frac{M(\bar{t})}{n}\cdot\mu(\timeset)$, we thus have $\tilde{\jmath}<M(\bar{t})$,
and so $\mu\bigl([0,t'_k)\bigr)\le\frac{M(\bar{t})-1}{n}\cdot\mu(\timeset)$, as required.
\end{proof}

\begin{lemma}\label{coarse-road-to-nash}
Let $\delta>0$, let $(\bar{t}^i,P_i)_{i=0}^{\infty}$ be a $\delta$-better-response dynamic in \coarsegame\ and let $i\in\mathbb{N}$. If $\bar{t}^i$ is not a Nash
equilibrium, then all of the following hold.
\begin{parts}
\item\label{coarse-road-to-nash-m-decreasing}
$M(\bar{t}^{i+1})\le M(\bar{t}^i)$.
\item\label{coarse-road-to-nash-q-decreasing}
$Q_{M(\bar{t}^i)}(\bar{t}^{i+1})\subseteq Q_{M(\bar{t}^i)}(\bar{t}^i)$.
\item\label{coarse-road-to-nash-bad-moving}
$\mu\bigl([t^{i+1}_j,t^i_j)\bigr)\ge\delta$, for every $j\in P_i\cap Q_{M(\bar{t}^i)}(\bar{t}^{i+1})$.
\end{parts}
\end{lemma}

\begin{remark}\label{coarse-road-to-nash-slightly-faster}
Finer analysis of similar nature may be used to show that $\delta$ may be replaced with $\bigl(n+1-M(\bar{t}^i)\bigr)\cdot\delta$ in \crefpart{coarse-road-to-nash}{bad-moving}.
\end{remark}

\begin{proof}[of \cref{coarse-road-to-nash}]
We commence by proving \cref{coarse-road-to-nash-m-decreasing}. Let $h\in\arg\Max_{j\in\prods} t^i_j$.
It is enough to show that $t^{i+1}_k<t^i_h$ for every $k\in P_i$. We consider two cases.
If $t^i_k$ is not a best response to $\bar{t}^i_{-k}$, then by definition of $\delta$-better-response dynamics and by \crefpart{producer-coarse-domination}{safe}, $t^{i+1}_k<t^i_k\le t^i_h$.
Otherwise, $t^i_k$, and hence also $t^{i+1}_k$, are best responses to $\bar{t}^i_{-k}$. Therefore, by \cref{coarse-best-response-far-from-least-load},
$\ell_h(\bar{t}^i)<\ell_k(\bar{t}^i)=\ell_k(\bar{t}^i_{-k},t^{i+1}_k)$; in particular, $k\ne h$.
By anonymity and by \cref{ell-single-valley} (for $j=h$), we have $\ell_k(\bar{t}^i_{-k},t^i_h)=\ell_h(\bar{t}^i_{-k},t^i_h)\le\ell_h(\bar{t}^i)<\ell_k(\bar{t}^i_{-k},t^{i+1}_k)$. Therefore, by \cref{ell-decreasing},
$t^{i+1}_k<t^i_h$ in this case as well, as required.

We now proceed to prove \cref{coarse-road-to-nash-q-decreasing}. Let  $k\in \prods\setminus Q_{M(\bar{t}^i)}(\bar{t}^i)$; we must show that $k\notin Q_{M(\bar{t}^i)}(\bar{t}^{i+1})$.
It is enough to consider the scenario in which $k\in P_i$, and to show that under this condition, $\mu\bigl([0,t^{i+1}_k)\bigr)\le\frac{M(\bar{t}^i)-1}{n}\cdot\mu(\timeset)$.
Once again, we consider two cases.
If $t^i_k$ is not a best response to $\bar{t}^i_{-k}$, then by definition of $\delta$-better-response dynamics and by \crefpart{producer-coarse-domination}{safe},
$t^{i+1}_k<t^i_k$ and so $\mu\bigl([0,t^{i+1}_k)\bigr)\le\mu\bigl([0,t^i_k)\bigr)\le\frac{M(\bar{t}^i)-1}{n}\cdot\mu(\timeset)$, as required.
Otherwise, $t^i_k$, and hence also $t^{i+1}_k$, are best responses to $\bar{t}^i_{-k}$. By \cref{coarse-jump-out-of-q}, in this case we have  $\mu\bigl([0,t^{i+1}_k)\bigr)\le\frac{M(\bar{t}^i)-1}{n}\cdot\mu(\timeset)$ as well, as required.

We conclude by proving \cref{coarse-road-to-nash-bad-moving}. Let $k\in P_i\cap Q_{M(\bar{t}^i)}(\bar{t}^{i+1})$.
As $\mu\bigl([0,t^{i+1}_k)\bigr)>\mbox{$\frac{M(\bar{t}^i)-1}{n}\cdot\mu(\timeset)$}$, by \cref{coarse-jump-out-of-q} we have that $t^{i+1}_k$ is not a best response to $\bar{t}^i_{-k}$. Therefore,
by definition of $\delta$-better-response dynamics, we have that $\ell_k(\bar{t}^i_{-k},t^{i+1}_k)\ge\ell_k(\bar{t}^i)+\delta$, and so by \cref{ell-decreasing,ell-lipschitz}, $\mu\bigl([t^{i+1}_k,t^i_k)\bigr)\ge\delta$ as required.
\end{proof}

\begin{proof}[of \cref{coarse-response-equilibrium}]
Let $\delta>0$ and let $(\bar{t}^i,P_i)_{i=0}^{\infty}$ be a $\delta$-better-response dynamic in \coarsegame.
By definition, $M(\bar{t}^0)\le n$.
By \cref{coarse-road-to-nash}, $M=M(\bar{t}^i)$ decreases by at least $1$ in every $\bigl\lceil\frac{\mu(\timeset)}{\delta n}\bigr\rceil$ rounds within which a Nash equilibrium is not reached.
Therefore, if a Nash equilibrium is not reached in at most $n\cdot\bigl\lceil\frac{\mu(\timeset)}{\delta n}\bigr\rceil$ rounds from $0$, then $M(\bar{t}^i)=0$ and so $\bar{t}^i\equiv 0$, which by
\cref{producer-coarse-nash-char} is a Nash equilibrium. If $\bar{t}^i$ is a Nash equilibrium, then by \cref{producer-coarse-nash-char}, $M(\bar{t}^{i+1})\le n-1$, and so if a Nash equilibrium is not reached in at most $(n-1)\cdot\bigl\lceil\frac{\mu(\timeset)}{\delta n}\bigr\rceil$ rounds from $i+1$, then we have $M=0$ once more, and so a Nash equilibrium is reached again.

The tighter bounds described in \cref{coarse-response-equilibrium-slightly-faster} may be shown in a similar manner, due to \cref{coarse-road-to-nash-slightly-faster}.
\end{proof}

\begin{definition}[$k$-Canonical Form]
Let $k\in\prods$ and let $\bar{t}\in\timeset^{\prods}$ be a pure-strategy Nash equilibrium in \coarsegame. We say that $\bar{t}$ is in \emph{$k$-canonical form} if all of the following hold.
\begin{conditions}
\item
$t_0\le t_1\le\cdots\le t_{k-1}\le t_{k+1}\le\cdots\le t_{n-1}$.
\item
$\mu\bigl([0,t_j)\bigr)\le \frac{j}{n}$ for every $j<k$.
\item
Either $k=n-1$, or $\mu\bigl([0,t_{k+1})\bigr)>\frac{k}{n}\cdot\mu(\timeset)$.
\end{conditions}
\end{definition}

\begin{lemma}\label{permute-to-canonical-form}
Let $k\in\prods$ and let $\bar{t}\in\timeset^{\prods}$ be a pure-strategy Nash equilibrium in \coarsegame.
There exists a permutation $\pi \in\prods!$ s.t.\ $(t_{\pi(0)},t_{\pi(1)},\ldots,t_{\pi(n-1)})$ is in \mbox{$\pi^{-1}(k)$-canonical} form.
\end{lemma}

\begin{proof}
We start by defining $\pi$ such that $t_{\pi(0)}\le t_{\pi(1)}\le\cdots\le t_{\pi(n-1)}$.
In particular, we have
\begin{equation}\label{permute-to-canonical-form-sort}
t_{\pi(0)}\le t_{\pi(1)}\le\cdots\le t_{\pi^{-1}(k)-1}\le t_{\pi^{-1}(k)+1}\le\cdots\le t_{\pi(n-1)}.
\end{equation}
By \cref{producer-coarse-nash-char}, we have that
\begin{equation}\label{permute-to-canonical-form-m}
\mu\bigl([0,t_{\pi(j)}\bigr)\le\tfrac{j}{n}\cdot\mu(\timeset)
\end{equation}
for every $j\in\prods$, and in particular for every $j<\pi^{-1}(k)$.
If $\pi^{-1}(k)=n-1$ or $\mu\bigl([0,t_{\pi(\pi^{-1}(k)+1)})\bigr)>\tfrac{\pi^{-1}(k)}{n}\cdot\mu(\timeset)$, then $\pi$ is a permutation as required.
Otherwise, we have
\begin{equation}\label{permute-to-canonical-form-p}
\mu\bigl([0,t_{\pi(\pi^{-1}(k)+1)})\bigr)\le\tfrac{\pi^{-1}(k)}{n}\cdot\mu(\timeset).
\end{equation}
In this case, we modify $\pi$ to create a new permutation $\pi'\in\prods!$ by incrementing $\pi^{-1}(k)$, or more formally --- by swapping the values of coordinates $\pi^{-1}(k)$ and $\pi^{-1}(k)+1$ of $\pi$.
We note that \cref{permute-to-canonical-form-sort} still holds w.r.t.\ $\pi'$ (i.e.\ by substituting $\pi'$ for $\pi$).
By \cref{permute-to-canonical-form-m} w.r.t.~$\pi$ for all $j<\pi^{-1}(k)$, we have that \cref{permute-to-canonical-form-m} holds w.r.t.\ $\pi'$
for all $j<\pi'^{-1}(k)-1$;
by \cref{permute-to-canonical-form-p} w.r.t.\ $\pi$, we have that \cref{permute-to-canonical-form-m} holds w.r.t.\ $\pi'$ for $j=\pi'^{-1}(k)-1$ as well.
Once again, if $\pi'^{-1}(k)=n-1$ or $\mu\bigl([0,t_{\pi'(\pi'^{-1}(k)+1)})\bigr)>\tfrac{\pi'^{-1}(k)}{n}\cdot\mu(\timeset)$, then $\pi'$ is a permutation as required. Otherwise, \cref{permute-to-canonical-form-p} holds w.r.t.\ $\pi'$, and we repeat the modification step. As $\pi^{-1}(k)$ is incremented in each modification step, this process concludes in at most $n-1$ steps, as it concludes if $\pi^{-1}(k)$ reaches $n-1$.
\end{proof}

\begin{lemma}\label{canonical-form-deviation}
Let $k\in\prods$ and let $\bar{t}\in\timeset^{\prods}$ be a pure-strategy Nash equilibrium in \coarsegame\ in $k$-canonical form.
Both of the following hold.
\begin{parts}
\item\label{canonical-form-deviation-all-good}
$\mu\bigl([0,t_j)\bigr)\le \frac{j}{n}$ for every $j\in\prods$.
\item\label{canonical-form-deviation-tk-ne}
Either $k=n-1$ or $t_k<t_{k+1}$.
\item\label{canonical-form-deviation-nash-char}
For every $t\in\timeset$, $(\bar{t}_{-k},t)$ is a Nash equilibrium in \coarsegame\ iff $\mu\bigl([0,t)\bigr)\le \frac{k}{n}$.
\end{parts}
\end{lemma}

\begin{proof}
By definition of $k$-canonical form, $\mu\bigl([0,t_j)\bigr)\le \frac{j}{n}$ for every $j<k$.
By definition of $k$-canonical form, we also have for every $n>j>k$ that $\mu\bigl([0,t_j)\bigr)\ge\mu\bigl([0,t_{k+1})\bigr)>\mbox{$\frac{k}{n}\cdot\mu(\timeset)$}$. Therefore, by
\cref{producer-coarse-nash-char}, $\mu\bigl([0,t_j)\bigr)\le\frac{k}{n}\cdot\mu(\timeset)$ for every $j\le k$, and in particular
$\mu\bigl([0,t_k)\bigr)\le\frac{k}{n}\cdot\mu(\timeset)$. Furthermore, we obtain that $(t_{k+1}, t_{k+2},\ldots,t_{n-1})$ are 
the $n-k-1$ numerically largest strategies in $\bar{t}$, and as they are sorted, by \cref{producer-coarse-nash-char} we have that $\mu\bigl([0,t_j)\bigr)\le\frac{j}{n}\cdot\mu(\timeset)$ for every $j>k$ as well and the proof of \cref{canonical-form-deviation-all-good} is complete.

Assume that $k<n-1$. As we have shown that $\mu\bigl([0,t_k)\bigr)\le\frac{k}{n}\cdot\mu(\timeset)$, but by definition of $k$-canonical form $\mu\bigl([0,t_{k+1})\bigr)>\frac{k}{n}\cdot\mu(\timeset)$, we have that $t_k<t_{k+1}$ and \cref{canonical-form-deviation-tk-ne} holds.

We conclude by proving \cref{canonical-form-deviation-nash-char};
let $t\in\timeset$. If $\mu\bigl([0,t)\bigr)\le \frac{k}{n}$, then by \cref{canonical-form-deviation-all-good,producer-coarse-nash-char}, $(\bar{t}_{-k},t)$
is a Nash equilibrium in \coarsegame. Recall that $\mu\bigl([0,t_j)\bigr)>\frac{k}{n}\cdot\mu(\timeset)$ for every $j>k$; therefore,
if $\mu\bigl([0,t)\bigr)>\frac{k}{n}\cdot\mu(\timeset)$ as well, then by \cref{producer-coarse-nash-char}, $(\bar{t}_{-k},t)$
is a not a Nash equilibrium in \coarsegame.
\end{proof}

\begin{proof}[of \cref{coarse-sequential}]
Let $\bar{t}\in\timeset^{\prods}$ be a Nash equilibrium in \coarsegame\ and let $k\in\prods$.
By \cref{permute-to-canonical-form}, assume w.l.o.g.\ that $\bar{t}$ is in $k$-canonical form.
By \crefpart{canonical-form-deviation}{nash-char}, it is enough to show that each $t_k<t'_k\le1$ s.t.\ $\mu\bigl([0,t'_k)\bigr)>\frac{k}{n}\cdot\mu(\timeset)$ is \emph{not}
a best response to $\bar{t}_{-k}$ in \coarsegame. If $k<n+1$, then by \crefpart{canonical-form-deviation}{tk-ne}, $t_k<t_{k+1}$ and so, by \cref{producer-coarse-domination}
and since $\mu\bigl([0,t_{k+1})\bigr)>\frac{k}{n}\cdot\mu(\timeset)$, it is enough to consider $t_k < t'_k \le t_{k+1}$ in this case.
By definition, for every $j\le k$, we have $\mu\bigl([0,t'_k)\bigr)>\frac{k}{n}\cdot\mu(\timeset)\ge\mu\bigl([0,t_j)\bigr)$
and so $t'_k>t_j$.

Let $s'$ and $s''$ be mixed-consumption Nash equilibria in the $k$-producer consumer game $(\mu|_{\cap[0,t'_k)};t_0,\ldots,t_{k-1})$
and in the $(n-k)$-producer game $(\mu|_{\cap[t'_k,1]};t'_k,t_{k+1},\ldots,t_{n-1})$, respectively; by abuse of notation,
we think of $s''$ as $s''=(s''_{\noconsumption},s''_k,s''_{k+1},\ldots,s''_{n-1})$ ($s''_{\noconsumption}\equiv0$ by definition of $s''$)
and for each $k\le j<n$ define $\loadtt{j}\eqdef\int_{\timeset} s''_j\,d(\mu|_{\cap[t'_k,1]})$.
For every $0\le j<k$, we have by definition of $k$-canonical form that
\[
\mu\bigl([0,t_j)\bigr)\le\frac{j}{n}\cdot\mu(\timeset)=\frac{j}{k}\cdot\frac{k}{n}\cdot\mu(\timeset)<\frac{j}{k}\cdot\mu\bigl([0,t'_k)\bigr)=\frac{j}{k}\cdot\mu|_{\cap[0,t'_k)}(\timeset).
\]
By \cref{producer-coarse-nash-char}, $(t_0,\ldots,t_{k-1})$ is therefore a Nash equilibrium in $(k,\mu|_{\cap[0,t'_k)},\succeq_C)$, and so by \cref{producer-coarse-nash-loads}
we have $\loadt{j}=\frac{\mu|_{\cap[0,t'_k)}(\timeset)}{k}=\frac{\mu([0,t'_k))}{k}>\frac{\mu(\timeset)}{n}$ for every $0\le j<k$.

For every $k < j < n$, we have by \crefpart{canonical-form-deviation}{all-good}  that
\[
\mu\bigl([t_j,1]\bigr)\ge\frac{n-j}{n}\cdot\mu(\timeset) > \frac{n-j}{n}\cdot\frac{n}{n-k}\cdot\mu\bigl([t'_k,1]\bigr)=\frac{n-j}{n-k}\cdot\mu\bigl([t'_k,1]\bigr),
\]
and therefore
\[
\mu\bigl([t'_k,t_j)\bigr)=\mu\bigl([t'_k,1]\bigr)-\mu\bigl([t_j,1]\bigr)<\mu\bigl([t'_k,1]\bigr)-\frac{n-j}{n-k}\cdot\mu\bigl([t'_k,1]\bigr)=\frac{j-k}{n-k}\cdot\mu|_{\cap[t'_k,1]}(\timeset).
\]
Note that $\mu\bigl([t'_k,t'_k)\bigr)=0=\frac{k-k}{n-k}\cdot\mu|_{\cap[t'_k,1]}(\timeset)$ trivially holds as well.
By \cref{producer-coarse-nash-char}, $(t'_k,t_{k+1},t_{k+2},\ldots,t_{n-1})$ is therefore a Nash equilibrium in $(n-k,\mu|_{\cap[t'_k,1]},\succeq_C)$, and so by \cref{producer-coarse-nash-loads}
we have that $\loadtt{j}=\frac{\mu|_{\cap[t'_k,1]}(\timeset)}{n-k}=\frac{\mu([t'_k,1])}{n-k}<\frac{\mu(\timeset)}{n}$ for every $k\le j<n$.

Let $s$ be the mixed-consumption profile defined by $s_j|_{[0,t'_k)}=s'_j$ for every $j\in\{\noconsumption,\linebreak0,1,\ldots,k-1\}$ and $s_j|_{[0,t'_k)}\equiv0$ for every $k\le j<n$, and by $s_j|_{[t'_k,1]}=s''_j$ for every $k\le j<n$ and $s_j|_{[t'_k,1]}\equiv0$ for every $j\in\{\noconsumption,0,1,\ldots,k-1\}$. By definition of $s'$ and of $s''$, we have that $s$ is a legal mixed-consumption
profile in $(\mu;\bar{t}_{-k},t'_k)$, and furthermore, that $\load{j}=\loadt{j}=\frac{\mu([0,t'_k))}{k}$ for every $j\in\{\noconsumption,0,1,\ldots,k-1\}$, and that $\load{j}=\loadtt{j}=\frac{\mu([t'_k,1])}{n-k}$ for every $k\le j<n$. By the former, and as $s'$ is a Nash equilibrium, no type $d\in[0,t'_k)$ has any incentive to deviate from $s$
in $(\mu;\bar{t}_{-k},t'_k)$, and by the latter, as $s''$ is a Nash equilibrium and as $\frac{\mu([t'_k,1])}{n-k}<\frac{\mu(\timeset)}{n}<\frac{\mu([0,t'_k))}{k}$, we have that neither does any type $d\in[t'_k,1]$. Therefore,
$s$ is a Nash equilibrium in $(\mu;\bar{t}_{-k},t'_k)$. As $\ell_k(\bar{t}_{-k},t'_k)=\load{k}=\frac{\mu([t'_k,1])}{n-k}<\frac{\mu(\timeset)}{n}=\ell_k(\bar{t})$ (with the last equality by \cref{producer-coarse-nash-loads}, since $\bar{t}$ is a Nash equilibrium in \coarsegame), we have that producer $k$
strictly prefers $t_k$ over $t'_k$ given $\bar{t}_{-k}$, and so $t'_k$ is not a best response to $\bar{t}_{-k}$ in \coarsegame, as required.

We note that an alternative proof may also be given via \cref{compute-ell} / \cref{compute-ell-no-algo}.
\end{proof}

\begin{proof}[of \cref{coarse-sequential-lazy-cor}]
A direct corollary of \cref{coarse-response-equilibrium,coarse-lazy,coarse-sequential}.
\end{proof}

\begin{proof}[of \cref{coarse-nonsequential}]
It is enough to show that some nonequilibrium can be reached in a finite number of steps from any Nash equilibrium.
Let $t\in\timeset$ s.t.\ $0<\mu\bigl([0,t)\bigr)\le\frac{n-1}{n}\cdot\mu(\timeset)$ (there must exist such $t$ by definition of $\mu$) and
let $j\in\{1,\ldots,n-1\}$ be minimal s.t.\ $\mu\bigl([0,t)\bigr)\le\frac{j}{n}\cdot\mu(\timeset)$.

By \cref{producer-coarse-dominant}, $0\in\timeset$ is a best-response by any producer to any Nash equilibrium, and so a (nonlazy)
best-response dynamic can reach $(0,0,\ldots,0)$ from any Nash equilibrium in one round.
Let $i\in\mathbb{N}$ be more than one round into the future after
reaching $(0,0,\ldots,0)$, s.t.\ $|P_i|>1$, and let $k,h\in P_i$ s.t.\ $k\ne h$.
By definition of $t$ and $j$ and by \cref{producer-coarse-nash-char}, any strategy profile
in which at most $n-j$ producers play $t$ and the rest play $0$ is a Nash equilibrium.
Therefore, within one round after reaching $(0,0,\ldots,0)$, a Nash equilibrium in which $j+1$ producers, including $k$ and $h$, play $0$ and the rest play $t$, can be reached,
and can be lazily maintained until the step $i$. In step $i$, all triggered producers may play $t\in\timeset$, which is a best response for each of them since at least $j$ producers playing~$0$ and the rest playing $t$ is a Nash equilibrium. Therefore, and as $k$ and $h$ both switch from playing $0$ to playing $t$ at $i$, at least $n-j+1$ producers play $t$ at $i+1$, which, by definition of $j$ and by \cref{producer-coarse-nash-char}, is a nonequilibrium.
\end{proof}

\begin{lemma}\label{must-jump-to-zero-coarse}
Let $k\in\prods$, let $\bar{t}_{-k}\in\timeset^{\prods\setminus\{k\}}$. If $\mu\bigl([0,t_j)\bigr)>0$ for all $j\in\prods\setminus\{k\}$, then 
$t'_k\in\timeset$ is a best response (by $k$) to $\bar{t}_{-k}$ in \coarsegame\ iff $\mu\bigl([0,t'_k)\bigr)=0$.
\end{lemma}

\begin{proof}
By \cref{producer-coarse-domination}, all strategies $t\in\timeset$ for which $\mu(\bigl([0,t)\bigr)=0$ are equivalent. As in particular, $t=0\in\timeset$ is such a strategy, it is therefore enough to show that $k$ strictly prefers to play
$0\in\timeset$ over any $t'_k$ s.t.\ $\mu\bigl([0,t'_k)\bigr)>0$. By
\crefpart{producer-coarse-domination}{safe}, it is enough to consider the case in which $t'_k\le t_j$ for all $j\in\prods\setminus\{k\}$.
By \cref{compute-ell} (for $t_n$ as defined there), $\ell_k(\bar{t}_{-k},0)=\Max_{0<j\le n}\frac{\mu([0,t_j))}{j}>\Max_{0<j\le n}\frac{\mu([t'_k,t_j))}{j}=\ell_k(\bar{t}_{-k},t'_k)$, as required.
\end{proof}

\begin{proof}[of \cref{coarse-best-response-fast}]
As in the proof of \cref{coarse-response-equilibrium}, and as a best-response dynamic is also a $\mu(\timeset)$-better-response dynamic,
we have that  $M(\bar{t}^0)\le n$ (and $M\le n-1$ on every step immediately following a Nash equilibrium), and that $M$ decreases by at
least $1$ every round if a Nash equilibrium is not reached.
Let $i\in\mathbb{N}$ s.t.\ $M(\bar{t}^i)=1$, $M(\bar{t}^{i-1})=2$, and $\bar{t}^{i-1}$ is not a Nash equilibrium;
it is enough to show that $\bar{t}^i$ is a Nash equilibrium.

Since $M(\bar{t}^i)=1$, by \cref{producer-coarse-nash-char} it is enough to show that there exists
$j\in\prods$ s.t.\ $\mu\bigl([0,t^i_j)\bigr)=0$. As $M(\bar{t}^i)\ne M(\bar{t}^{i-1})$, we have $P_{i-1}\ne\emptyset$.
By \cref{producer-coarse-nash-char}, as $M(\bar{t}^{i-1})=2$ but $\bar{t}^{i-1}$ is not a Nash equilibrium, there exists
at most one producer $k\in\prods$ s.t.\ $\mu\bigl([0,t^{i-1}_k)\bigr)=0$. If there exists no such producer, then by \cref{must-jump-to-zero-coarse}, we have $\mu\bigl([0,t^i_j)\bigr)=0$
for every $j\in P_i$, and as $P_i\ne\emptyset$, the proof is complete. Otherwise, there exists a unique $k\in\prods$ s.t.\ $\mu\bigl([0,t^{i-1}_k)\bigr)=0$. If $k\notin P_i$,
then $t^i_k=t^{i-1}_k$ and the proof is complete. Otherwise, $k\in P_i$ and by \cref{must-jump-to-zero-coarse}, we have $\mu\bigl([0,t^i_k)\bigr)=0$, as required.
\end{proof}

\begin{proof}[of \cref{coarse-best-response-fast-cor}]
A direct corollary of \cref{coarse-best-response-fast,coarse-lazy,coarse-sequential,coarse-best-response-fast-tight}.
\end{proof}

\subsection{Proofs and Auxiliary Results for Section~\refintitle{producers-fine}}\label{producers-fine-proofs}

\subsubsection{Proofs and Auxiliary Results for Section~\refintitle{producers-fine-statics}}

\begin{lemma}[Domination]\label{producer-fine-domination}
Let $t\ne t' \in \timeset$ be strategies in \finegame. $t$ is a safe alternative to $t'$ iff either
of the following hold. In either case, $t$ strongly dominates $t'$
\begin{conditions}
\item\label{producer-fine-domination-slight-right}
$t> t'$ and $\mu\bigr([t',t)\bigl)=0$.
\item\label{producer-fine-domination-of-end}
$\mu\bigl([t',1]\bigr)<\frac{\mu([t,1])}{n}$.
\end{conditions}
\end{lemma}

\begin{proof}
$t$ is a safe alternative to (alternatively, strongly dominates) $t'$ iff either $t$ always produces greater load than $t'$, or $t$ always produces at least as much load as $t'$ and
in addition $t>t'$. By \crefpart{producer-coarse-domination}{strong}, the former occurs iff $\mu\bigl([t',1]\bigr)<\frac{\mu([t,1])}{n}$;
by \cref{producer-coarse-domination}(\labelcref{producer-coarse-domination-safe},\labelcref{producer-coarse-domination-weak}), if $t>t'$, then the latter occurs iff $\mu\bigr([t',t)\bigl)=0$
\end{proof}

\begin{lemma}\label{union-up}
$\mu\Bigl(\bigcup\Bigl\{[t,t') ~\Big|~ 0\le t<t' \And \mu\bigl([t,t')\bigr)\le m\Bigr\}\Bigr) \le m$, for every $t' \in \timeset$ and $m \in \Rge$.
\end{lemma}

\begin{proof}
Define $U\eqdef\bigcup\Bigl\{[t,t') ~\Big|~ 0\le t<t' \And \mu\bigl([t,t')\bigr)\le m\Bigr\}$. If $U=\emptyset$, then $\mu(U)=0\le m$ and the proof
is complete; assume, therefore, that $U\ne\emptyset$ and let $u\eqdef\inf U\ge0$. By definition, $U$ is connected,
and therefore by definition of $U$ and $u$, we have that either $U=[u,t')$ or $U=(u,t')$. If $U=[u,t')$, then $u \in U$, and by definition of $U$, we obtain $\mu(U)=\mu\bigl([u,t')\bigr)\le m$, as required; assume therefore, that $U=(u,t')$. In this case, $U=\bigcup\Bigl\{[t,t') ~\Big|~ t \in [0,t')\cap\mathbb{Q} \And \mu\bigl([t,t')\bigr)\le m\Bigr\}$, and by continuity of $\mu$ from below, we obtain $\mu(U)\le m$, as required.
\end{proof}

\begin{proof}[of \cref{producer-fine-dominant}]
Let $t \in \timeset$ be a dominant strategy in this game. By \cref{producer-fine-domination}, both $\mu\bigl([0,t)\bigr)=0$ (as $t$ is a safe alternative to $0$) and
$\mu\bigl([t',1]\bigr)<\frac{\mu([t,1])}{n}$ for every $t'>t$ (as $t$ is a safe alternative to every such $t'$). (Alternatively, the former holds as by definition, any strategy dominant w.r.t.\ fine preferences is also dominant w.r.t.\ coarse preferences, and by \cref{producer-coarse-dominant}.) By the former, $\mu\bigl([t,1]\bigr)=\mu(\timeset)$, and therefore and by the latter,
$\mu\bigl([t',1]\bigr)<\frac{\mu(\timeset)}{n}$ for every $t'\in(t,1)\cap\mathbb{Q}$. Therefore, by continuity of $\mu$ from below, $\mu\bigl((t,1]\bigr)\le\frac{\mu(\timeset)}{n}$. Hence,
$\mu\bigl(\{t\}\bigr)=\mu\bigl([0,t]\bigr)\ge\frac{n-1}{n}\cdot\mu(\timeset)$.
(Conversely, by \cref{producer-fine-domination}, it is easy to see that if there indeed exists $t\in\timeset$  s.t.\ $ \mu\bigl([0,t)\bigr)=0$ and
$\mu\bigl([t',1]\bigr)<\frac{\mu(\timeset)}{n}$ for every $t'>t$, then it constitutes the unique dominant strategy.)

For the second statement, we note that
by \cref{producer-fine-domination}, the set of dominated strategies is
\begin{align*}
&\;\Bigl\{t \in \timeset \:\Big|\: \mu\bigl([t,1]\bigr)<\tfrac{\mu(\timeset)}{n}\Bigr\}\cup\Bigl\{t \in \timeset \:\Big|\: \exists t<t'\le1:\mu\bigl([t,t')\bigr)=0\Bigr\}= \\
=&\;\bigcup\Bigl\{[t,1] \:\Big|\: t\!\in\!\timeset \And \mu\bigl([t,1]\bigr)\!<\!\tfrac{\mu(\timeset)}{n}\Bigr\} \cup \bigcup\Bigl\{[t,t') \:\Big|\: t'\!\in\!\timeset\!\cap\!\mathbb{Q} \And 0\!\le\!t\!<\!t' \And \mu\bigl([t,t')\bigr)\!=\!0\Bigr\}.
\end{align*}
By  $\sigma$-additivity of $\mu$ and by \cref{union-up} (applied twice),
\begin{align*}
&\;\mu\Bigl(\Bigl\{t \in \timeset \:\Big|\: \mu\bigl([t,1]\bigr)<\tfrac{\mu(\timeset)}{n}\Bigr\}\cup\Bigl\{t \in \timeset \:\Big|\: \exists t<t'\le1:\mu\bigl([t,t')\bigr)=0\Bigr\}\Bigr)\le \\
\le&\;\mu\Bigl(\bigcup\Bigl\{[t,1] \:\Big|\: t \in \timeset \And \mu\bigl([t,1]\bigr)<\tfrac{\mu(\timeset)}{n}\Bigr\}\Bigr)+\\
&\qquad\qquad\qquad\qquad\qquad\qquad\qquad\qquad+\smashoperator[lr]{\sum_{t'\in\timeset\cap\mathbb{Q}}}\mu\Bigl(\bigcup\Bigl\{[t,t') \:\Big|\: 0\le t<t' \And \mu\bigl([t,t')\bigr)=0\Bigr\}\Bigr)\le\\
\le&\;\frac{\mu(\timeset)}{n}+0=\frac{\mu(\timeset)}{n}.
\end{align*}
We note that this bound is attained if $\mu$ is atomless, as in this case it is straightforward to verify that $\mu\bigl(\bigcup\bigl\{[t,1] \mid t \in \timeset \And \mu\bigl([t,1]\bigr)<\frac{\mu(\timeset)}{n}\bigr\}\bigr)=\frac{\mu(\timeset)}{n}$.

The second statement leads to an extremely concise, yet somewhat more obscure, proof for the first one.
If a dominant strategy exists, then it is a safe alternative to all other strategies; in particular, all other strategies
have safe alternatives (other than themselves).
By the second statement, at least $\frac{n-1}{n}$ of $\mu$ is therefore concentrated on this dominated strategy, and the proof is complete.
\end{proof}

\begin{lemma}\label{max-bounded-measure-strategy-attained}
$\bigl\{t \in \timeset \mid \mu\bigl([t',t)\bigr)\le m\bigr\}$ attains a maximum value for every $t'\in\timeset$ and $m \in \Rge$.
\end{lemma}

\begin{proof}
Denote $S\eqdef \bigl\{t \in \timeset \mid \mu\bigl([t',t)\bigr)\le m\bigr\}$. We note that $t'\in S$. If $t'=\sup S$, then $t'=\Max S$ and the proof is complete.
Assume, therefore, that $t'<\sup S$.
Let $U\eqdef\bigcup\bigl\{[t',t)\mid t \in \timeset \And \mu\bigl([t',t)\bigr)\le m\bigr\}$.  As $t'<\sup S$, we have that $U\ne\emptyset$.
Let $u\eqdef\sup U\le1$. By definition of $U$ and of $u$, we have $U=[t',u)$.
By definition of $U$ and of $S$, we have that $u=\sup U=\sup S$, and so it is enough to show that $u \in S$,
i.e.\ that $\mu\bigl([t',u)\bigr)\le m$. As $U=[t',u)$, this is equivalent to showing that $\mu(U)\le m$.
Observe that
\[
U=\bigcup\Bigl\{[t',t) \:\Big|\: t \in \timeset\cap\mathbb{Q} \And \mu\bigl([t',t)\bigr)\le m\Bigr\}.\]
By continuity of $\mu$ from below, we thus obtain $\mu(U)\le m$, as required.
\end{proof}

\begin{proof}[of \cref{producer-fine-nash-char}]
For every $j\in\prods$, let $t_j\eqdef\Max\bigl\{t \in \timeset \mid \mu\bigl([0,t)\bigr)\le\mbox{$\frac{j}{n}\cdot\mu(\timeset)$}\bigr\}$. ($t_j$ is well-defined by \cref{max-bounded-measure-strategy-attained}.)
By \cref{producer-coarse-nash-char}, $\bar{t}$ is a Nash equilibrium in \coarsegame, and so by \cref{producer-coarse-nash-loads}, $\ell_j(\bar{t})=\frac{\mu(\timeset)}{n}$ for every $j\in\prods$.
We now first show that no Nash equilibrium other than $\bar{t}$ (up to permutations) exists in \finegame, and then show that $\bar{t}$ (and hence all permutations thereof)
is a super-strong Nash equilibrium in \finegame.

Let $t'_0\le\cdots\le t'_{n-1}\in\timeset$ s.t.\ $\bar{t}'$ is a Nash equilibrium in \finegame. We will show that $t'_j=t_j$ for every $j\in\prods$.
By definition of coarse and fine preferences, $\bar{t}'$ is also a Nash equilibrium in \coarsegame. Therefore, by \cref{producer-coarse-nash-char} we have that $\mu\bigl([0,t'_j)\bigr)\le\frac{j}{n}\cdot\mu(\timeset)$ for every $j\in\prods$. Hence we have for every $j\in\prods$ both that $t'_j\le t_j$ and (by \cref{producer-coarse-nash-char} again) that $(\bar{t}'_{-j},t_j)$ is a Nash equilibrium in \coarsegame\ as well. Therefore, by \cref{producer-coarse-nash-loads}, $\ell_j(\bar{t}')=\frac{\mu(\timeset)}{n}=\ell_j(\bar{t}'_{-j},t_j)$. As $\bar{t}'$ is a Nash equilibrium in \finegame, we therefore have that $t'_j\ge t_j$, and so $t'_j=t_j$, as required.

We now show that $\bar{t}$ is a super-strong Nash equilibrium in \finegame.
Assume for contradiction that there exists a coalition $P\subseteq\prods$ and strategies $\bar{t}'=(t'_j)_{j\in P}\in\timeset^P$ s.t.\ $j$ weakly prefers $(\bar{t}_{-P},\bar{t}')$ over $\bar{t}$ w.r.t.\ fine preferences for every $j\in P$, with a strict preference for at least one producer $j\in P$.
For every $j\in P$, as $j$ weakly prefers $(\bar{t}_{-P},\bar{t}')$ over $\bar{t}$, we have that $\ell_j(\bar{t}_{-P},\bar{t}')\ge\ell_j(\bar{t})$.
As $\bar{t}$ is a Nash equilibrium in \coarsegame, by \cref{producer-coarse-super-strong}, we therefore have $\ell_j(\bar{t}_{-P},\bar{t}')=\ell_j(\bar{t})$ for every $j\in P$. Therefore, by definition of $P$ and $\bar{t}'$,
we have $t'_j\ge t_j$ for every $j\in P$, with a strict inequality for at least one producer $j\in P$ --- let $j$ be such a producer for which $t'_j$ is greatest. Assume w.l.o.g.\
that either $j=n-1$ or $t_j<t_{j+1}$; therefore, $\bar{t}$ is in $j$-canonical form.
As $\bar{t}$ is also a Nash equilibrium in \coarsegame, by \cref{coarse-sequential}, by \crefpart{canonical-form-deviation}{nash-char}, by definition of $t_j$, and as $t'_j>t_j$, we have $\ell_j(\bar{t})>\ell_j(\bar{t}_{-j},t'_j)$. By \cref{ell-single-valley} and by definition of $j$, we have $\ell_j(\bar{t}_{-j},t'_j)\ge\ell_j(\bar{t}_{-P},\bar{t}')$. Therefore, $\ell_j(\bar{t})>\ell_j(\bar{t}_{-j},t'_j)\ge\ell_j(\bar{t}_{-P},\bar{t}')$ --- a contradiction.
\end{proof}

\begin{proof}[of \cref{producer-fine-nash-char-special}]
Since the CDF of $\mu$ is continuous and strictly increasing, for every $j\in\prods$ there exists a unique strategy
$t_j \in \timeset$ s.t.\ $\mu\bigl([0,t_j)\bigr)=\frac{j}{n}\cdot\mu(\timeset)$; hence,
$t_j=\Max\bigl\{t \in \timeset \mid \mu\bigl([0,t)\bigr)\le\frac{j}{n}\cdot\mu(\timeset)\bigr\}$ and by \cref{producer-fine-nash-char}
the proof is complete.
\end{proof}

\begin{proof}[of \cref{producer-fine-nash-pure}]
Direct from definition of \finegame, as no player is ever indifferent between any two strategies,
regardless of the information such player possesses regarding the strategies of the other players.
\end{proof}

\begin{proof}[of \cref{producer-fine-market-allocation}]
By \cref{producer-fine-nash-char}, for every $j\in\prods$ we have
\[\bigl\{d\in\timeset \mid \mu\bigl([0,d)\bigr)\in\bigl(\tfrac{j}{n}\cdot\mu(\timeset),\tfrac{j+1}{n}\cdot\mu(\timeset)\bigr]\bigr\}\subseteq[t_j,t_{j+1}).\]
Therefore, it is enough to show that $s_j(d)=1$ for almost all $d\in[t_j,t_{j+1})$.
By definition of $s$, this is equivalent to showing that $\int_{[t_j,t_{j+1})}s_j d\mu=\mu\bigl([t_j,t_{j+1})\bigr)$ for every $j\in\prods$, where $t_n\eqdef1$.
As $\mu$ is atomless, by \cref{producer-fine-nash-char} and by definition of $t_n$
we have $\mu\bigl([t_j,t_{j+1})\bigr)=\mu\bigl([0,t_{j+1})\bigr)-\mu\bigl([0,t_j)\bigr)=\frac{j+1}{n}\cdot\mu(\timeset)-\frac{j}{n}\cdot\mu(\timeset)=\frac{\mu(\timeset)}{n}$ for every $j\in\prods$.
We prove the \lcnamecref{producer-fine-market-allocation} by induction on $j$. Let $k\in\prods$ and assume that the \lcnamecref{producer-fine-market-allocation} holds for every $0\le j<k$; we now show that it holds for $j=k$ as well.

For every $0\le j<k$, by the induction hypothesis, $\int_{[t_j,t_{j+1})}s_j d\mu=\mu\bigl([t_j,t_{j+1})\bigr)=\frac{\mu(\timeset)}{n}$.
By \cref{producer-fine-nash-char} and by definition of $\ell_j(\bar{t})$ and of $\load{j}$,
\[\frac{\mu(\timeset)}{n}=\ell_j(\bar{t})=\load{j}=\int_{\timeset}s_j d\mu\ge\int_{[t_j,t_{j+1})}s_j d\mu+\int_{[t_k,t_{k+1})}s_j d\mu=\frac{\mu(\timeset)}{n}+\int_{[t_k,t_{k+1})}s_j d\mu,\]
and so $\int_{[t_k,t_{k+1})}s_j d\mu=0$. By definition of Nash equilibrium, $s_{\noconsumption}(d)=0$ for every $d\ge t_k$, and therefore $\int_{[t_k,t_{k+1})}s_{\noconsumption} d\mu=0$
as well. Let $S_k\eqdef\{\noconsumption,0,1,2,\ldots,k-1\}$; by definition of $s$, we have that $s(d)\in\Delta^{S_k}$ for every $d\in[t_k,t_{k+1})$.
Therefore,
\[\int_{[t_k,t_{k+1})}s_k d\mu=\mu\bigl([t_k,t_{k+1})\bigr)-\sum_{j\in S_k}\int_{[t_k,t_{k+1})}s_j d\mu=\frac{\mu(\timeset)}{n}-0=\mu\bigl([t_k,t_{k+1})\bigr),\]
and the proof by induction is complete.
\end{proof}

\subsubsection{Proofs for Section~\refintitle{producers-fine-dynamics}}

\begin{proof}[of \cref{fine-unique-best-response} (Nonconstructive)]
Let $\tilde{\ell}\eqdef\ell_j(\bar{t}_{-j},0)$; by \cref{producer-coarse-dominant}, $\tilde{\ell}$ is the maximum load attainable by $j$
given $\bar{t}_{-j}$. Define $S\eqdef\bigl\{\mu\bigl([0,t)\bigr) ~\big|~ t\in\timeset\And\ell_j(\bar{t}_{-j},t)=\tilde{\ell}\bigr\}$. Observe that $S\ne\emptyset$ as $0\in S$ (given by $t=0$);
let $m\eqdef\sup S$. By \crefpart{producer-coarse-domination}{safe},
every $t\in\timeset$ s.t.\ $\mu\bigl([0,t)\bigr)<m$ maximizes $\ell_j(\bar{t}_{-j},t)$,
while every $t\in\timeset$ s.t.\ $\mu\bigl([0,t)\bigr)>m$ does not. Assume for contradiction that there exists $t\in\timeset$
s.t.\ $\mu\bigl([0,t)\bigr)=m$ and $\ell_j(\bar{t}_{-j},t)<\tilde{\ell}$. Let $\varepsilon\eqdef\tilde{\ell}-\ell_j(\bar{t}_{-j},t)>0$.
By definition of $m$, there exists $t'$ s.t.\ $m\ge\mu\bigl([0,t')\bigr)>m-\varepsilon$ and $\ell_j(\bar{t}_{-j},t')=\tilde{\ell}=\ell_j(\bar{t}_{-j},t)+\varepsilon$; by \cref{ell-lipschitz}, this is a contradiction. (We note that we have not shown (yet) that there exists $t\in\timeset$
s.t.\ $\mu\bigl([0,t)\bigr)=m$, but rather that every such $t$ maximizes the load on $j$.)
Therefore, we have that the set of load-maximizing strategies for $j$ is precisely $\bigl\{t\in\timeset~\big|~\mu\bigl([0,t)\bigr)\le m\bigr\}$.
By \cref{max-bounded-measure-strategy-attained}, this set attains a maximum value.
As by definition we have that a best response in \finegame\ is a numerically largest load-maximizing response, we obtain that this
maximum value is a best response as required. Uniqueness follows directly from definition of fine preferences.
\end{proof}

Before constructively proving \cref{fine-unique-best-response}, we first constructively prove it for two special cases.

\begin{corollary}[\cref{fine-unique-best-response} --- Special Case: Large Strategies]\label{must-jump-to-zero-fine}
Let $k\in\prods$, and let $\bar{t}_{-k}\in\timeset^{\prods\setminus\{k\}}$.
If $\mu\bigl([0,t_j)\bigr)>0$ for all $j\in\prods\setminus\{k\}$, then the unique best response (by $k$) to $\bar{t}_{-k}$ in \finegame\ is
$\Max\bigl\{t \in \timeset \mid \mu\bigl([0,t)\bigr)=0\bigr\}$.
\end{corollary}

\begin{proof}
A direct corollary of \cref{must-jump-to-zero-coarse}, as a best response in \finegame\ is a numerically largest load-maximizing response; the specified strategy is well defined
by \cref{max-bounded-measure-strategy-attained}.
\end{proof}

\begin{lemma}[\cref{fine-unique-best-response} --- Special Case: Coarse Equilibrium]\label{fine-best-response-to-coarse-nash}
Let $\bar{t}\in\timeset^{\prods}$ be a Nash equilibrium in \coarsegame. For every $j\in\prods$, a best response (by $j$) to $\bar{t}_{-j}$ exists in \finegame.
\end{lemma}

\begin{proof}[(Constructive)]
By \cref{permute-to-canonical-form}, assume w.l.o.g.\ that $\bar{t}$ is in $j$-canonical form.
We will show that a best response as required is given by $t'_j\eqdef\Max\bigl\{t \in \timeset \mid \mu\bigl([0,t)\bigr)\le\frac{j}{n}\cdot\mu(\timeset)\bigr\}$. ($t'_j$ is well defined by \cref{max-bounded-measure-strategy-attained}.)
By \crefpart{canonical-form-deviation}{nash-char} and by \cref{coarse-sequential}, a strategy $t\in\timeset$ maximizes $\ell_j(\bar{t}_{-j},t)$ iff $\mu\bigl([0,t)\bigr)\le\frac{j}{n}\cdot\mu(\timeset)$.
As by definition we have that a best response in \finegame\ is a numerically largest load-maximizing response, we obtain that $t'_j$ is a best response
as required.
\end{proof}

\begin{proof}[of \cref{fine-unique-best-response} (Constructive)]
W.l.o.g.\ we prove the result for $j=0$. Assume w.l.o.g.\ that $t_1\le t_2\le\cdots\le t_{n-1}$.
Uniqueness follows directly from definition of fine preferences; it is therefore enough to show that a best response exists.
If $\mu\bigl([0,t_1)\bigr)>0$, then by \cref{must-jump-to-zero-fine}
a best response exists as required.
Assume therefore henceforth that $\mu\bigl([0,t_1)\bigr)=0$.
If $(\bar{t}_{-0},0)$ is a Nash equilibrium in \coarsegame, then by \cref{fine-best-response-to-coarse-nash} a best response exists as required. Assume therefore henceforth that $(\bar{t}_{-0},0)$ is not a Nash equilibrium in \coarsegame.

Let $k\in\prods$ be minimal s.t.\ $\ell_k(\bar{t}_{-0},0)<\ell_0(\bar{t}_{-0},0)$. (Such $k$ exists by \cref{producer-coarse-nash-least-most-loads},
since $(\bar{t}_{-0},0)$ is not a Nash equilibrium in \coarsegame; by definition, $k>0$.)
By definition of $k$ and by \cref{high-untouched,equal-loads-consumer-nash}, $(0,t_1,t_2,\ldots,t_{k-1})$ is a Nash equilibrium
in $(k,\mu|_{\cap[0,t_k)},\succeq_C)$. By \cref{fine-best-response-to-coarse-nash}, there exists a best response $t_0\in\timeset$
to $(t_1,t_2,\ldots,t_{k-1})$ in $(k,\mu|_{\cap[0,t_k)},\succeq_F)$. We claim that $t_0$ is a best response to $\bar{t}_{-0}$ in
\finegame\ as well.

As $\ell_k(\bar{t}_{-0},0)<\ell_0(\bar{t}_{-0},0)$, we have that $\ell_0(\bar{t}_{-0},0)\le\mu\bigl([0,t_k)\bigr)$, and in particular $\mu\bigl([0,t_k)\bigr)>0$. Therefore, 
by \crefpart{producer-coarse-domination}{strong} and by definition of $t_0$, we obtain $t_0<t_k$.
By \cref{coarse-sequential}, $(t_0,t_1,\ldots,t_{k-1})$ is also a Nash equilibrium in $(k,\mu|_{\cap[0,t_k)},\succeq_C)$. Therefore, by \cref{producer-coarse-nash-loads},
$\ell_j(\mu|_{\cap[0,t_k)};0,t_1,t_2,\ldots,t_{k-1})=\frac{\mu([0,t_k))}{k}=\ell_j(\mu|_{\cap[0,t_k)};t_0,t_1,\ldots,t_{k-1})$ for every $j\in\prodsk$. By
the construction in the proof of \cref{consumer-symmetric-nash-exists} and as $t_0<t_k$, therefore $\ell_j(\bar{t})=\ell_j(\bar{t}_{-0},0)$
for every $j\in\prods$. By \cref{high-untouched} and by definition of $k$, we have $\ell_j(\bar{t})=\ell_j(\bar{t}_{-0},0)=\ell_j(\mu|_{\cap[0,t_k)};0,t_1,t_2,\ldots,t_{k-1})=\frac{\mu([0,t_k))}{k}$ for every $j\in\prodsk$. As $\ell_0(\bar{t})=\ell_0(\bar{t}_{-0},0)$,
by \cref{producer-coarse-dominant} we have that $t_0$ maximizes the load on producer $0$ in \finegame.

Let $h\in\prodsk$ s.t.\ $t_h\le t_0<t_{h+1}$. Such $h>0$ exists as $\mu\bigl([0,t_1)\bigr)=0$ and since
$t_0\ge\Max\bigl\{t \in \timeset \mid \mu\bigl([0,t)\bigr)=0\bigr\}\ge t_1$ by \cref{producer-coarse-domination} (this maximum value
is attained by \cref{max-bounded-measure-strategy-attained}), and $h<k$ since $t_0<t_k$.
It remains to show that every $t'_0\in(t_0,1])$
does not maximize the load on producer $0$ in \finegame. By \crefpart{producer-coarse-domination}{safe}, it is enough to consider
the case $t_0<t'_0<t_{h+1}$. Note that $t_k>t_0 \ge t_1$ and so $k>1$.

By definition of $t'_0$ and $t_0$, we have that $t'_0$ does not maximize
the load on producer $0$ in $(k,\mu|_{\cap[0,t_k)},\succeq_F)$.
By \cref{high-untouched} and by \cref{producer-coarse-nash-least-most-loads} (since $\mu\bigl([0,t_1)\bigr)=0$), $\ell_1(\bar{t}_{-0},t'_0)\ge\ell_1(\mu|_{\cap[0,t_k)};t'_0,t_1,\ldots,t_{k-1})>\frac{\mu([0,t_k))}{k}=\ell_1(\bar{t})$,
and so by \cref{compute-ell}, $\ell_1(\bar{t}_{-0},t'_0)=\frac{\mu([0,t'_0))}{h}$ and so $\ell_j(\bar{t}_{-0},t'_0)=\frac{\mu([0,t'_0))}{h}$ for every $0<j\le h$, and in particular for $j=h$.
Therefore, and by \crefpart{producer-coarse-domination}{safe}, $\ell_h(\bar{t}_{-0},t'_0)=\ell_1(\bar{t}_{-0},t'_0)>\ell_1(\bar{t})=\ell_0(\bar{t})\ge\ell_0(\bar{t}_{-0},t'_0)$.
As $\frac{\mu([0,t'_0))}{h}=\ell_1(\bar{t}_{-0},t'_0)>\frac{\mu([0,t_k))}{k}$,
we have $\mu\bigl([t'_0,t_j)\bigr)=\mu\bigl([0,t_j)\bigr)-\mu\bigl([0,t'_0)\bigr)<\mu\bigl([0,t_j)\bigr)-\frac{h}{k}\cdot\mu\bigl([0,t_k)\bigr)$ for every $j>h$.
As $\ell_h(\bar{t}_{-0},t'_0)>\ell_0(\bar{t}_{-0},t'_0)$, by definition of $h$ and by \cref{compute-ell,compute-ell-no-algo}, we obtain (for $t_n$ as defined there) that
$\ell_0(\bar{t}_{-0},t'_0)=\Max_{h<j\le n}\frac{\mu([t'_0,t_j))}{j-h}<\Max_{h<j\le n}\Bigl(\mu\bigl([0,t_j)\bigr)-\frac{h}{k}\cdot\mu\bigl([0,t_k)\bigr)\Bigr)\cdot\frac{1}{j-h}=
\Max_{h<j\le n}\Bigl(\mu\bigl([0,t_j)\bigr)-\sum_{i=1}^{h}\ell_i(\bar{t})\Bigr)\cdot\frac{1}{j-h}=
\ell_0(\bar{t})$,
and the proof is complete.
\end{proof}

\begin{proof}[of \cref{fine-response-equilibrium}]
A direct corollary of \cref{coarse-response-equilibrium,coarse-best-response-fast}, as any weakly-/$\delta$-better-/best-response dynamic w.r.t.\ fine preferences is also a \mbox{weakly-/$\delta$}-\linebreak better-/best-response dynamic w.r.t.\ coarse preferences.
\end{proof}

\begin{proof}[of \cref{fine-sequential}]
Let $(\bar{t}^0,P_i)_{i=0}^{\infty}$ be a sequential $\delta$-better-response dynamic in \finegame\ s.t.\ $\bar{t}^0$ is a
Nash equilibrium w.r.t.\ \coarsegame. Let $k\in\prods$ and let $\tilde{\imath}$ be minimal s.t.\ $k\in P_{\tilde{\imath}}$. It is enough
to show that $t^i_k$ is constant for $i>\tilde{\imath}$, as this implies that after one round $\bar{t}^i$ is constant regardless
of $P_i$, and is thus a Nash equilibrium in \finegame.

By \cref{coarse-sequential}, $\bar{t}^i$ is a Nash equilibrium w.r.t.\ \coarsegame\ for every $i\in\mathbb{N}$;
therefore, by \cref{producer-coarse-nash-loads}, the loads on all producers are constant throughout this dynamic.
Therefore, by definition of $\delta$-better-response dynamics, we have both that 
$(\bar{t}^0,P_i)_{i=0}^{\infty}$ is a best-response dynamic in \finegame, and that
$(t^i_j)_{i=0}^{\infty}$ is monotone-nondecreasing for every $j\in\prods$.

Let $i>\tilde{\imath}$ s.t.\ $k\in P_i$. As $t^{i+1}_k\ge t^{\tilde{\imath}+1}_k$, it is enough to show that $t^{i+1}_k\le t^{\tilde{\imath}+1}_k$.
As $\bar{t}^{i+1} \ge (\bar{t}^{\tilde{\imath}+1}_{-k},t^{i+1}_k)$ in every coordinate, by \cref{producer-coarse-nash-char}
and since $\bar{t}^{i+1}$ is a Nash equilibrium w.r.t.\ \coarsegame, so is $(\bar{t}^{\tilde{\imath}+1}_{-k},t^{i+1}_k)$,
and so $t^{i+1}_k$ maximizes the load on $k$ given $\bar{t}^{\tilde{\imath}+1}_{-k}$; therefore, and as $t^{\tilde{\imath}+1}_k$
is a best response to $\bar{t}^{\tilde{\imath}+1}_{-k}$ w.r.t.\ \finegame, we have that $t^{i+1}_k\le t^{\tilde{\imath}+1}_k$,
and the proof is complete.
\end{proof}

\begin{proof}[of \cref{fine-sequential-cor}]
A direct corollary of \cref{fine-response-equilibrium} (\cref{coarse-response-equilibrium}) and \cref{fine-sequential}.
\end{proof}

\begin{proof}[of \cref{fine-best-response-fast-cor}]
A direct corollary of \cref{fine-response-equilibrium} (\cref{coarse-best-response-fast}) and \cref{fine-sequential}.
\end{proof}

\subsection{Proof of Theorem~\refintitle{main-street-super-strong}}\label{main-street-proofs}

\begin{proof}[of \cref{main-street-super-strong}]
The fact that each such strategy profile is a super-strong equilibrium with load $\tilde{\ell}^g_j$ on each producer $j\in\prodsg$ of good $g\in\{1,2\}$
(and with the market split between the producers of each good as in the one-good scenario of \cref{heterogeneous})
is an immediate consequence of \cref{heterogeneous-fine}; we therefore show that
no other super-strong equilibrium exists.
We give a proof for atomless $\mu$; the proof for general $\mu$ is similar and is left to the reader.
Let $\bigl((t^g_j,\theta^g_j)\bigr)_{j\in\prodsg}^{g\in\{1,2\}}$ be a super-strong equilibrium among producers, given in polar coordinates. For
every $g\in\{1,2\}$ and $j\in\prodsg$, let $\ell^g_j$ be the load on producer $j$ of good $g$ in $\bigl((t^g_j,\theta^g_j)\bigr)_{j\in\prodsg}^{g\in\{1,2\}}$.

We begin by noting that a producer $j\in\prodsg$ of good $g\in\{1,2\}$ may still secure a load of at least $\tilde{\ell}^g_j$ by choosing the origin as its location,
and so $\ell^g_j\ge\tilde{\ell}^g_j$ for every $g\in\{1,2\}$ and $j\in\prodsg$. As for every $g\in\{1,2\}$, we have $\sum_{j\in\prodsg}\ell^g_j\le\mu(\timeset)=\sum_{j\in\prodsg} \tilde{\ell}^g_j$, we have that $\ell^g_j=\tilde{\ell}^g_j$ for every $j\in\prodsg$, i.e.\ the load on every producer is as in the super-strong equilibria described in the statement
of the \lcnamecref{main-street-super-strong}.

For every $g\in\{1,2\}$, let $\pi_g\in \prodsg!$ be a permutation s.t.\ $t^g_{\pi_g(0)}\le t^g_{\pi_g(1)}\le\cdots\le t^g_{\pi_g(n-1)}$.
For every $g\in\{1,2\}$ and $j\in\prods$, define $m^g_j\eqdef\sum_{k=0}^{j-1}\tilde{\ell}^g_{\pi_g(k)}$; by the genericity assumption on partial sums of producer-equilibrium loads,
and by positivity of equilibrium loads, we have $m^1_j\ne m^2_k$ for every $j\in\prodso$ and $k\in\prodst$ s.t.\ either $j>0$ or $k>0$.

As for every $j\in\prodsg$, the distance $t^g_{\pi_g(j)}$ is accessible by at least all consumer types consuming a positive amount from any of the producers $\pi_g(j),\pi_g(j+1),\ldots,\pi_g(n-1)$ of good $g$,
we have that $\mu\bigl([0,t^g_{\pi(j)})\bigr)\le\sum_{k=0}^{j-1}\tilde{\ell}^g_{\pi_g(k)}=m^g_j$ for every $j\in\prodsg$. Therefore, deviating to
a super-strong equilibrium as in the statement of the \lcnamecref{main-street-super-strong}, while maintaining the order of distances from the origin among producers of the same good, harms no producer.
If $t^g_{\pi_g(j)}<\Max\bigl\{t \in \timeset \mid \mu\bigl([0,t)\bigr)\le m^g_j\bigr\}$ for some $g\in\{1,2\}$ and $j\in\prodsg$, then
producer $\pi_g(j)$ of load $g$ strictly benefits from such a deviation;
therefore, $t^g_{\pi_g(j)}=\Max\bigl\{t \in \timeset \mid \mu\bigl([0,t)\bigr)\le m^g_j\bigr\}$ for every $g\in\{1,2\}$ and $j\in\prodsg$.
Therefore, as $\mu$ is atomless, the market is split between the producers of each good as in the one-good scenario of \cref{heterogeneous}.

Assume for contradiction that not all producer strategies in $\bigl((t^g_j,\theta^g_j)\bigr)_{j\in\prodsg}^{g\in\{1,2\}}$ lie on the same ray from the origin. 
Therefore, w.l.o.g.\ there exist producers $j\in\prodso$ and $k\in\prodst$ whose strategies do not lie on the same ray, s.t.\ $t^1_j \le t^2_k$ and either $j=\pi_1(n_1-1)$ or $t^2_k \le t^1_{\pi_1({\pi_1}^{-1}(j)+1)}$. (The w.l.o.g.\ assumption refers to the part played by each good.)
If $j<\pi_1(n_1-1)$, then as $\mu$ is atomless, we have $\mu\bigl([0,t^2_k)\bigr)=m^2_{\pi_2^{-1}(k)}\ne m^1_{\pi_1^{-1}(j)+1}=\mu\bigl([0,t^1_{\pi_1({\pi_1}^{-1}(j)+1)})\bigr)$ and
so $t^2_k < t^1_{\pi_1({\pi_1}^{-1}(j)+1)}$; otherwise, since by assumption $\tilde{\ell}^2_k>0$ and as $\mu$ is atomless, we have $t^2_k<1$. Either way, by market split
there exists $\varepsilon>0$ s.t.\ almost all (w.r.t.\ $\mu$) consumer types $d\in[t^2_k,t^2_k+\varepsilon)$ consume a positive amount both from producer $j$ of good $1$ and from producer $k$ of good~$2$. Let~$c$ be the circumference of the triangle whose vertices are the origin, $(t^1_j,\theta^1_j)$ and $(t^2_k,\theta^2_k)$; as the latter two do not lie on the same ray from the origin, $(t^1_j,\theta^1_j)$ is not a convex combination of the origin and $(t^2_k,\theta^2_k)$, and so by the triangle
inequality we have $c>2t^2_k$. By definition of $c$, no consumer with type $d\in[t^2_k,\frac{c}{2})$ can consume from both producer $j$ of good $1$ and producer $k$ of good $2$ without violating the consumer's QoS limit. By combining these two, we have that for $\delta\eqdef\min\bigl\{\frac{c}{2}-t^2_k,\varepsilon\}>0$,
almost all consumer types $d\in[t^2_k,t^2_k+\delta)$ consume from both these producers, while no such consumer consumes from both of them --- a contradiction, since by definition of $t^2_k$ we have that $\delta>0$ implies $\mu\bigl([t^2_k,t^2_k+\delta)\bigr)>0$.
\end{proof}

We note that the requirements in \cref{main-street-super-strong}, both for every producer-equilibrium load to be positive and for the genericity of partial sums of producer-equilibrium loads, are required.
Indeed, any producer with zero producer-equilibrium load can be moved to any ray without destabilizing the equilibrium. Furthermore, if there exist permutations $\pi_1\in \prodso!$ and
$\pi_2\in \prodst!$ and producers $j\in\prodso\setminus\{0\}$ and $k\in\prodst\setminus\{0\}$ s.t.\ $\sum_{i=0}^{j-1}\tilde{\ell}^1_{\pi_1(i)}=\sum_{i=0}^{k-1}\tilde{\ell}^2_{\pi_2(i)}$,
then moving all producers $j'$ of good $1$ s.t.\ $\pi_1(j')\ge\pi_1(j)$ and all producers $k'$ of good $2$ s.t.\ $\pi_2(k')\ge\pi_2(k)$ together to any ray does not destabilize the equilibrium either.

\end{document}